\documentclass[12pt]{elsarticle}
\usepackage[margin=2.5cm]{geometry}
\usepackage{lipsum}

\makeatletter
\def\ps@pprintTitle{%
   \let\@oddhead\@empty
   \let\@evenhead\@empty
   \def\@oddfoot{\reset@font\hfil\thepage\hfil}
   \let\@evenfoot\@oddfoot
}
\makeatother

\usepackage{amsfonts,color,morefloats,pslatex}
\usepackage{amssymb,amsthm, amsmath,latexsym}
\UseRawInputEncoding

\newtheorem{theorem}{Theorem}
\newtheorem{lemma}[theorem]{Lemma}

\newtheorem{definition}[theorem]{Definition}
\newtheorem{example}[theorem]{Example}

\newtheorem{remark}[theorem]{Remark}

\newcommand{\lcm}{{\mathrm{lcm}}}
\newcommand{\tr}{{\mathrm{Tr}}}

\newcommand{\Norm}{{\mathrm{N}}}
\newcommand{\gf}{{\mathbb{F}}}

\newcommand{\support}{{\mathrm{suppt}}}

\newcommand{\cC}{{\mathcal{C}}}

\newcommand{\bc}{{\mathbf{c}}}
\newcommand{\bg}{{\mathbf{g}}}

\newcommand{\bu}{{\mathbf{u}}}

\usepackage{enumerate}
\usepackage{booktabs}
\usepackage{arydshln}
\allowdisplaybreaks[4]
\begin{document}

\begin{frontmatter}

\title{Self-orthogonal codes from $p$-divisible codes}

\tnotetext[fn1]{
This research was supported in part by the National Natural Science Foundation of China under Grant 12271059 and in part by the Scientific Innovation Practice Project of Postgraduates of Chang¡¯an University under Grant 300103723069.
}

\author{Xiaoru Li}
\ead{lx\underline{ }lixiaoru@163.com}
\author{Ziling Heng$^{\ast}$}
\ead{zilingheng@chd.edu.cn}

\cortext[cor]{Corresponding author}
\address{School of Science, Chang'an University, Xi'an 710064, China}

\begin{abstract}
Self-orthogonal codes are an important subclass of linear codes which have nice applications in quantum codes and lattices.
It is known that a binary linear code is self-orthogonal if its every codeword has weight  divisible by four, and a ternary linear code
is self-orthogonal if and only if its every codeword has weight  divisible by three. It remains open for a long time to establish
the relationship between the self-orthogonality of a general $q$-ary linear code and the divisibility of its weights, where $q=p^m$ for a prime $p$.
In this paper, we mainly prove that any $p$-divisible  code containing the all-1 vector over the finite field $\gf_q$ is self-orthogonal
for odd prime $p$, which solves this open problem under certain conditions.
 Thanks to this result, we characterize that any projective two-weight code containing the all-1 codeword over $\gf_q$ is self-orthogonal. Furthermore, by the extending and augmentation techniques, we construct six new families of self-orthogonal divisible codes from known cyclic codes. Finally, we construct two more families of self-orthogonal divisible codes with locality 2 which have nice application in distributed storage systems.
\end{abstract}

\begin{keyword}
Linear code \sep self-orthogonal code  \sep divisible code

\MSC  94B05 \sep 94A05

\end{keyword}

\end{frontmatter}

\section{Introduction}\label{sec1}
Let $q=p^m$ for a prime $p$ and a positive integer $m$. Denote by $\gf_q$ the finite field with $q$ elements.
An $[n,k,d]$ $q$-ary linear code $\cC$ is defined as a $k$-dimensional $\gf_q$-linear subspace of $\gf_q^n$.
For fixed $n$ and $k$, it is desirable to construct a linear code with $d$ as large as possible.
However, there exists a tradeoff among the parameters $n,k$ and $d$. An $[n,k,d]$ $q$-ary linear code is said to be optimal
if there exists no $[n,k,d+1]$ code over $\gf_q$.  An $[n,k,d]$ $q$-ary linear code is said to be almost optimal
if there exists an $[n,k,d+1]$ optimal code over $\gf_q$. Let $A_i$ denote the number of all codewords with weight $i$
in an $[n,k,d]$ linear code $\cC$. Then the sequence $(1,A_1,A_2,\cdots,A_n)$ is called the weight distribution of $\cC$.
The polynomial $A(z)=1+A_1z+A_2z^2+\cdots+A_nz^n$ is referred to as the weight enumerator of $\cC$.
The weight distribution of a linear code is of particularly importance as it contains the capabilities of error detection and
correction of the code, and allows to compute the error probability of the code's error detection and correction \cite{K1}.
The weight distributions of linear codes have been extensively studied in the literature \cite{D1, D3, D4, D2, Heng2016, HL, Heng2023, HWL, LiH}.

A linear code over $\gf_q$ is said to be divisible if all its codewords have weights divisible by an integer $\Delta>1$ \cite{H}.
Then the code is said to be $\Delta$-divisible and $\Delta$ is called a divisor of the code \cite{KK}.
The most interesting case is that $\Delta$ is a power of the characteristic of $\gf_q$.
Ward introduced the divisible codes in 1981 \cite{Ward1} and gave a survey in 2001 \cite{Ward2}.
Divisible codes have many applications including Galois geometries, subspace codes, partial spreads, vector space partitions, and Griesmer
codes \cite{KK, K, Ward2, Ward3}.

Define the dual of an $[n,k]$ linear code $\cC$ over $\gf_q$ by
$$\cC^\perp=\left\{\bc\in \gf_q^n:\langle\bc,\bu\rangle=0\text{ for any }\bu\in \cC\right\},$$
where $\langle \bc,\bu\rangle$ denotes the standard inner product of two vectors $\bc$ and $\bu$.
If $\cC\subset \cC^\perp$, then  $\cC$ is called a self-orthogonal code.
In particular, if $\cC=\cC^\perp$, then $\cC$ is said to be self-dual. Self-orthogonal codes are an very important subclass
of linear codes as they have many nice applications.
The generator matrix $G$ of a self-orthogonal code is known as the row-self-orthogonal matrix such that $GG^T=\mathbf{0}$, where $\mathbf{0}$ is the zero matrix and $G^T$ is the transpose of $G$ \cite{M}. In \cite{W}, Wan used self-orthogonal codes to construct even lattices.
In \cite{C, LLX, Steane1, Steane2}, it was shown that self-orthogonal codes can be used to construct quantum codes by the well-known CSS and Steane constructions.

In coding theory, it is interesting to establish the relationship between the divisibility and self-orthogonality of an $[n,k]$ linear code $\cC$ over $\gf_q$.
For $q=2,3,4$, the following results are known from \cite{H}:
\begin{enumerate}
\item[$\bullet$] If $\cC$ is a binary self-orthogonal code, then $\cC$ is divisible by $\Delta=2$. If $\cC$ is a  binary code divisible by $\Delta=4$, then $\cC$ is self-orthogonal.
\item[$\bullet$] $\cC$ is a ternary self-orthogonal code if and only if $\cC$ is divisible by $\Delta=3$.
\item[$\bullet$] $\cC$ is a Hermitian self-orthgognal code over $\gf_4$ if and only if $\cC$ is divisible by $\Delta=2$.
\end{enumerate}
In particular, if $n$ is even and $k=n/2$, the well-known Gleason-Pierce-Ward Theorem implies that
divisible $[n,n/2]$ codes  exist only for the values of $q$ and $\Delta$ given above, except in one trivial situation, and that the codes are always
self-dual except possibly when $q=\Delta=2$ \cite{H}.
Then it is natural to consider the relationship between the divisibility and self-orthogonality of a general $q$-ary linear code.
However, it remains an open problem for a long time.

In this paper, we mainly prove that any $p$-divisible  code containing the all-1 vector over the finite field $\gf_q$ is self-orthogonal
for odd prime $p$, which solves the open problem stated above under certain conditions.
 Thanks to this result, we characterize that any projective two-weight code containing the all-1 codeword over $\gf_q$ is self-orthogonal. Furthermore, by the extending and augmentation techniques, we construct six new families of self-orthogonal divisible codes from known cyclic codes. Finally, we construct two more families of self-orthogonal divisible codes with locality 2 which have nice application in distributed storage systems.

\section{Preliminaries}\label{sec2}
In this section, we will recall some fundamental results on the extended and augmented codes of linear codes, characters and Gaussian sums over finite fields, cyclic codes, BCH codes, $t$-designs, locally recoverable codes, weakly regular bent functions, the MacWilliams equations and the Pless power moments.

\subsection{The extended and augmented codes of linear codes}
Let $\cC$ be an $[n,k,d]$ linear code over $\gf_q$. The extended code $\widehat{\cC}$ of $\cC$ is defined by
\begin{eqnarray*}
  \widehat{\cC} = \left\{(x_0,x_1,\cdots,x_{n-1},x_n) \in \gf_q^{n+1}: (x_0,x_1,\cdots,x_{n-1}) \in \cC \mbox{ with } \sum_{i=0}^{n}x_i=0\right\}.
\end{eqnarray*}
It is obvious that $\widehat{\cC}$ is also linear and it is an $[n+1, k, \widehat{d}]$ linear code, where $\widehat{d}=d$ or $d+1$.
Let $G$ and $H$ denote the generator matrix and parity check matrix of $\cC$, respectively.
Let $\widehat{G}$ and $\widehat{H}$ denote the generator matrix and parity check matrix of $\widehat{\cC}$, respectively.
Then $\widehat{G}$ can be obtained from $G$ by adding a column  to $G$ such that the sum of the elements of each row of $\widehat{G}$ is $0$. Furthermore, $\widehat{H}$ is given by
\begin{eqnarray*}
\left.\widehat{H}=\left[\begin{array}{ccc|c}1&\cdots&1&1\\\hline&&&0\\&H&&\vdots\\&&&0\end{array}\right.\right].
\end{eqnarray*}

Let $\cC$ be an $[n,k,d]$ linear code with generator matrix $G$. Assume that all-$1$ vector $\mathbf{1}=(1,1,\cdots,1)$ is not a codeword in $\cC$. Then the augmented code $\overline{\cC}$ of $\cC$ is the linear code over $\gf_q$ with generator matrix
\begin{eqnarray*}
\left.\left[\begin{array}{c}G\\ \mathbf{1}\end{array}\right.\right].
\end{eqnarray*}
 Obviously that the length of $\overline{\cC}$ is $n$ and the dimension of $\overline{\cC}$ is $k+1$. Generally, we may require information of the complete weight distribution of the original code $\cC$ to determine the minimum distance of the augmented
code $\overline{\cC}$ \cite{D}.
\subsection{Characters over finite fields}
Let $q=p^e$ with $p$ a prime. Let $\zeta_p$ denote the primitive $p$-th root of complex unity.
 Define an additive character of $\gf_q$ by the homomorphism $\phi$ from the additive group $\gf_q$ to the complex multiplicative group $\mathbb{C}^*$ such that
$
\phi(x+y)=\phi(x)\phi(y)
$ for all $x, y \in \gf_q$.
For any $a \in \gf_q$,
an additive character of $\gf_q$ can be defined by the function $\phi_a(x)=\zeta_p^{\tr_{q/p}(ax)},\ x\in \gf_q$, where $\tr_{q/p}(x)=\sum_{i=0}^{e-1}x^{p^i}$ is the trace function from $\gf_q$ to $\gf_p$. Furthermore, the set $\widehat{\gf_q}:=\{\phi_a:a\in \gf_q\}$ gives all $q$ different additive characters of $\gf_q$. Obviously, $\phi_a(x)=\phi_1(ax)$. In particular, $\phi_0$ and $\phi_1$ are called the trivial additive character and the canonical additive character of $\gf_q$, respectively. The orthogonal relation of additive characters (see \cite{L}) is given by
\begin{eqnarray*}
\sum_{x\in \gf_q}\phi_a(x)=\begin{cases}
q    &\text{if $a=0,$ }\\
0     &\text{otherwise .}
\end{cases}
\end{eqnarray*}

Let $\gf_q^*=\langle\beta\rangle$. For each $0\leq j\leq q-2$, a multiplicative character of $\gf_q$ is defined as the homomorphism $\psi$ from the multiplicative group $\gf_q^*$ to the complex multiplicative group $\mathbb{C}^*$ such that
$
\psi(xy)=\psi(x)\psi(y)
$ for all $x, y \in \gf_q^*$.
The function
$\psi_j(\beta^k)=\zeta_{q-1}^{jk}$ for $k=0,1,\cdots,q-2$
 gives a multiplicative character of $\gf_q$, where $0\leq j \leq q-2$. The set $\widehat{\gf_q}^*:=\{\psi_j:0\leq \psi \leq q-2\}$ consists of all the multiplicative characters of $\gf_q$.
 In particular, $\psi_0$ is referred to as the trivial multiplicative character and $\eta:=\psi_{\frac{q-1}{2}}$ is said to be the quadratic multiplicative character of $\gf_q$ if $q$ is odd. The orthogonal relation of multiplicative characters (see \cite{L}) is given by
\begin{eqnarray*}
\sum_{x\in\gf_q^*}\psi_j(x)=\begin{cases}
q-1    &\text{if $j=0,$ }\\
0    &\text{if $j \neq 0$.}
\end{cases}
\end{eqnarray*}

\subsection{Gaussian sums over finite fields}
For an additive character $\phi$ and a multiplicative character $\psi$ of $\gf_q$, the \emph{Gaussian sum} $G(\psi, \phi)$ over $\gf_q$ is defined by
$$G(\psi, \phi)=\sum_{x\in \gf_q^*}\psi(x)\phi(x).$$
Particularly, $G(\eta,\phi)$ is referred to as the \emph{quadratic} \emph{Gaussian sum} over $\gf_q$ for nontrivial $\phi$.

The explicit values of quadratic \emph{Gaussian sum} are given as follows.
\begin{lemma}[\cite{L}, Theorem 5.15]\label{Guassum}
Let $q=p^e$ with $p$ an odd prime. Then
\begin{eqnarray*}
G(\eta,\phi_1)=(-1)^{e-1}(\sqrt{-1})^{(\frac{p-1}{2})^2e}\sqrt{q}
=\left\{
\begin{array}{lll}
(-1)^{e-1}\sqrt{q}    &   \mbox{ for }p\equiv 1\pmod{4},\\
(-1)^{e-1}(\sqrt{-1})^{e}\sqrt{q}    &   \mbox{ for }p\equiv 3\pmod{4}.
\end{array}
\right.
\end{eqnarray*}
\end{lemma}

The following lemma gives the explicit value of a family of weil sums.
\begin{lemma}[\cite{L}, Theorem 5.33]\label{weilsum}
Let $\phi$ be a nontrivial additive character of $\gf_q$, where $q$ is power of an odd prime. Let $f(x)=a_2x^2+a_1x+a_0 \in \gf_q[x]$, where $a_2 \neq 0$. Then
\begin{eqnarray*}
\sum_{c \in \gf_q}\phi\left(f(c)\right) = \phi\left(a_0-a_1^2(4a_2)^{-1}\right)\eta(a_2)G(\eta, \phi).
\end{eqnarray*}
\end{lemma}
\subsection{Cyclic codes and BCH codes}
Let $\cC$ be an $[n,k]$ linear code over $\gf_q$.
 For each codeword $\bc=(c_0,c_1,\cdots,c_{n-1}) \in \cC$, if its cyclic shift $\bc'=(c_{n-1},c_0,c_1,\cdots,c_{n-2})$ is also in $\cC$, then $\cC$ is called a cyclic code. Since there is a bijective between the vector $\bc=(c_0,c_1,\cdots,c_{n-1}) \in \gf_{q^n}$ and the polynomial $c(x)=c_0+c_1x+\cdots+c_{n-1}x^{n-1} \in \gf_q[x]/(x^n-1)$,  a cyclic code $\cC$ of length $n$ over $\gf_q$ corresponds to an ideal of $\gf_q[x]/(x^n-1)$. Since every ideal of $\gf_q[x]/(x^n-1)$ is principal, there exists a monic polynomial $g(x)$ of the least degree such that
$\cC=\langle g(x)\rangle$ and $g(x) \mid (x^n-1)$. Then $g(x)$ is called the generator polynomial and $h(x)=(x^n-1)/g(x)$ said to be the check polynomial of the cyclic code $\cC$.

Let $n=q^m-1$. Denote by $\gf_{q^m}^*=\langle\alpha\rangle$. For any integer $i$ with $0 \leq i \leq n-1$, let $m_i(x)$ be the minimal polynomial of $\alpha^i$ over $\gf_q$. Define
\begin{eqnarray*}
  g_{(q,m,\delta)}(x) &=& \lcm\left(m_1(x), m_2(x), \cdots, m_{\delta-1}(x)\right)
\end{eqnarray*}
as the least common multiple of these minimal polynomials, where $2 \leq \delta \leq n$. Let $\cC_{(q,m,\delta)}$ be the cyclic code of length $n=q^m-1$ over $\gf_q$ with generator polynomial $g_{(q,m,\delta)}(x)$. Then $\cC_{(q,m,\delta)}$ is referred to as a primitive BCH code with designed distance $\delta$.
Define
\begin{eqnarray*}
  \tilde{g}_{(q,m,\delta)}(x) &=& (x-1)g_{(q,m,\delta)}(x).
\end{eqnarray*}
Let $\tilde{\cC}_{(q,m,\delta)}$ be the cyclic code of length $n=q^m-1$ over $\gf_q$ with generator polynomial $\tilde{g}_{(q,m,\delta)}(x)$. Obviously, $\tilde{\cC}_{(q,m,\delta)}$ is a subcode of $\cC_{(q,m,\delta)}$. For more details on BCH codes, the reader is referred to \cite{D}.

\subsection{Combinatorial $t$-designs from linear codes}
In this subsection, we will first give the definition of $t$-designs and then present the well-known coding-theoretic construction of $t$-designs.

\begin{definition}
 Let $n, \kappa$ and $t$ be positive integers with $1 \leq t \leq \kappa \leq n$. Denote by $\mathcal{P}$ a set of $n$ elements and $\mathcal{B}$ a set of elements that are $\kappa$-subsets of $\mathcal{P}$. If each $t$-subset of $\mathcal{P}$ is contained in precisely $\lambda$ elements of $\mathcal{B}$, then the pair $\mathbb{D}:=(\mathcal{P},\mathcal{B})$ is referred to as a $t$-$(n,\kappa,\lambda)$ design (or for short, $t$-design). The elements of $\mathcal{P}$ are called points and those of $\mathcal{B}$ are called blocks. If $\mathcal{B}$ does not contain repeated blocks, then the $t$-design is said to be simple.
\end{definition}

The following is the well-known coding-theoretic construction of $t$-designs. Let $\mathcal{C}$ be a code of length $n$. For each codeword $\bc \in \mathcal{C}$, the support of $\textbf{c}=\{c_1,c_2,\cdots,c_n \}$ is defined by $$\support(\textbf{c})=\{1 \leq i \leq n : c_i \neq 0 \}.$$
Let $\mathcal{B}_\kappa$ represent the set of supports of all codewords with Hamming weight $\kappa$ in $\mathcal{C}$ and let $\mathcal{P}=\{ 1,2,\cdots,n \}$ represent the set of coordinate positions of the codewords in $\mathcal{C}$.
The pair $(\mathcal{P},\mathcal{B}_\kappa)$ may be a $t$-$(n,\kappa,\lambda)$ design for some positive integer $\lambda$, which is called a support design of code $\mathcal{C}$. In other words, we say that the code $\mathcal{C}$ holds a $t$-$(n,\kappa,\lambda)$ design. If the pair $(\mathcal{P},\mathcal{B}_\kappa)$ is a simple $t$-$(n,\kappa,\lambda)$ design, the following relation holds \cite{H}:
\begin{eqnarray}\label{eqn-t}
|\mathcal{B}_\kappa|=\frac{1}{q-1}A_\kappa, \binom{n}{t}\lambda=\binom{\kappa}{t} \frac{1}{q-1}A_\kappa.
\end{eqnarray}

The following lemma is the well-known Assmus-Mattson Theorem.
\begin{lemma}\cite{H}\label{lem-AMth}
 Let $\mathcal{C}$ be an $[n,k,d]$ code over $\mathbb{F}_{q}$, and let $d^\perp$ denote the minimum distance of $\mathcal{C}^\perp$. Let $w$ be the largest integer satisfying $w \leq n$ and
$
w-\left\lfloor \frac{w+q-1}{q-2} \right\rfloor < d.
$
Define $w^\perp$ analogously with $d^\perp$. Let $(A_i)_{i=0}^{n}$ and $(A_i^\perp)_{i=0}^{n}$ be the weight distributions of $\mathcal{C}$ and $\mathcal{C}^\perp$, respectively. Let $t$ be a positive integer with $t<d$ such that there are at most $d^\perp-t$ weights of $\mathcal{C}$ in the sequence $(A_0,A_1,\cdots,A_{n-t})$. Then
\begin{enumerate}
\item $(\mathcal{P},\mathcal{B}_\kappa)$ is a simple $t$-design provided that $A_\kappa \neq 0$ and $d \leq \kappa \leq w$;

\item $(\mathcal{P},\mathcal{B}_\kappa^\perp)$ is a simple $t$-design provided that $A_\kappa^\perp \neq 0$ and $d^\perp \leq \kappa \leq w^\perp$, where $\mathcal{B}_\kappa^\perp$ denotes the set of supports of all codewords of weight $\kappa$ in $\mathcal{C}^\perp$.
\end{enumerate}
\end{lemma}

The Assmus-Mattson Theorem is a powerful tool for constructing $t$-designs from linear codes \cite{D}.

\subsection{Locally recoverable codes}
In order to recover the data in distributed and cloud storage systems,
locally recoverable codes (LRCs for short) were proposed by Gopalan, Huang, Simitci and Yikhanin \cite{GH}.
For a positive integer $n$, we denote by $[n]=\{ 0,1,\cdots,n-1 \}$.  Let $\mathcal{C}$ be an $[n,k,d]$ linear code over $\mathbb{F}_q$. We index the coordinates of the codewords in $\mathcal{C}$ with the elements in $[n]$.
For each $i \in [n]$, if there exist a subset $R_{i} \subseteq [n] \backslash {i}$ of size $r$ and a function $f_{i}(x_1,x_2,\cdots,x_r)$ on $\mathbb{F}_q^{r}$ meeting $c_i=f_{i}(\mathbf{c}_{R_i})$ for any $\mathbf{c}=(c_0,\cdots,c_{n-1}) \in \mathcal{C}$, then we call $\mathcal{C}$ an $(n,k,d,q;r)$-LRC, where $\mathbf{c}_{R_i}$ is the projection of $\mathbf{c}$ at $R_{i}$. The set $R_{i}$ is called  the repair set of $c_i$ and $r$ is referred to as the locality of $\mathcal{C}$. Locally recoverable codes also have been implemented in practice by Microsoft and Facebook \cite{HC, SM}.

We also present the conventional definition of linear locally recoverable codes as follows.

\begin{definition}\label{Def-locality}\cite{HY}
Let $\cC$ be a linear code over $\gf_q$ with a generator matrix $G=[\bg_1, \bg_2, \cdots, \bg_n]$. If $\bg_i$ is a linear combination of $l(\leq r)$ other columns of $G$, then
$r$ is called the locality of the $i$-th symbol of each codeword of $\cC$. Besides, $\cC$ is called a locally recoverable code with locality $r$ if all the symbols of codeword of $\cC$ have locality $r$.
\end{definition}

Constructing LRCs with small locality has been an interesting research topic because LRCs have nice application in large-scale distributed storage systems.
A lot of progress has been made on the research of locally recoverable codes. Very recently, in \cite{TP}, the authors established general theory on the minimum locality of linear codes and studied the minimum locality of some families of linear codes.

By the following lemma, we can derive the locality of a linear code if its dual holds a $1$-design.

\begin{lemma}\label{lem-locality}\cite{TP}
Let $\cC$ be a nontrivial linear code of length $n$ and $d^{\perp}$ be the minimum distance of $\cC^{\perp}$. If $(\mathcal{P}(\cC^{\perp}), \mathcal{B}_{d^{\perp}}(\cC^{\perp}))$ is a $1-(n, d^{\perp}, \lambda^{\perp})$ design with $\lambda^{\perp} \geq 1$, then $\cC$ has locality $d^{\perp}-1$.
\end{lemma}

\subsection{Weakly regular bent functions}
Let $f(x)$ denote a function from $\gf_{p^e}$ to $\gf_p$. Define the Walsh transform of $f(x)$ as follows:
\begin{eqnarray*}
\text{W}_f(\beta):=\sum_{x \in \gf_{p^e}}\zeta_p^{f(x)-\tr_{p^e/p}(\beta x)},\ \beta \in \gf_{p^e}.
\end{eqnarray*}
The function $f(x)$ is said to be a $p$-ary bent function if $|\text{W}_f(\beta)|=p^{\frac{e}{2}}$ for any $\beta \in \gf_{p^e}$. For a bent function $f(x)$, if there exists some $p$-ary function $f^*(x)$ such that $\text{W}_f(\beta)=p^{\frac{e}{2}}\zeta_p^{f^*(\beta)}$ for any $\beta \in \gf_{p^e}$, then $f(x)$ is called a regular bent function. If there exists some $p$-ary function $f^*(x)$ and a complex $u$ with unit magnitude such that $\text{W}_f(\beta)=up^{\frac{e}{2}}\zeta_p^{f^*(\beta)}$ for any $\beta \in \gf_{p^e}$, then $f(x)$ is  said to be a weakly regular bent function, where $f^*(x)$ is called the dual of $f(x)$. It turns out in \cite{Helleseth1} and \cite{Helleseth2} that if $f(x)$ is a weakly regular bent function, then
\begin{eqnarray}\label{eq-Wf}
\text{W}_f(\beta)=\varepsilon \sqrt{p^*}^e\zeta_p^{f^*(\beta)},
\end{eqnarray}
where $\varepsilon=\pm1$ is called the sign of the Walsh transform of $f(x)$ and $p^*=(-1)^{\frac{p-1}{2}}p$. Obviously, the dual of a weakly regular bent function $f(x)$ is also a weakly regular bent function and $(f^*)^*(x)=f(-x)$. Besides, the sign of the Walsh transform of $f^*(x)$ is $\eta_0^e(-1)\varepsilon$, where $\eta_0$ is the quadratic multiplicative character of $\gf_p$.

\begin{definition}\label{def-weakly}
Denote by $\mathcal{RF}$ the set of all $p$-ary weakly regular bent functions $f(x)$ such that $f(0)=0$ and $f(ax)=a^hf(x)$ for any $x \in \gf_{p^e}$ and $a \in \gf_p^*$, where $h$ is a positive even integer with $\gcd(h-1, p-1)=1$.
\end{definition}

Note that almost all known weakly regular bent functions are contained in $\mathcal{RF}$.
Some known weakly regular bent functions are listed in Table \ref{tab1}, where $p$ is an odd prime.

\begin{table}[h!]
\scriptsize
\begin{center}
\caption{Known weakly regular bent functions.}\label{tab1}
\begin{tabular}{@{}lllll@{}}
\toprule
Bent function & $e$ & $p$ & $\varepsilon$ &Reference  \\
\midrule
$f(x)=\sum_{i=0}^{\lfloor\frac{e}{2}\rfloor}\tr_{p^e/p}(c_ix^{p^i+1})$ (nondegenerate) & arbitrary & arbitrary & \cite[Proposition 1]{Helleseth1} & \cite{Helleseth1}\\
$\substack{\mbox{$f(x)=\sum_{i=1}^{p^k-1}\tr_{p^e/p}(c_ix^{i(p^k-1)})+\tr_{p^l/p}(\delta x^{\frac{p^e-1}{t}})$, where}\\ \mbox{$t \mid (p^k+1), c_i \in \gf_{p^e},\delta \in \gf_{p^l}$ and $l$ is the least positive}\\
 \mbox{integer satisfying $l \mid m$ and $t \mid (p^l-1)$}}$& $e=2k$ & arbitrary & $(-1)^{\frac{(p-1)e}{4}}$ \cite{T} & \cite{Helleseth1, LH}\\
$f(x)=\tr_{3^e/3}(cx^{\frac{3^e-1}{4}+3^k+1}), \gf_{3^e}^*=\langle\alpha\rangle, c=\alpha^{\frac{3^k+1}{4}}$& $e=2k$ with $k$ odd& $p=3$ & $(-1)^{\frac{e}{2}+1}$ \cite{T} & \cite{Helleseth1}\\
$f(x)=\tr_{p^e/p}(x^{p^{3k}+p^{2k}-p^k+1}+x^2)$& $e=4k$ & arbitrary & $-1$ \cite{T} & \cite{Helleseth2}\\
$f(x)=\tr_{3^e/3}(cx^{\frac{3^i+1}{2}}), c \in \gf_{3^e}^*$ & $i$ odd and $gcd(i,e)=1$ & $p=3$ & $(-1)^{m-1}\eta(c)$ \cite{T} & \cite{CRS}\\
\bottomrule
\end{tabular}
\end{center}
\end{table}

\subsection{The MacWilliams equations and the Pless power moments}
The MacWilliams equations and the Pless power moments are useful in calculating the minimum distance of the dual of a linear code.
\begin{lemma}\cite[Pless power moments, Page 260]{H} \label{lem-Pless}
For an $[n,k,d]$ linear code $\cC$ over $\gf_q$, let $(1, A_1, A_2, \cdots, A_n)$ and $(1, A_1^{\perp}, A_2^{\perp}, \cdots, A_n^{\perp})$ represent the weight distributions of $\cC$ and $\cC^{\perp}$, respectively. The first four Pless power moments are listed in the following:
\begin{eqnarray*}
\sum_{i=0}^{n}A_i&=&q^k,\\
\sum_{i=0}^{n}iA_i&=&q^{k-1}(qn-n-A_1^{\perp}), \\
\sum_{i=0}^{n}i^2A_{i}&=&q^{k-2}\left((q-1)n(qn-n+1)-(2qn-q-2n+2)A_1^{\perp}+2A_2^{\perp}\right),\\
\sum_{i=0}^{n}i^3A_{i}&=&q^{k-3}\left[(q-1)n(q^2n^2-2qn^2+3qn-q+n^2-3n+2)\right.\nonumber\\&&\left.-(3q^2n^2-3q^2n-6qn^2+12qn+q^2-6q+3n^2-9n+6)A_1^{\perp}\right.\nonumber\\&&\left.
+6(qn-q-n+2)A_2^{\perp}-6A_3^\perp\right].
\end{eqnarray*}
\end{lemma}

\begin{lemma}\cite[MacWilliams equations]{H}\label{lem-Mac}
  Let $\cC$ be an $[n,k,d]$ code over $\gf_q$. Let $A(z)=\sum_{i=0}^{n}A_iz^i$ and $A^{\perp}(z)$ be the weight enumerator of $\cC$ and $\cC^{\perp}$, respectively. Then
  \begin{eqnarray*}
    A^{\perp}(z) &=& q^{-k}(1+(q-1)z)^nA\left(\frac{1-z}{1+(q-1)z}\right).
  \end{eqnarray*}
\end{lemma}

\section{The main results}\label{sec3}
Let $q$ be an odd prime power. In this section, we will give a sufficient condition for a $q$-ary linear code to be self-orthogonal by the divisibility of the code. Then we will apply this result to study the self-orthogonality of projective two-weight codes, Griesmer codes and generalized ${\rm Reed}$-${\rm Muller}$ codes.
\subsection{A sufficient condition for a $q$-ary linear code to be self-orthgognal}
\begin{lemma}\cite{W}\label{lem-self-orthogonal}
  Let $q$ be a power of an odd prime and $\cC$ be a $q$-ary linear code. Then $\cC$ is self-orthogonal if and only if $\bc \cdot \bc=0$ for all $\bc \in \cC$.
\end{lemma}

\begin{theorem}\label{th-selforthogonal}
  Let $q=p^e$, where $p$ is an odd prime. Let $\cC$ be an $[n,k,d]$ linear code over $\gf_q$ with $\mathbf{1} \in \cC$, where $\mathbf{1}$ is the all-$1$ vector of length $n$. If $\cC$ is $p$-divisible, then $\cC$ is self-orthogonal.
\end{theorem}

\begin{proof}
Let $a$ be any fixed nonzero element in $\gf_q$.
  Since $\mathbf{1} \in \cC$, $a \mathbf{1} \in \cC$. Then we have $\bc +a \mathbf{1} =(c_0+a, c_1+a, \cdots, c_{n-1}+a) \in \cC$ for any $\bc = (c_0,c_1,\cdots,c_{n-1}) \in \cC$. Since $\cC$ is a linear code, we deduce that $\bc$ runs over $\cC$ if and only if $\bc+a\mathbf{1}$ runs over $\cC$ for fixed $a \in \gf_q^*$. By Lemma \ref{lem-self-orthogonal}, $\cC$ is self-orthogonal if and only if $(\bc+a\mathbf{1}) \cdot (\bc+a\mathbf{1}) = 0$ for all $\bc \in \cC$. Denote by $$N_s=\left| \left\{0 \leq i \leq n-1 : c_i = s\right\} \right|\mbox{ for } s \in \gf_q^*.$$ Then
  \begin{eqnarray*}
   (\bc+a\mathbf{1}) \cdot (\bc+a\mathbf{1}) &=& \left(\begin{array}{c}
                                                         c_0+a,\  c_1+a,\  \cdots,\ c_{n-1}+a
                                                       \end{array}\right)
                                                       \left(\begin{array}{c}
                                                               c_0+a \\
                                                               c_1+a \\
                                                               \vdots \\
                                                               c_{n-1}+a
                                                             \end{array}\right) \\
    &=& \sum_{i=0}^{n-1}(c_i+a)^2\\
    &=& \sum_{i=0}^{n-1}c_i^2+2a\sum_{i=0}^{n-1}c_i+na^2\\
    &=& \sum_{s \in \gf_q^*} N_s s^2 +2a\sum_{s \in \gf_q^*} N_s s+na^2.
  \end{eqnarray*}
Let $\text{wt}(\bc)$ denote the weight of the codeword $\bc \in \cC$. For any $s \in \gf_q^*$,
\begin{eqnarray*}
  \text{wt}(\bc - s\mathbf{1}) = n-|\{0 \leq i \leq n-1: c_i-s =0\} |= n-N_{s}.
\end{eqnarray*}
Then $N_{s}=n-\text{wt}(\bc - s\mathbf{1})$. Hence,
\begin{eqnarray*}
 (\bc+a\mathbf{1}) \cdot (\bc+a\mathbf{1}) = \sum_{s \in \gf_q^*}(n-\text{wt}(\bc - s\mathbf{1}))s^2+2a\sum_{s \in \gf_q^*}(n-\text{wt}(\bc - s\mathbf{1}))s+na^2.
\end{eqnarray*}
Since $p \mid \text{wt}(\bc)$ for any codeword $\bc \in \cC$ and $\mathbf{1} \in \cC$, we derive that $p \mid \text{wt}(\bc - s\mathbf{1})$ and $p \mid n$. Hence, we have $(\bc+a\mathbf{1}) \cdot (\bc+a\mathbf{1})=0$ for all $\bc \in \cC$. Then $\cC$ is self-orthogonal.
\end{proof}

\begin{remark}
  We remark that the condition in Theorem \ref{th-selforthogonal} such that $\cC$ is self-orthogonal is not necessary. In fact, there exist self-orthogonal codes which are not $p$-divisible. In the following, we give an example.

    Let $q=3^3$, $k=3$, $\gf_q^*=\langle\beta\rangle$, $\textbf{a}=(0,\beta^0,\beta^1,\ldots,\beta^{q-2})$ and $\textbf{v}$ be the all-$1$ vector of length $q$. Then the generalized Reed-Solomon code $GRS_k(\textbf{a},\textbf{v})$ has parameters $[27,3,25]$ and weight enumerator
  \begin{eqnarray*}
    A(z) &=& 1+9126z^{25}+1404z^{26}+9152z^{27}.
  \end{eqnarray*}
  By Magma, $\mathbf{1} \in GRS_k(\textbf{a},\textbf{v})$ and $GRS_k(\textbf{a},\textbf{v})$ is self-orthogonal. However, $GRS_k(\textbf{a},\textbf{v})$ is not $p$-divisible.
\end{remark}

In the following, we give a new approach to prove the self-orthogonality of generalized Reed-Muller codes (GRM codes for short \cite{RMcodes}) by Theorem \ref{th-selforthogonal}.
Let $\gf_q[x_1,x_2, \cdots, x_m]$ be the polynomial ring of $m$ variables. Denote by $\text{Eval}_\textbf{z}(f):=f(z_1,z_2,\cdots,z_m)$ the evaluation of $f$ at the vector $\textbf{z}$. Let $\text{Eval}(f):=(\text{Eval}_\textbf{z}(f): \textbf{z} \in \gf_q^m)$ denote the evaluation vector of $f$ whose coordinates are the evaluations of $f$ at all $q^m$ vectors in $\gf_q^m$.
\begin{definition}
  Let $n:=q^m$ and $\rho \leq m(q-1)$. The $\rho$-th order $q$-ary ${\rm Reed}$-${\rm Muller}$ code ${\rm RM}_q(\rho,m)$ is defined by
  \begin{eqnarray*}
   {\rm RM}_q(\rho,m) := \{{\rm Eval}(f): f \in \gf_q[x_1,x_2,\cdots,x_m], {\rm deg}(f) \leq \rho\}.
  \end{eqnarray*}
\end{definition}

\begin{lemma}\cite{AA, EA, KL, PR}
  Let notations be the same as above. Suppose $0 \leq \rho <m(q-1)$. Then ${\rm RM}_q(\rho,m)$ has parameters
    \begin{eqnarray*}
    \left[q^m, \sum_{j=0}^m(-1)^j\binom{m}{j}\binom{m+r-jq}{r-jq}, (b+1)q^a\right],
    \end{eqnarray*}
    where $m(q-1)-\rho=a(q-1)+b$ with $a,b \geq 0$ and $b < q-1$.
\end{lemma}

Several classical codes have non-trivial divisibility and the most prominent one is the generalized ${\rm Reed}$-${\rm Muller}$ code. The Theorem of AX \cite{AX} described the divisibility of generalized ${\rm Reed}$-${\rm Muller}$ codes.

\begin{theorem}\label{th-AX}\cite{AX}
  Let $q$ be a power of prime $p$. Then the $\rho$-th order generalized ${\rm Reed}$-${\rm Muller}$ code ${\rm RM}_q(\rho,m)$ over $\gf_q$ is divisible by $q^{\lceil m/\rho \rceil-1}$. Moreover, this divisor is the highest power of $p$ that divides the code.
\end{theorem}

By Theorem \ref{th-selforthogonal}, it is easy to derive the self-orthogonality of the generalized ${\rm Reed}$-${\rm Muller}$ code for odd prime power $q$.
\begin{theorem}
  Let $q$ be an odd prime power and $\rho < m$. Then all the $\rho$-th order generalized ${\rm Reed}$-${\rm Muller}$ code ${\rm RM}_q(\rho,m)$ over $\gf_q$ is self-orthogonal.
\end{theorem}

\begin{proof}
By the definition of the $\rho$-th order generalized ${\rm Reed}$-${\rm Muller}$ code ${\rm RM}_q(\rho,m)$, we have $\mathbf{1} \in {\rm RM}_q(\rho,m)$.
Then the desired conclusion follows from Theorems  \ref{th-AX} and \ref{th-selforthogonal}.
\end{proof}

\subsection{The self-orthogonality of projective two-weight codes and Griesmer codes}

A linear $[n,k,d]$ code $\cC$ is called projective if no two columns of a generator matrix $G$ are linearly dependent. Then the minimum distance of $\cC^{\perp}$ is $d^{\perp} \geq 3$ and the columns of $G$ are pairwise different points in a projective $(k-1)$-dimensional space. In \cite{BI}, the authors studied the divisibility of  projective two-weight codes.
\begin{lemma}\label{lem-projective}\cite{BI}
  Let $q=p^e$, where $p$ is a prime and $e$ is a positive integer. Then all the projective two-weight codes over $\gf_q$ of dimension $k \geq 3$ are $p$-divisible expect the MacDonald codes with parameters $[(q^k-q)/(q-1),k,q^{k-1}-1]$ and weights $q^{k-1}-1$ and $q^{k-1}$.
\end{lemma}

By Theorem \ref{th-selforthogonal}, we derive the following theorem.

\begin{theorem}\label{col2}
  Let $q$ be an odd prime power and $k \geq 3$. Then any $[n,k]$ projective two-weight code $\cC$ over $\gf_q$ containing the all-$1$ vector is self-orthogonal.
\end{theorem}

\begin{proof}
  By Lemma \ref{lem-projective}, we know that all the projective two-weight codes of dimension $k \geq 3$ are $p$-divisible expect the MacDonald codes with parameters $[(q^k-q)/(q-1),k,q^{k-1}-1]$ and weights $q^{k-1}-1$ and $q^{k-1}$. Besides, we deduce that the MacDonald code does not contain the all-$1$ vector of length $(q^k-q)/(q-1)$ as it has no codeword with weight $(q^k-q)/(q-1)$. They by Theorem \ref{th-selforthogonal}, the desired conclusion follows.
\end{proof}

The following example confirms the conclusion in Theorem \ref{col2}.
\begin{example}\label{exa-CDbar}
Let $q$ be an odd prime power and $m > 2$ be an integer. Denote by the defining set $D=\{x \in \gf_{q^m}: \tr_{q^m/q}(x)=0\}$. Then the linear code $\overline{\cC_D}=\{(\tr_{q^m/q}(bx)+c)_{x \in D}: b \in \gf_{q^m}, c \in \gf_q\}$ has the following properties which are easy to prove.
\begin{itemize}
  \item $\overline{\cC_D}$ has parameters $[q^{m-1}, m, q^{m-1}-q^{m-2}]$.
  \item $\overline{\cC_D}$ has weight enumerator $1+(q-1)z^{q^{m-1}}+q(q^{m-1}-1)z^{q^{m-1}-q^{m-2}}$.
  \item $\overline{\cC_D}^{\perp}$ has parameters $[q^{m-1}, q^{m-1}-m, 3]$ and then $\overline{\cC_D}$ is a projective two-weight linear code.
\end{itemize}
Then by Theorem \ref{col2}, $\overline{\cC_D}$ is a self-orthogonal code.
\end{example}

\begin{remark}
  We remark that a projective two-weight code $\cC$ also may be self-orthogonal code if $\mathbf{1} \notin \cC$. For example, let $$\cC_D=\left\{(\tr_{q^m/q}(bd_1), \tr_{q^m/q}(bd_2), \cdots, \tr_{q^m/q}(bd_n)): b \in \gf_{q^m}\right\}$$ be a linear code with defining set $$D=\{x \in \gf_{q^m}^*: \tr_{q^s/q}(\Norm_{q^m/q^s}(x))+c=0\}, c \in \gf_q.$$ This linear code is a projective two-weight code and its weight distribution was given in \cite{HL}. We find that $\mathbf{1} \notin \cC_D$, while $\cC_D$ is a self-orthogonal code.
\end{remark}

The Griesmer bound for an $[n,k,d]$ linear code $\cC$ over $\gf_q$ is given by
\begin{eqnarray*}
 n \geq \sum_{i=0}^{k-1}\left\lceil\frac{d}{q^i}\right\rceil,
\end{eqnarray*}
where $\lceil\cdot\rceil$ denotes the ceiling function. A linear code achieving this bound is called a Griesmer code. In \cite{DS, Ward3}, the authors  studied the  divisibility of the Griesmer codes.
\begin{lemma}\cite{DS, Ward3}\label{lem-Griesmer}
  Let $p$ be a prime and $\cC$ be an $[n,k,d]$ Griesmer code over $\gf_p$. If $p^e$ divides the minimum distance of $\cC$, then $\cC$ is $p^e$-divisible.
\end{lemma}

Based on Theorem \ref{th-selforthogonal}, we derive the following theorem.

\begin{theorem}\label{th-G}
  Let $p$ be an odd prime. Then any $[n,k,d]$ Griesmer code $\cC$ over $\gf_p$ such that $p \mid d$ and $\mathbf{1} \in \cC$ is self-orthogonal.
\end{theorem}

\begin{proof}
  By Lemma \ref{lem-Griesmer}, we deduce that $\cC$ is $p$-divisible if $p\mid d$. Then by Theorem \ref{th-selforthogonal}, $\cC$ is self-orthogonal.
\end{proof}

\begin{remark}
  We remark that the projective code $\overline{\cC_D}$ in Example \ref{exa-CDbar} is also a Griesmer code. Particularly, when $q=p$, we can derive that $\overline{\cC_D}$ is self-orthogonal by Theorem \ref{th-G}.
\end{remark}

\section{Self-orthogonal codes from cyclic codes}\label{sec4}
In this section, we will apply Theorem \ref{th-selforthogonal} to construct several families of $p$-divisibility self-orthogonal codes from cyclic codes.
For a codeword $\bc$ in a linear code $\cC$, let $\text{wt}(\bc)$ denote the Hamming weight of $\bc$.
\subsection{The first family of self-orthogonal codes from Zhou-Ding cyclic codes}

In this subsection, let $m,k$ be two positive integers and $e=\text{gcd}(m,k)$. Let $q=p^e$, where $p$ is an odd prime. Let $s=\frac{m}{e} \geq 3$ and $\frac{k}{e}$ be odd. Denote by $\gf_{q^s}^*=\langle\alpha\rangle$. Let $h_1(x)$ and $h_2(x)$ denote the minimal polynomials of $(-\alpha)^{-1}$ and $\alpha^{-\frac{p^k+1}{2}}$ over $\gf_p$, respectively. Define a family of cyclic codes as
\begin{eqnarray}\label{eq-C}
  \cC_1 &=& \left\{ \bc_{(a,b)}=\left(\tr_{q^s/p}\left(a(-\alpha)^t + b\alpha^{\frac{(p^k+1)t}{2}}\right)\right)_{t=0}^{q^s-2}: a,b \in \gf_{q^s} \right\}
\end{eqnarray}
with parity-check polynomial $h_1(x)h_2(x)$. The weight distribution of $\cC_1$ in Equation (\ref{eq-C}) was given in \cite{Zhou2014} and the complete weight distribution of $\cC_1$ in Equation (\ref{eq-C}) was studied in \cite{Heng2016}.

Let $\overline{\widehat{\cC_1}}$ be the augmented code of extended code of cyclic code $\cC_1$. By the definitions of the extended and augmented codes of linear codes, we have
\begin{eqnarray}\label{eq-Cbar}
  \overline{\widehat{\cC_1}} &=& \Bigg\{\bc_{(a,b,c)}=\bigg(\tr_{q^s/p}\left(a(-\alpha)^0 + b\alpha^{\frac{p^k+1}{2}\cdot0}\right)+c, \tr_{q^s/p}\left(a(-\alpha)^1 + b\alpha^{\frac{p^k+1}{2}\cdot1}\right)+c, \nonumber\\ && \cdots, \tr_{q^s/p}\left(a(-\alpha)^{q^s-2} + b\alpha^{\frac{p^k+1}{2}\cdot(q^s-2)}\right)+c, c\bigg): a, b \in \gf_{q^s},c \in \gf_p \Bigg\}.
\end{eqnarray}
In this subsection, we will determine the weight distribution of $\overline{\widehat{\cC_1}}$ and prove that $\overline{\widehat{\cC_1}}$ is self-orthogonal.

In order to determine the weight distribution of $\overline{\widehat{\cC_1}}$, we recall some lemmas as follows.

\begin{lemma}\cite{Zhou2014}\label{lem-weightC}
  Let $\frac{k}{e}$ and $s$ be odd. Then the cyclic code $\cC_1$ in Equation (\ref{eq-C}) has parameters $[p^m-1, 2m, p^m-p^{m-1}-(p-1)p^{(m+e-2)/2}]$ and its weight distribution is listed in Table \ref{tab4.1}.
  \begin{table}[h]
\begin{center}
\caption{The weight distribution of $\cC_1$ in Lemma \ref{lem-weightC}.}\label{tab4.1}
\begin{tabular}{@{}ll@{}}
\toprule%
Weight & Frequency  \\
\midrule
$0$ & $1$\\
$p^m-p^{m-1}-(p-1)p^{\frac{m+e-2}{2}}$ & $\frac{(p^{m-e} + p^{(m-e)/2})(p^m-1)}{2}$\\
$p^m-p^{m-1}$ & $(p^m-p^{m-e}+1)(p^m-1)$\\
$p^m-p^{m-1}+(p-1)p^{\frac{m+e-2}{2}}$ & $\frac{(p^{m-e} - p^{(m-e)/2})(p^m-1)}{2}$\\
\bottomrule
\end{tabular}
\end{center}
\end{table}
\end{lemma}

\begin{lemma}\cite{Heng2016}\label{lem-Nc}
 Let $\frac{k}{e}$ and $s$ be odd. Let $$N(c):=\left|\left\{ 0 \leq t \leq q^s-2: \tr_{q^s/p}(a(-\alpha)^t+b\alpha^{\frac{p^k+1}{2}t})=c\right\} \right|,$$ where $c \in \gf_p^*$. Then the value distribution of $N(c)$ is listed in Table \ref{tab4.2} if $e$ is even, or $e$ is odd and $p \equiv 1\pmod{4}$. The value distribution of $N(c)$ is listed in Table \ref{tab4.3} if $e$ is odd and $p \equiv 3\pmod{4}$.
  \begin{table}[h]
\begin{center}
\caption{The value distribution of $N(c)$ in Lemma \ref{lem-Nc} ($e$ is even, or $e$ is odd and $p \equiv 1\pmod{4}$).}\label{tab4.2}
\begin{tabular}{@{}ll@{}}
\toprule%
Value & Frequency  \\
\midrule
$0$ & $1$\\
$p^{m-1}+p^{\frac{m+e-2}{2}}$ & $\frac{(p^{m-e} - p^{(m-e)/2})(p^m-1)}{2}$\\
$p^{m-1}$ & $(p^m-p^{m-e}+1)(p^m-1)$\\
$p^{m-1}-p^{\frac{m+e-2}{2}}$ & $\frac{(p^{m-e} + p^{(m-e)/2})(p^m-1)}{2}$\\
\bottomrule
\end{tabular}
\end{center}
\end{table}
 \begin{table}[h]
\begin{center}
\caption{The value distribution of $N(c)$ in Lemma \ref{lem-Nc} ($e$ is odd and $p \equiv 3\pmod{4}$).}\label{tab4.3}
\begin{tabular}{@{}ll@{}}
\toprule%
Value & Frequency  \\
\midrule
$0$ & $1$\\
$p^{m-1}+p^{\frac{m+e-2}{2}}$ & $\frac{(p^{m-e} - p^{(m-e)/2})(p^m-1)}{2}$\\
$p^{m-1}+p^{\frac{m-1}{2}}$ & $\frac{(p^{m}-1)(p^m-p^{m-e}+1-\frac{p^{m-e}-1}{p^{2e}-1})}{2}$\\
$p^{m-1}-p^{\frac{m-1}{2}}$ & $\frac{(p^{m}-1)(p^m-p^{m-e}+1-\frac{p^{m-e}-1}{p^{2e}-1})}{2}$\\
$p^{m-1}+p^{\frac{m+2e-1}{2}}$ & $\frac{(p^{m}-1)(p^{m-e}-1)}{2(p^{2e}-1)}$\\
$p^{m-1}-p^{\frac{m+2e-1}{2}}$ & $\frac{(p^{m}-1)(p^{m-e}-1)}{2(p^{2e}-1)}$\\
$p^{m-1}-p^{\frac{m+e-2}{2}}$ & $\frac{(p^{m-e} + p^{(m-e)/2})(p^m-1)}{2}$\\
\bottomrule
\end{tabular}
\end{center}
\end{table}
\end{lemma}


\begin{theorem}\label{th-Ceven}
Let $m,k$ be two positive integers and $e=\text{gcd}(m,k)$. Let $q=p^e$, where $p$ is an odd prime. Let $s=\frac{m}{e} \geq 3$ and $\frac{k}{e}$ be odd.  If $e$ is even, or $e$ is odd and $p \equiv 1 \pmod{4}$, then the code $\overline{\widehat{\cC_1}}$ defined in Equation (\ref{eq-Cbar}) is a self-orthogonal $p$-divisible code with parameters $[p^m, 2m+1, p^m-p^{m-1}-(p-1)p^{\frac{m+e-2}{2}}]$ and its weight distribution is listed in Table \ref{tab4.4}. Besides, $\overline{\widehat{\cC_1}}^{\perp}$ has parameters $[p^m, p^m-2m-1, 3]$ and is at least almost optimal according to the sphere-packing bound if $p > 3$.
 \begin{table}[h]
\begin{center}
\caption{The weight distribution of $\overline{\widehat{\cC_1}}$ in Theorem \ref{th-Ceven} ($e$ is even, or $e$ is odd and $p \equiv 1\pmod{4}$).}\label{tab4.4}
\begin{tabular}{@{}ll@{}}
\toprule%
Weight & Frequency  \\
\midrule
$0$ & $1$\\
$p^m-p^{m-1}-(p-1)p^{\frac{m+e-2}{2}}$ & $\frac{(p^{m-e} + p^{(m-e)/2})(p^m-1)}{2}$\\
$p^m-p^{m-1}-p^{\frac{m+e-2}{2}}$ & $\frac{(p^{m-e} - p^{(m-e)/2})(p^m-1)(p-1)}{2}$\\
$p^m-p^{m-1}$ & $p(p^m-p^{m-e}+1)(p^m-1)$\\
$p^m-p^{m-1}+p^{\frac{m+e-2}{2}}$ & $\frac{(p^{m-e} + p^{(m-e)/2})(p^m-1)(p-1)}{2}$\\
$p^m-p^{m-1}+(p-1)p^{\frac{m+e-2}{2}}$ & $\frac{(p^{m-e} - p^{(m-e)/2})(p^m-1)}{2}$\\
$p^m$&$p-1$\\
\bottomrule
\end{tabular}
\end{center}
\end{table}
\end{theorem}

\begin{proof}
 For any $\bc_{(a,b,c)} \in \overline{\widehat{\cC_1}}$ and $\bc_{(a,b)} \in \cC_1$, it is obvious that $\text{wt}(\bc_{(a,b,c)})=\text{wt}(\bc_{(a,b)})$ for $c=0$. Then by Lemma \ref{lem-weightC}, we have
 \begin{eqnarray*}
   \text{wt}(\bc_{(a,b,0)}) &=& \begin{cases}
                                  0 & \mbox{$1$ time,} \\
                                 p^m-p^{m-1}-(p-1)p^{\frac{m+e-2}{2}} & \mbox{$\frac{(p^{m-e} + p^{(m-e)/2})(p^m-1)}{2}$ times,} \\
                                 p^m-p^{m-1} & \mbox{$(p^m-p^{m-e}+1)(p^m-1)$ times,}\\
                                  p^m-p^{m-1}+(p-1)p^{\frac{m+e-2}{2}} & \mbox{$\frac{(p^{m-e} - p^{(m-e)/2})(p^m-1)}{2}$ times.}
                                \end{cases}
 \end{eqnarray*}

 For $c \neq 0$, we have $\text{wt}(\bc_{(a,b,c)})=n-N(-c)=p^m-N(-c)$, where $N(-c)$ is defined in Lemma \ref{lem-Nc}. By Lemma \ref{lem-Nc}, we deduce that
 \begin{eqnarray*}
   \text{wt}(\bc_{(a,b,c)}) &=& \begin{cases}
                                  p^m & \mbox{$p-1$ times,} \\
                                 p^m-p^{m-1}-p^{\frac{m+e-2}{2}} & \mbox{$\frac{(p^{m-e} - p^{(m-e)/2})(p^m-1)(p-1)}{2}$ times,} \\
                                 p^m-p^{m-1} & \mbox{$(p^m-p^{m-e}+1)(p^m-1)(p-1)$ times,}\\
                                  p^m-p^{m-1}+p^{\frac{m+e-2}{2}} & \mbox{$\frac{(p^{m-e} + p^{(m-e)/2})(p^m-1)(p-1)}{2}$ times.}
                                \end{cases}
 \end{eqnarray*}
Furthermore, the dimension of $\overline{\widehat{\cC_1}}$ is $2m+1$ as the zero codeword in $\overline{\widehat{\cC_1}}$ occurs only once. Then the parameters and weight distribution of $\overline{\widehat{\cC_1}}$ directly follow.

 In the following, we determine the minimum distance of $\overline{\widehat{\cC_1}}^{\perp}$. Denote by $w_1=p^m-p^{m-1}-(p-1)p^{\frac{m+e-2}{2}}$, $w_2=p^m-p^{m-1}-p^{\frac{m+e-2}{2}}$, $w_3=p^m-p^{m-1}$, $w_4=p^m-p^{m-1}+p^{\frac{m+e-2}{2}}$, $w_5=p^m-p^{m-1}+(p-1)p^{\frac{m+e-2}{2}}$ and $w_6=p^m$. Let $A_{w_i} (1 \leq i \leq 6)$ denote the frequency of the weight $w_i$ in Table \ref{tab4.4}. By the second, third and fourth Pless power moments in \cite[Page 260]{H}, we have
 \begin{eqnarray*}
 \left\{\begin{array}{l}
         \sum_{i=1}^{6}w_iA_{w_i}=p^{2m}(pn-n-A_1^{\perp}), \\ \sum_{i=1}^{6}w_i^2A_{w_i}=p^{2m-1}\left((p-1)n(pn-n+1)-(2pn-p-2n+2)A_1^{\perp}+2A_2^{\perp}\right),\\ \sum_{i=1}^{6}w_i^3A_{w_i}=p^{2m-2}\left[(p-1)n(p^2n^2-2pn^2+3pn-p+n^2-3n+2)\right.\\
         \qquad\qquad\qquad\quad\left.-(3p^2n^2-3p^2n-6pn^2+12pn+p^2-6p+3n^2-9n+6)A_1^{\perp}\right.\\
         \qquad\qquad\qquad\quad\left. +6(pn-p-n+2)A_2^{\perp}-6A_3^\perp\right],
       \end{array}\right.
\end{eqnarray*}
where $n=p^m$. Solving the above system of linear equations gives
\begin{eqnarray*}
A_1^{\perp}=A_2^{\perp}=0,\ A_3^{\perp}=\frac{p^e(p^m-1)(p^2-3p+2)}{6}.
\end{eqnarray*}
Then $d(\overline{\widehat{\cC_1}}^{\perp})=3$ and $\overline{\widehat{\cC_1}}^{\perp}$ has parameters $[p^m, p^m-2m-1, 3]$. Besides, it is easy to verify that $\overline{\widehat{\cC_1}}^{\perp}$ is at least almost optimal according to the sphere-packing bound if $p > 3$ as the $p$-ary $[p^m, p^m-2m-1, 5]$ linear code dose not exist.

 It is obvious that $\mathbf{1} \in \overline{\widehat{\cC_1}}$.  Moreover, $\overline{\widehat{\cC_1}}$ is a $p$-divisible code and $\overline{\widehat{\cC_1}}$ is also a self-orthogonal code by Theorem \ref{th-selforthogonal}.
\end{proof}

\begin{theorem}\label{th-Codd}
Let $m,k$ be two positive integers and $e=\gcd(m,k)$. Let $q=p^e$, where $p$ is an odd prime. Let $s=\frac{m}{e} \geq 3$ and $\frac{k}{e}$ be odd. If $e$ is odd and $p \equiv 3 \pmod{4}$, then the code $\overline{\widehat{\cC_1}}$ defined in Equation (\ref{eq-Cbar}) is a self-orthogonal $p$-divisible code over $\gf_p$ with parameters $[p^m, 2m+1, p^m-p^{m-1}-(p-1)p^{\frac{m+e-2}{2}}]$ and its weight distribution is listed in Table \ref{tab4.5}.
 \begin{table}[h]
\begin{center}
\caption{The weight distribution of $\overline{\widehat{\cC_1}}$ in Theorem \ref{th-Ceven} ($e$ is odd and $p \equiv 3\pmod{4}$).}\label{tab4.5}
\begin{tabular}{@{}ll@{}}
\toprule%
Value & Frequency  \\
\midrule
$0$ & $1$\\
$p^m-p^{m-1}-(p-1)p^{\frac{m+e-2}{2}}$ & $\frac{(p^{m-e} + p^{(m-e)/2})(p^m-1)}{2}$\\
$p^m-p^{m-1}-p^{\frac{m+e-2}{2}}$ & $\frac{(p^{m-e} - p^{(m-e)/2})(p^m-1)(p-1)}{2}$\\
$p^m-p^{m-1}-p^{\frac{m+2e-1}{2}}$ & $\frac{(p^m-1)(p^{m-e}-1)(p-1)}{2(p^{2e}-1)}$\\
$p^m-p^{m-1}-p^{\frac{m-1}{2}}$ & $\frac{(p^m-1)(p^m-p^{m-e}+1-\frac{p^{m-e}-1}{p^{2e}-1})(p-1)}{2}$\\
$p^m-p^{m-1}$ & $(p^m-p^{m-e}+1)(p^m-1)$\\
$p^m-p^{m-1}+p^{\frac{m-1}{2}}$ & $\frac{(p^m-1)(p^m-p^{m-e}+1-\frac{p^{m-e}-1}{p^{2e}-1})(p-1)}{2}$\\
$p^m-p^{m-1}+p^{\frac{m+2e-1}{2}}$ & $\frac{(p^m-1)(p^{m-e}-1)(p-1)}{2(p^{2e}-1)}$\\
$p^m-p^{m-1}+p^{\frac{m+e-2}{2}}$ & $\frac{(p^{m-e} + p^{(m-e)/2})(p^m-1)(p-1)}{2}$\\
$p^m-p^{m-1}+(p-1)p^{\frac{m+e-2}{2}}$ & $\frac{(p^{m-e} - p^{(m-e)/2})(p^m-1)}{2}$\\
$p^m$&$p-1$\\
\bottomrule
\end{tabular}
\end{center}
\end{table}
\end{theorem}

\begin{proof}
   For any $\bc_{(a,b,c)} \in \overline{\widehat{\cC_1}}$ and $\bc_{(a,b)} \in \cC_1$, it is obvious that $\text{wt}(\bc_{(a,b,c)})=\text{wt}(\bc_{(a,b)})$ for $c=0$. Then by Lemma \ref{lem-weightC}, we have
 \begin{eqnarray*}
   \text{wt}(\bc_{(a,b,0)}) &=& \begin{cases}
                                  0 & \mbox{$1$ time,} \\
                                 p^m-p^{m-1}-(p-1)p^{\frac{m+e-2}{2}} & \mbox{$\frac{(p^{m-e} + p^{(m-e)/2})(p^m-1)}{2}$ times,} \\
                                 p^m-p^{m-1} & \mbox{$(p^m-p^{m-e}+1)(p^m-1)$ times,}\\
                                  p^m-p^{m-1}+(p-1)p^{\frac{m+e-2}{2}} & \mbox{$\frac{(p^{m-e} - p^{(m-e)/2})(p^m-1)}{2}$ times.}
                                \end{cases}
 \end{eqnarray*}
 For $c \neq 0$, we have $\text{wt}(\bc_{(a,b,c)})=n-N(-c)=p^m-N(-c)$, where $N(-c)$ is defined in Lemma \ref{lem-Nc}. By Lemma \ref{lem-Nc}, we deduce that
 \begin{eqnarray*}
   \text{wt}(\bc_{(a,b,c)}) &=& \begin{cases}
                                  p^m & \mbox{$p-1$ times,} \\
                                 p^m-p^{m-1}-p^{\frac{m+e-2}{2}} & \mbox{$\frac{(p^{m-e} - p^{(m-e)/2})(p^m-1)(p-1)}{2}$ times,} \\
                                 p^m-p^{m-1}-p^{\frac{m-1}{2}} & \mbox{$\frac{(p^{m}-1)(p^m-p^{m-e}+1-\frac{p^{m-e}-1}{p^{2e}-1})(p-1)}{2}$ times,}\\
                                 p^m-p^{m-1}+p^{\frac{m-1}{2}} & \mbox{$\frac{(p^{m}-1)(p^m-p^{m-e}+1-\frac{p^{m-e}-1}{p^{2e}-1})(p-1)}{2}$ times,}\\
                                 p^m-p^{m-1}-p^{\frac{m+2e-1}{2}} & \mbox{$\frac{(p^{m}-1)(p^{m-e}-1)(p-1)}{2(p^{2e}-1)}$ times,}\\
                                 p^m-p^{m-1}+p^{\frac{m+2e-1}{2}} & \mbox{$\frac{(p^{m}-1)(p^{m-e}-1)(p-1)}{2(p^{2e}-1)}$ times,}\\
                                 p^m-p^{m-1}+p^{\frac{m+e-2}{2}} & \mbox{$\frac{(p^{m-e} + p^{(m-e)/2})(p^m-1)(p-1)}{2}$ times.}
                                \end{cases}
 \end{eqnarray*}
Furthermore, the dimension of $\overline{\widehat{\cC_1}}$ is $2m+1$ as the zero codeword in $\overline{\widehat{\cC_1}}$ occurs only once. Then the parameters and weight distribution of $\overline{\widehat{\cC_1}}$ directly follow.

 It is easy to deduce that $\mathbf{1} \in \overline{\widehat{\cC_1}}$ and $\overline{\widehat{\cC_1}}$ is a $p$-divisible code and then $\overline{\widehat{\cC_1}}$ is a self-orthogonal code by Theorem \ref{th-selforthogonal}.
\end{proof}

\subsection{The second family of self-orthogonal codes from Luo-Feng cyclic codes}

In this subsection, let $m,k$ be two positive integers and $e=\text{gcd}(m,k)$. Let $q=p^e$, where $p$ is an odd prime. Let $s=\frac{m}{e} \geq 3$ be odd and $\frac{k}{e}$ be even. Denote by $\gf_{q^s}^*=\langle\alpha\rangle$. Let $h_1(x)$ and $h_2(x)$ denote the minimal polynomials of $\alpha^{-1}$ and $\alpha^{-\frac{p^k+1}{2}}$ over $\gf_p$, respectively. Define a family of $p$-ary cyclic codes as
\begin{eqnarray}\label{eq-C1}
  \cC_2 &=& \left\{ \bc_{(a,b)}=\left(\tr_{q^s/p}\left(a\alpha^t + b\alpha^{\frac{(p^k+1)t}{2}}\right)\right)_{t=0}^{q^s-2}: a,b \in \gf_{q^s} \right\}
\end{eqnarray}
with parity-check polynomial $h_1(x)h_2(x)$. In \cite{LF}, Luo and Feng gave the weight distribution of $\cC_2$ in Equation (\ref{eq-C1}) for odd $\frac{k}{e}$ and any $s$. They proposed an open problem to give the weight distribution of $\cC_2$ for even $\frac{k}{e}$ and any $s$. This open problem was solved by Heng and Yue \cite{Heng2016}. The weight distribution and complete weight distribution of $\cC_2$ for even $\frac{k}{e}$ and odd $s$ were given in \cite{Heng2016}.

Let $\overline{\widehat{\cC_2}}$ be the augmented code of extended code of cyclic code $\cC_2$. By the definitions of the extended and augmented codes of linear codes, we have
\begin{eqnarray}\label{eq-C1bar}
  \overline{\widehat{\cC_2}} &=&\Bigg \{\bc_{(a,b,c)}=\bigg(\tr_{q^s/p}\left(a\alpha^0 + b\alpha^{\frac{p^k+1}{2}\cdot0}\right)+c, \tr_{q^s/p}\left(a\alpha^1 + b\alpha^{\frac{p^k+1}{2}\cdot1}\right)+c, \nonumber\\ && \cdots, \tr_{q^s/p}\left(a\alpha^{q^s-2} + b\alpha^{\frac{p^k+1}{2}\cdot(q^s-2)}\right)+c, c\bigg): a, b \in \gf_{q^s},c \in \gf_p \Bigg\}.
\end{eqnarray}
In this subsection, we will determine the weight distribution of $\overline{\widehat{\cC_2}}$ at first and then prove that $\overline{\widehat{\cC_2}}$ is self-orthogonal.

In order to determine the weight distribution of $\overline{\widehat{\cC_2}}$, we recall some lemmas as follows.

\begin{lemma}\cite{Heng2016}\label{lem-weightC1}
  Let $\frac{k}{e}$ be even and $s$ be odd. Then the cyclic code $\cC_2$ in Equation (\ref{eq-C1}) has parameters $[p^m-1, 2m, p^m-p^{m-1}-(p-1)p^{(m+e-2)/2}]$ and its weight distribution is listed in Table \ref{tab4.6}.
  \begin{table}[h]
\begin{center}
\caption{The weight distribution of $\cC_2$ in Lemma \ref{lem-weightC1}.}\label{tab4.6}
\begin{tabular}{@{}ll@{}}
\toprule%
Weight & Frequency  \\
\midrule
$0$ & $1$\\
$p^m-p^{m-1}-(p-1)p^{\frac{m+e-2}{2}}$ & $\frac{(p^{m-e} + p^{(m-e)/2})(p^m-1)}{2}$\\
$p^m-p^{m-1}$ & $(p^m-p^{m-e}+1)(p^m-1)$\\
$p^m-p^{m-1}+(p-1)p^{\frac{m+e-2}{2}}$ & $\frac{(p^{m-e} - p^{(m-e)/2})(p^m-1)}{2}$\\
\bottomrule
\end{tabular}
\end{center}
\end{table}
\end{lemma}

\begin{lemma}\cite{Heng2016}\label{lem-Nc1}
 Let $\frac{k}{e}$ be even and $s$ be odd. Let $N(c)=| \{ 0 \leq t \leq q^s-2: \tr_{q^s/p}(a\alpha^t+b\alpha^{\frac{p^k+1}{2}t})=c\}|$, where $c \in \gf_p^*$. Then the value distribution of $N(c)$ is listed in Table \ref{tab4.7}.
  \begin{table}[h]
\begin{center}
\caption{The value distribution of $N(c)$ in Lemma \ref{lem-Nc1}.}\label{tab4.7}
\begin{tabular}{@{}ll@{}}
\toprule%
Value & Frequency  \\
\midrule
$0$ & $1$\\
$p^{m-1}+p^{\frac{m+e-2}{2}}$ & $\frac{(p^{m-e} - p^{(m-e)/2})(p^m-1)}{2}$\\
$p^{m-1}$ & $(p^m-p^{m-e}+1)(p^m-1)$\\
$p^{m-1}-p^{\frac{m+e-2}{2}}$ & $\frac{(p^{m-e} + p^{(m-e)/2})(p^m-1)}{2}$\\
\bottomrule
\end{tabular}
\end{center}
\end{table}
\end{lemma}


\begin{theorem}\label{th-C1}
Let $m,k$ be two positive integers and $e=\text{gcd}(m,k)$. Let $q=p^e$, where $p$ is an odd prime. Let $s=\frac{m}{e} \geq 3$ be odd and $\frac{k}{e}$ be even.  Then the code $\overline{\widehat{\cC_2}}$ defined in Equation (\ref{eq-C1bar}) is a self-orthogonal $p$-divisible code over $\gf_p$ with parameters $[p^m, 2m+1, p^m-p^{m-1}-(p-1)p^{\frac{m+e-2}{2}}]$ and its weight distribution is listed in Table \ref{tab4.8}. Besides, $\overline{\widehat{\cC_2}}^{\perp}$ has parameters $[p^m,p^m-2m-1,3]$ and is at least almost optimal according to the sphere-packing bound if $p > 3$.
 \begin{table}[h]
\begin{center}
\caption{The weight distribution of $\overline{\widehat{\cC_2}}$ in Theorem \ref{th-C1}.}\label{tab4.8}
\begin{tabular}{@{}ll@{}}
\toprule%
Weight & Frequency  \\
\midrule
$0$ & $1$\\
$p^m-p^{m-1}-(p-1)p^{\frac{m+e-2}{2}}$ & $\frac{(p^{m-e} + p^{(m-e)/2})(p^m-1)}{2}$\\
$p^m-p^{m-1}-p^{\frac{m+e-2}{2}}$ & $\frac{(p^{m-e} - p^{(m-e)/2})(p^m-1)(p-1)}{2}$\\
$p^m-p^{m-1}$ & $p(p^m-p^{m-e}+1)(p^m-1)$\\
$p^m-p^{m-1}+p^{\frac{m+e-2}{2}}$ & $\frac{(p^{m-e} + p^{(m-e)/2})(p^m-1)(p-1)}{2}$\\
$p^m-p^{m-1}+(p-1)p^{\frac{m+e-2}{2}}$ & $\frac{(p^{m-e} - p^{(m-e)/2})(p^m-1)}{2}$\\
$p^m$&$p-1$\\
\bottomrule
\end{tabular}
\end{center}
\end{table}
\end{theorem}

\begin{proof}
 Similarly to the proof of Theorem \ref{th-Ceven}, the parameters and weight distribution of $\overline{\widehat{\cC_2}}$ can also be determined by Lemmas \ref{lem-weightC1} and \ref{lem-Nc1}. The weight distribution of $\overline{\widehat{\cC_2}}$ is given in Table \ref{tab4.8}. In the following, we will determine the minimum distance of $\overline{\widehat{\cC_2}}^{\perp}$. Denote by $w_1=p^m-p^{m-1}-(p-1)p^{\frac{m+e-2}{2}}$, $w_2=p^m-p^{m-1}-p^{\frac{m+e-2}{2}}$, $w_3=p^m-p^{m-1}$, $w_4=p^m-p^{m-1}+p^{\frac{m+e-2}{2}}$, $w_5=p^m-p^{m-1}+(p-1)p^{\frac{m+e-2}{2}}$ and $w_6=p^m$. Let $A_{w_i} (1 \leq i \leq 6)$ denote the frequency of the weight $w_i$ in Table \ref{tab4.8}. By the second, third and fourth Pless power moments in \cite[Page 260]{H}, we have
 \begin{eqnarray*}
 \left\{\begin{array}{l}
         \sum_{i=1}^{6}w_iA_{w_i}=p^{2m}(pn-n-A_1^{\perp}), \\ \sum_{i=1}^{6}w_i^2A_{w_i}=p^{2m-1}\left((p-1)n(pn-n+1)-(2pn-p-2n+2)A_1^{\perp}+2A_2^{\perp}\right),\\ \sum_{i=1}^{6}w_i^3A_{w_i}=p^{2m-2}\left[(p-1)n(p^2n^2-2pn^2+3pn-p+n^2-3n+2)\right.\\
         \qquad\qquad\qquad\quad\left.-(3p^2n^2-3p^2n-6pn^2+12pn+p^2-6p+3n^2-9n+6)A_1^{\perp}\right.\\
         \qquad\qquad\qquad\quad\left. +6(pn-p-n+2)A_2^{\perp}-6A_3^\perp\right],
       \end{array}\right.
\end{eqnarray*}
where $n=p^m$. Solving the above system of linear equations gives
\begin{eqnarray*}
A_1^{\perp}=A_2^{\perp}=0,\ A_3^{\perp}=\frac{p^e(p^m-1)(p^2-3p+2)}{6}.
\end{eqnarray*}
Then $d(\overline{\widehat{\cC_2}}^{\perp})=3$ and $\overline{\widehat{\cC_2}}^{\perp}$ has parameters $[p^m, p^m-2m-1, 3]$. Besides, it is easy to verify that $\overline{\widehat{\cC_2}}^{\perp}$ is at least almost optimal according to the sphere-packing bound if $p > 3$ as the $p$-ary $[p^m, p^m-2m-1, 5]$ code does not exist.

 It is obvious that $\mathbf{1} \in \overline{\widehat{\cC_2}}$.  Moreover, $\overline{\widehat{\cC_2}}$ is a $p$-divisible code and $\overline{\widehat{\cC_2}}$ is a self-orthogonal code by Theorem \ref{th-selforthogonal}.
\end{proof}

\subsection{The third family of self-orthogonal codes from Ma-Zeng-Liu-Feng-Ding cyclic codes}

In this subsection, let $q=p^e$, where $p$ is an odd prime. Let $\gf_{q^m}^*=\langle\theta\rangle$, $g_1=\theta$ and $g_2=\theta^{\frac{q^m-1}{2}+1}$. Let $h_1(x)$ and $h_2(x)$ denote the minimal polynomials of $g_1^{-1}$ and $g_2^{-1}$ over $\gf_q$, respectively. Define a family of cyclic codes as
\begin{eqnarray}\label{eq-C3}
  \cC_3 =\left \{ \bc_{(a,b)}=\left(\tr_{q^m/q}(ag_1^t + bg_2^{t})\right)_{t=0}^{q^m-2}: a,b \in \gf_{q^m} \right\}
\end{eqnarray}
with parity-check polynomial $h_1(x)h_2(x)$. The weight distribution of $\cC_3$ in Equation (\ref{eq-C3}) was given by Ma et al. in \cite{Ma2011} and its complete weight distribution  was given in \cite{Li2015}.

Let $\overline{\widehat{\cC_3}}$ be the augmented code of extended code of cyclic code $\cC_3$. By the definitions of the extended and augmented codes of linear codes, we have
\begin{eqnarray}\label{eq-C3bar}
  \overline{\widehat{\cC_3}} &=& \bigg\{\bc_{(a,b,c)}=\Big(\tr_{q^m/q}(ag_1^0 + bg_2^{0})+c, \tr_{q^m/q}(ag_1^1 + bg_2^{1})+c, \cdots, \nonumber\\&& \tr_{q^m/q}(ag_1^{q^m-2} + bg_2^{q^m-2})+c, c\Big): a, b \in \gf_{q^m},c \in \gf_q \bigg\}.
\end{eqnarray}

In the subsection, we will determine the weight distribution of $\overline{\widehat{\cC_3}}$ and prove that $\overline{\widehat{\cC_3}}$ is self-orthogonal.

In order to determine the weight distribution of $\overline{\widehat{\cC_3}}$, we recall some lemmas as follows.

\begin{lemma}\cite{Ma2011}\label{lem-weightC3}
  Let $q$ be an odd prime power and $m$ be a positive integer. Then the cyclic code $\cC_3$ in Equation (\ref{eq-C3}) has parameters $[q^m-1, 2m, d]$, where
  \begin{eqnarray*}
    d &=& \begin{cases}
            \frac{(q-1)q^{m-1}}{2}, & \mbox{if $m$ is odd}, \\
            \frac{(q-1)q^{\frac{m-2}{2}}(q^{\frac{m}{2}}-1)}{2}, & \mbox{if $m$ is even}.
          \end{cases}
  \end{eqnarray*}
   Besides, its weight distributions are listed in Tables \ref{tab4.9} and \ref{tab4.10} for odd $m$ and even $m$, respectively.
 \begin{table}[h]
\begin{center}
\caption{The weight distribution of $\cC_3$ in Lemma \ref{lem-weightC3} ($m$ is odd).}\label{tab4.9}
\begin{tabular}{@{}ll@{}}
\toprule%
Weight & Frequency  \\
\midrule
$0$ & $1$\\
$(q-1)q^{m-1}$ & $(q^m-1)^2$\\
$\frac{(q-1)q^{m-1}}{2}$ & $2(q^m-1)$\\
\bottomrule
\end{tabular}
\end{center}
\end{table}
 \begin{table}[h]
\begin{center}
\caption{The weight distribution of $\cC_3$ in Lemma \ref{lem-weightC3} ($m$ is even).}\label{tab4.10}
\begin{tabular}{@{}ll@{}}
\toprule%
Weight & Frequency  \\
\midrule
$0$ & $1$\\
$(q-1)q^{m-1}$ & $\frac{(q^m-1)^2}{2}$\\
$\frac{(q-1)q^{\frac{m-2}{2}}(q^{\frac{m}{2}}+1)}{2}$ & $q^m-1$\\
$\frac{(q-1)q^{\frac{m-2}{2}}(q^{\frac{m}{2}}-1)}{2}$ & $q^m-1$\\
$(q-1)q^{\frac{m-2}{2}}(q^{\frac{m}{2}}+1)$ & $\frac{(q^m-1)^2}{4}$\\
$(q-1)q^{\frac{m-2}{2}}(q^{\frac{m}{2}}-1)$ & $\frac{(q^m-1)^2}{4}$\\
\bottomrule
\end{tabular}
\end{center}
\end{table}
\end{lemma}

\begin{lemma}\cite{Li2015}\label{lem-Nc3}
 Let $q$ be an odd prime power and $m$ be a positive integer. Let $N(c)=\mid \{ 0 \leq t \leq q^m-2: \tr_{q^m/q}(ag_1^t+bg_2^{t})=c\} \mid$, where $c \in \gf_q^*$. Then the value distributions of $N(c)$ are listed in Tables \ref{tab4.11} and \ref{tab4.12} for odd $m$ and even $m$, respectively.
  \begin{table}[h]
\begin{center}
\caption{The value distribution of $N(c)$ in Lemma \ref{lem-Nc3} ($m$ is odd).}\label{tab4.11}
\begin{tabular}{@{}ll@{}}
\toprule%
Value & Frequency  \\
\midrule
$0$ & $q-1$\\
$q^{m-1}$ & $\frac{(q^m-1)^2(q-1)}{2}$\\
$q^{m-1}+q^{\frac{m-1}{2}}$ & $\frac{(q^m-1)^2(q-1)}{4}$\\
$q^{m-1}-q^{\frac{m-1}{2}}$ & $\frac{(q^m-1)^2(q-1)}{4}$\\
$\frac{q^{m-1}+q^{\frac{m-1}{2}}}{2}$ & $(q^m-1)(q-1)$\\
$\frac{q^{m-1}-q^{\frac{m-1}{2}}}{2}$ & $(q^m-1)(q-1)$\\
\bottomrule
\end{tabular}
\end{center}
\end{table}
  \begin{table}[h]
\begin{center}
\caption{The value distribution of $N(c)$ in Lemma \ref{lem-Nc3} ($m$ is even).}\label{tab4.12}
\begin{tabular}{@{}ll@{}}
\toprule%
Value & Frequency  \\
\midrule
$0$ & $q-1$\\
$q^{m-1}$ & $\frac{(q^m-1)^2(q-1)}{2}$\\
$q^{m-1}+q^{\frac{m-2}{2}}$ & $\frac{(q^m-1)^2(q-1)}{4}$\\
$q^{m-1}-q^{\frac{m-2}{2}}$ & $\frac{(q^m-1)^2(q-1)}{4}$\\
$\frac{q^{m-1}+q^{\frac{m-2}{2}}}{2}$ & $(q^m-1)(q-1)$\\
$\frac{q^{m-1}-q^{\frac{m-2}{2}}}{2}$ & $(q^m-1)(q-1)$\\
\bottomrule
\end{tabular}
\end{center}
\end{table}
\end{lemma}


\begin{theorem}\label{th-C3odd}
Let $q=p^e$ be an odd prime power and $m$ be an odd integer with $m \geq 3$. Then the code $\overline{\widehat{\cC_3}}$ defined in Equation (\ref{eq-C3bar}) is a self-orthogonal $p$-divisible code over $\gf_q$ with parameters $[q^m, 2m+1, \frac{q^{m-1}(q-1)}{2}]$ and its weight distribution is listed in Table \ref{tab4.13}. Besides, $\overline{\widehat{\cC_3}}^{\perp}$ has parameters $[q^m,q^m-2m-1,3]$ and is at least almost optimal according to the sphere-packing bound if $q > 3$.
 \begin{table}[h]
\begin{center}
\caption{The weight distribution of $\overline{\widehat{\cC_3}}$ in Theorem \ref{th-C3odd} ($m$ is odd).}\label{tab4.13}
\begin{tabular}{@{}ll@{}}
\toprule%
Weight & Frequency  \\
\midrule
$0$ & $1$\\
$q^{m}-q^{m-1}$ & $\frac{(q^m-1)^2(q+1)}{2}$\\
$q^{m}-q^{m-1}-q^{\frac{m-1}{2}}$ & $\frac{(q^m-1)^2(q-1)}{4}$\\
$q^{m}-q^{m-1}+q^{\frac{m-1}{2}}$ & $\frac{(q^m-1)^2(q-1)}{4}$\\
$q^m-\frac{q^{m-1}-q^{\frac{m-1}{2}}}{2}$ & $(q^m-1)(q-1)$\\
$q^m-\frac{q^{m-1}+q^{\frac{m-1}{2}}}{2}$ & $(q^m-1)(q-1)$\\
$\frac{q^{m-1}(q-1)}{2}$ & $2(q^m-1)$\\
$q^m$ & $q-1$\\
\bottomrule
\end{tabular}
\end{center}
\end{table}
\end{theorem}

\begin{proof}
  For any $\bc_{(a,b,c)} \in \overline{\widehat{\cC_3}}$ and $\bc_{(a,b)} \in \cC_3$, it is easy to deduce that $\text{wt}(\bc_{(a,b,c)})=\text{wt}(\bc_{(a,b)})$ for $c=0$. Then by Lemma \ref{lem-weightC3}, we have
 \begin{eqnarray*}
   \text{wt}(\bc_{(a,b,0)}) &=& \begin{cases}
                                  0 & \mbox{$1$ time,} \\
                                 q^{m}-q^{m-1} & \mbox{$(q^m-1)^2$ times,} \\
                                  \frac{q^{m-1}(q-1)}{2} & \mbox{$2(q^m-1)$ times.}
                                \end{cases}
 \end{eqnarray*}
 For $c \neq 0$, we have $\text{wt}(\bc_{(a,b,c)})=n-N(-c)=q^m-N(-c)$, where $N(-c)$ is defined in Lemma \ref{lem-Nc3}. By Lemma \ref{lem-Nc3}, we deduce that
 \begin{eqnarray*}
   \text{wt}(\bc_{(a,b,c)}) &=& \begin{cases}
                                  q^m & \mbox{$q-1$ times,} \\
                                  q^{m}-q^{m-1} & \mbox{$\frac{(q^m-1)^2(q-1)}{2}$ times,}\\
                                  q^{m}-q^{m-1}-q^{\frac{m-1}{2}} & \mbox{$\frac{(q^m-1)^2(q-1)}{4}$ times,}\\
                                  q^{m}-q^{m-1}+q^{\frac{m-1}{2}} & \mbox{$\frac{(q^m-1)^2(q-1)}{4}$ times,}\\
                                  q^m-\frac{q^{m-1}-q^{\frac{m-1}{2}}}{2} & \mbox{$(q^m-1)(q-1)$ times,}\\
                                  q^m-\frac{q^{m-1}+q^{\frac{m-1}{2}}}{2} & \mbox{$(q^m-1)(q-1)$ times.}
                                \end{cases}
 \end{eqnarray*}
Furthermore, the dimension of $\overline{\widehat{\cC_3}}$ is $2m+1$ as the zero codeword in $\overline{\widehat{\cC_3}}$ occurs only once. Then the parameters and weight distribution of $\overline{\widehat{\cC_3}}$ directly follow.

 In the following, we will determine the minimum distance of $\overline{\widehat{\cC_3}}^{\perp}$. Denote by $w_1=q^{m}-q^{m-1}$, $w_2=q^{m}-q^{m-1}-q^{\frac{m-1}{2}}$, $w_3=q^{m}-q^{m-1}+q^{\frac{m-1}{2}}$, $w_4=q^m-\frac{q^{m-1}-q^{\frac{m-1}{2}}}{2}$, $w_5=q^m-\frac{q^{m-1}+q^{\frac{m-1}{2}}}{2}$, $w_6=\frac{q^{m-1}(q-1)}{2}$ and $w_7=q^m$. Let $A_{w_i} (1 \leq i \leq 7)$ represent the frequency of the weight $w_i$ in Table \ref{tab4.13}. By the second, third and fourth Pless power moments in \cite[Page 260]{H}, we have
 \begin{eqnarray*}
 \left\{\begin{array}{l}
         \sum_{i=1}^{7}w_iA_{w_i}=q^{2m}(qn-n-A_1^{\perp}), \\ \sum_{i=1}^{7}w_i^2A_{w_i}=q^{2m-1}\left((q-1)n(qn-n+1)-(2qn-q-2n+2)A_1^{\perp}+2A_2^{\perp}\right),\\ \sum_{i=1}^{7}w_i^3A_{w_i}=q^{2m-2}\left[(q-1)n(q^2n^2-2qn^2+3qn-q+n^2-3n+2)\right.\\
         \qquad\qquad\qquad\quad\left.-(3q^2n^2-3q^2n-6qn^2+12qn+q^2-6q+3n^2-9n+6)A_1^{\perp}\right.\\
         \qquad\qquad\qquad\quad\left. +6(qn-q-n+2)A_2^{\perp}-6A_3^\perp\right],
       \end{array}\right.
\end{eqnarray*}
where $n=q^m$. Solving the above system of linear equations gives
\begin{eqnarray*}
A_1^{\perp}=A_2^{\perp}=0,\ A_3^{\perp}=\frac{(q^m-1)(q-1)(q^m(q-2)-3)}{24}.
\end{eqnarray*}
Then $d(\overline{\widehat{\cC_3}}^{\perp})=3$ and $\overline{\widehat{\cC_3}}^{\perp}$ has parameters $[q^m, q^m-2m-1, 3]$. Besides, it is easy to verify that $\overline{\widehat{\cC_3}}^{\perp}$ is at least almost optimal according to the sphere-packing bound if $q > 3$ as the $q$-ary $[q^m, q^m-2m-1, 5]$ code does not exist.

 It is easy to deduce that $\mathbf{1} \in \overline{\widehat{\cC_3}}$ and $\overline{\widehat{\cC_3}}$ is a $p$-divisible code and $\overline{\widehat{\cC_3}}$ is a self-orthogonal code by Theorem \ref{th-selforthogonal}.
\end{proof}

\begin{theorem}\label{th-C3even}
Let $q=p^e$ be an odd prime power and $m$ be an even integer with $m \geq 4$. Then the code $\overline{\widehat{\cC_3}}$ defined in Equation (\ref{eq-C3bar}) is a self-orthogonal $p$-divisible code over $\gf_q$ with parameters $[q^m, 2m+1, \frac{q^{\frac{m-2}{2}}(q-1)(q^{\frac{m}{2}}-1)}{2}]$ and its weight distribution is listed in Table \ref{tab4.14}. Besides, $\overline{\widehat{\cC_3}}^{\perp}$ has parameters $[q^m,q^m-2m-1,3]$ and is at least almost optimal according to the sphere-packing bound if $q > 3$.
 \begin{table}[h]
\begin{center}
\caption{The weight distribution of $\overline{\widehat{\cC_3}}$ in Theorem \ref{th-C3even} ($m$ is even).}\label{tab4.14}
\begin{tabular}{@{}ll@{}}
\toprule%
Weight & Frequency  \\
\midrule
$0$ & $1$\\
$q^{m}-q^{m-1}$ & $\frac{q(q^m-1)^2}{2}$\\
$q^{m}-q^{m-1}-q^{\frac{m-2}{2}}$ & $\frac{(q^m-1)^2(q-1)}{4}$\\
$q^{m}-q^{m-1}+q^{\frac{m-2}{2}}$ & $\frac{(q^m-1)^2(q-1)}{4}$\\
$q^m-\frac{q^{m-1}-q^{\frac{m-2}{2}}}{2}$ & $(q^m-1)(q-1)$\\
$q^m-\frac{q^{m-1}+q^{\frac{m-2}{2}}}{2}$ & $(q^m-1)(q-1)$\\
$q^{\frac{m-2}{2}}(q-1)(q^{\frac{m}{2}}+1)$ & $\frac{(q^m-1)^2}{4}$\\
$q^{\frac{m-2}{2}}(q-1)(q^{\frac{m}{2}}-1)$ & $\frac{(q^m-1)^2}{4}$\\
$\frac{q^{\frac{m-2}{2}}(q-1)(q^{\frac{m}{2}}+1)}{2}$ & $q^m-1$\\
$\frac{q^{\frac{m-2}{2}}(q-1)(q^{\frac{m}{2}}-1)}{2}$ & $q^m-1$\\
$q^m$ & $q-1$\\
\bottomrule
\end{tabular}
\end{center}
\end{table}
\end{theorem}

\begin{proof}
  For any $\bc_{(a,b,c)} \in \overline{\widehat{\cC_3}}$ and $\bc_{(a,b)} \in \cC_3$, it is easy to deduce that $\text{wt}(\bc_{(a,b,c)})=\text{wt}(\bc_{(a,b)})$ for $c=0$. Then by Lemma \ref{lem-weightC3}, we have
 \begin{eqnarray*}
   \text{wt}(\bc_{(a,b,0)}) &=& \begin{cases}
                                  0 & \mbox{$1$ time,} \\
                                 q^{m}-q^{m-1} & \mbox{$\frac{(q^m-1)^2}{2}$ times,} \\
                                 q^{\frac{m-2}{2}}(q-1)(q^{\frac{m}{2}}+1) & \mbox{$\frac{(q^m-1)^2}{4}$ times,}\\
                                 q^{\frac{m-2}{2}}(q-1)(q^{\frac{m}{2}}-1) & \mbox{$\frac{(q^m-1)^2}{4}$ times,}\\
                                 \frac{q^{\frac{m-2}{2}}(q-1)(q^{\frac{m}{2}}+1)}{2} & \mbox{$q^m-1$ times,}\\
                                 \frac{q^{\frac{m-2}{2}}(q-1)(q^{\frac{m}{2}}-1)}{2} & \mbox{$q^m-1$ times.}
                                \end{cases}
 \end{eqnarray*}
 For $c \neq 0$, we have $\text{wt}(\bc_{(a,b,c)})=n-N(-c)=q^m-N(-c)$, where $N(-c)$ is defined in Lemma \ref{lem-Nc3}. By Lemma \ref{lem-Nc3}, we deduce that
 \begin{eqnarray*}
   \text{wt}(\bc_{(a,b,c)}) &=& \begin{cases}
                                  q^m & \mbox{$q-1$ times,} \\
                                  q^{m}-q^{m-1} & \mbox{$\frac{(q^m-1)^2(q-1)}{2}$ times,}\\
                                  q^{m}-q^{m-1}-q^{\frac{m-2}{2}} & \mbox{$\frac{(q^m-1)^2(q-1)}{4}$ times,}\\
                                  q^{m}-q^{m-1}+q^{\frac{m-2}{2}} & \mbox{$\frac{(q^m-1)^2(q-1)}{4}$ times,}\\
                                  q^m-\frac{q^{m-1}-q^{\frac{m-2}{2}}}{2} & \mbox{$(q^m-1)(q-1)$ times,}\\
                                  q^m-\frac{q^{m-1}+q^{\frac{m-2}{2}}}{2} & \mbox{$(q^m-1)(q-1)$ times.}
                                \end{cases}
 \end{eqnarray*}
Furthermore, the dimension of $\overline{\widehat{\cC_3}}$ is $2m+1$ as the zero codeword in $\overline{\widehat{\cC_3}}$ occurs only once. Then the parameters and weight distribution of $\overline{\widehat{\cC_3}}$ directly follow.

 In the following, we will determine the minimum distance of $\overline{\widehat{\cC_3}}^{\perp}$. Denote by $w_1=q^{m}-q^{m-1}$, $w_2=q^{m}-q^{m-1}-q^{\frac{m-2}{2}}$, $w_3=q^{m}-q^{m-1}+q^{\frac{m-2}{2}}$, $w_4=q^m-\frac{q^{m-1}-q^{\frac{m-2}{2}}}{2}$, $w_5=q^m-\frac{q^{m-1}+q^{\frac{m-2}{2}}}{2}$, $w_6=q^{\frac{m-2}{2}}(q-1)(q^{\frac{m}{2}}+1)$, $w_7=q^{\frac{m-2}{2}}(q-1)(q^{\frac{m}{2}}-1)$, $w_8=\frac{q^{\frac{m-2}{2}}(q-1)(q^{\frac{m}{2}}+1)}{2}$, $w_9=\frac{q^{\frac{m-2}{2}}(q-1)(q^{\frac{m}{2}}-1)}{2}$ and $w_{10}=q^m$. Let $A_{w_i} (1 \leq i \leq 10)$ denote the frequency of the weight $w_i$ in Table \ref{tab4.14}. By the second, third and fourth Pless power moments in \cite[Page 260]{H}, we have
 \begin{eqnarray*}
 \left\{\begin{array}{l}
         \sum_{i=1}^{10}w_iA_{w_i}=q^{2m}(qn-n-A_1^{\perp}), \\ \sum_{i=1}^{10}w_i^2A_{w_i}=q^{2m-1}\left((q-1)n(qn-n+1)-(2qn-q-2n+2)A_1^{\perp}+2A_2^{\perp}\right),\\ \sum_{i=1}^{10}w_i^3A_{w_i}=q^{2m-2}\left[(q-1)n(q^2n^2-2qn^2+3qn-q+n^2-3n+2)\right.\\
         \qquad\qquad\qquad\quad\left.-(3q^2n^2-3q^2n-6qn^2+12qn+q^2-6q+3n^2-9n+6)A_1^{\perp}\right.\\
         \qquad\qquad\qquad\quad\left. +6(qn-q-n+2)A_2^{\perp}-6A_3^\perp\right],
       \end{array}\right.
\end{eqnarray*}
where $n=q^m$. Solving the above system of linear equations gives
\begin{eqnarray*}
A_1^{\perp}=A_2^{\perp}=0,\ A_3^{\perp}=\frac{(q^m-1)(q-1)(q-2)(q^m+3)}{24}>0.
\end{eqnarray*}
Then $d(\overline{\widehat{\cC_3}}^{\perp})=3$ and $\overline{\widehat{\cC_3}}^{\perp}$ has parameters $[q^m, q^m-2m-1, 3]$. Besides, it is easy to verify that $\overline{\widehat{\cC_3}}^{\perp}$ is at least almost optimal according to the sphere-packing bound if $q > 3$ as the $q$-ary $[q^m, q^m-2m-1, 5]$ code does not exist.

 It is easy to deduce that $\mathbf{1} \in \overline{\widehat{\cC_3}}$ and $\overline{\widehat{\cC_3}}$ is a $p$-divisible code and then $\overline{\widehat{\cC_3}}$ is also a self-orthogonal code by Theorem \ref{th-selforthogonal}.
\end{proof}

\subsection{The fourth family of self-orthogonal codes from irreducible cyclic codes}

In this subsection, let $q=p^e$, where $p$ is an odd prime. Let $m$ and $s$ be two positive integers with $m=2s$ and $s \geq 2$. Let $N=p^e+1$, $\gf_{q^m}^*=\langle\alpha\rangle$ and $\theta=\alpha^N$. Let $h(x)$ denote the minimal polynomials of $\theta^{-1}$ over $\gf_q$. Define a family of
cyclic codes over $\gf_q$ as
\begin{eqnarray}\label{eq-C4}
  \cC_4 &=& \left\{ \bc_{(a,b)}=(\tr_{q^m/q}\left(a\theta^{t})\right)_{t=0}^{\frac{q^m-1}{N}-1}: a \in \gf_{q^m} \right\}
\end{eqnarray}
with parity-check polynomial $h(x)$. The complete weight distribution of $\cC_4$ in Equation (\ref{eq-C4}) was studied in \cite{Li2015DCC}.

Let $\overline{\widehat{\cC_4}}$ be the augmented code of extended code of cyclic code $\cC_4$. By the definitions of the extended and augmented codes of linear codes, we have
\begin{eqnarray}\label{eq-C4bar}
  \overline{\widehat{\cC_4}} &=& \Big\{\bc_{(a,b,c)}=\big(\tr_{q^m/q}(a)+c, \tr_{q^m/q}(a\theta)+c, \cdots, \tr_{q^m/q}(a\theta^{\frac{q^m-1}{N}-1})+c, c\big): \nonumber\\&& a \in \gf_{q^m},c \in \gf_q \Big\}.
\end{eqnarray}
In this subsection, we will determine the weight distribution of $\overline{\widehat{\cC_4}}$ and prove that $\overline{\widehat{\cC_4}}$ is self-orthogonal.

In order to determine the weight distribution of $\overline{\widehat{\cC_4}}$, we recall some lemmas as follows.

\begin{lemma}\cite{Li2015DCC}\label{lem-Nc4odd}
 Let $q=p^e$ with odd prime $p$ and a positive integer $e$. Let $m=2s$ with odd integer $s \geq 3$. Let $N(c)=| \{ 0 \leq t \leq \frac{q^m-1}{N}-1: \tr_{q^m/q}(a\theta^{t})=c\} |$, where $c \in \gf_q$. Then
 \begin{equation*}
   N(0)=\begin{cases}
          \frac{q^m-1}{N} & \mbox{ $1$ time, }\\
          \frac{q^{m-1}-1+(q-1)(N-1)q^{\frac{m-2}{2}}}{N} & \mbox{$\frac{q^m-1}{N}$ times,}\\
          \frac{q^{m-1}-1-(q-1)q^{\frac{m-2}{2}}}{N} & \mbox{$\frac{(q^m-1)(N-1)}{N}$ times,}
        \end{cases}
 \end{equation*}
for $c=0$ and
 \begin{equation*}
   N(c)=\begin{cases}
          0 & \mbox{ $1$ time, }\\
          \frac{q^{m-1}-(N-1)q^{\frac{m-2}{2}}}{N} & \mbox{$\frac{q^m-1}{N}$ times,}\\
          \frac{q^{m-1}+q^{\frac{m-2}{2}}}{N} & \mbox{$\frac{(q^m-1)(N-1)}{N}$ times,}
        \end{cases}
 \end{equation*}
for $c \in \gf_q^*$.
\end{lemma}

\begin{lemma}\cite{Li2015DCC}\label{lem-Nc4even}
 Let $q=p^e$ with odd prime $p$ and a positive integer $e$. Let $m=2s$ with even integer $s \geq 2$. Let $N(c)=|\{ 0 \leq t \leq \frac{q^m-1}{N}-1: \tr_{q^m/q}(a\theta^{t})=c\} |$, where $c \in \gf_q$. Then
 \begin{equation*}
   N(0)=\begin{cases}
          \frac{q^m-1}{N} & \mbox{ $1$ time, }\\
          \frac{q^{m-1}-1-(q-1)(N-1)q^{\frac{m-2}{2}}}{N} & \mbox{$\frac{q^m-1}{N}$ times,}\\
          \frac{q^{m-1}-1+(q-1)q^{\frac{m-2}{2}}}{N} & \mbox{$\frac{(q^m-1)(N-1)}{N}$ times,}
        \end{cases}
 \end{equation*}
for $c=0$ and
 \begin{equation*}
   N(c)=\begin{cases}
          0 & \mbox{ $1$ time, }\\
          \frac{q^{m-1}+(N-1)q^{\frac{m-2}{2}}}{N} & \mbox{$\frac{q^m-1}{N}$ times,}\\
          \frac{q^{m-1}-q^{\frac{m-2}{2}}}{N} & \mbox{$\frac{(q^m-1)(N-1)}{N}$ times,}
        \end{cases}
 \end{equation*}
for $c \in \gf_q^*$.
\end{lemma}


\begin{theorem}\label{th-C4odd}
 Let $q=p^e$ with odd prime $p$ and a positive integer $e$. Let $m=2s$ with odd integer $s \geq 3$ and $N=p^e+1$. Then the code $\overline{\widehat{\cC_4}}$ defined in Equation (\ref{eq-C4bar}) is a self-orthogonal $p$-divisible code over $\gf_q$ with parameters $[\frac{q^m-1}{N}+1, m+1, \frac{(q-1)q^{m-1}-(q-1)(N-1)q^{\frac{m-2}{2}}}{N}]$ and its weight distribution is listed in Table \ref{tab4.15}. Besides, $\overline{\widehat{\cC_4}}^{\perp}$ has parameters $[\frac{q^m-1}{N}+1,\frac{q^m-1}{N}-m,3]$ and $\overline{\widehat{\cC_4}}^{\perp}$ is at least almost optimal according to the sphere-packing bound.
 \begin{table}[h]
\begin{center}
\caption{The weight distribution of $\overline{\widehat{\cC_4}}$ in Theorem \ref{th-C4odd}.}\label{tab4.15}
\begin{tabular}{@{}ll@{}}
\toprule%
Weight & Frequency  \\
\midrule
$0$ & $1$\\
$\frac{(q-1)q^{m-1}-(q-1)(N-1)q^{\frac{m-2}{2}}}{N}$ & $\frac{q^m-1}{N}$\\
$\frac{(q-1)q^{m-1}+(q-1)q^{\frac{m-2}{2}}}{N}$ & $\frac{(q^m-1)(N-1)}{N}$\\
$\frac{q^m-q^{m-1}+(N-1)(q^{\frac{m-2}{2}}+1)}{N}$ & $\frac{(q^m-1)(q-1)}{N}$\\
$\frac{q^m-q^{m-1}-q^{\frac{m-2}{2}}+N-1}{N}$ & $\frac{(q^m-1)(q-1)(N-1)}{N}$\\
$\frac{q^m-1}{N}+1$ & $q-1$\\
\bottomrule
\end{tabular}
\end{center}
\end{table}
\end{theorem}

\begin{proof}
It is obvious that the length of $\overline{\widehat{\cC_4}}$ is $n=\frac{q^m-1}{N}+1$.
   Choose any codeword $\bc_{(a,b,c)} \in \overline{\widehat{\cC_4}}$. For $c = 0$, we have $\text{wt}(\bc_{(a,b,0)})=n-1-N(0)=\frac{q^m-1}{N}-N(0)$, where $N(0)$ is defined in Lemma \ref{lem-Nc4odd}. Then by Lemma \ref{lem-Nc4odd}, we derive
 \begin{eqnarray*}
   \text{wt}(\bc_{(a,b,0)}) &=& \begin{cases}
                                  0 & \mbox{$1$ time,} \\
                                 \frac{(q-1)q^{m-1}-(q-1)(N-1)q^{\frac{m-2}{2}}}{N} & \mbox{$\frac{q^m-1}{N}$ times,} \\
                                 \frac{(q-1)q^{m-1}+(q-1)q^{\frac{m-2}{2}}}{N} & \mbox{$\frac{(q^m-1)(N-1)}{N}$ times.}
                                \end{cases}
 \end{eqnarray*}
 For $c \neq 0$, we have $\text{wt}(\bc_{(a,b,c)})=n-N(-c)=p^m-N(-c)$, where $N(-c)$ is defined in Lemma \ref{lem-Nc4odd}. By Lemma \ref{lem-Nc4odd}, we deduce that
 \begin{eqnarray*}
   \text{wt}(\bc_{(a,b,c)}) &=& \begin{cases}
                                  \frac{q^m-1}{N}+1 & \mbox{$q-1$ times,} \\
                                 \frac{q^m-q^{m-1}+(N-1)(q^{\frac{m-2}{2}}+1)}{N} & \mbox{$\frac{(q^m-1)(q-1)}{N}$ times,}\\
                                  \frac{q^m-q^{m-1}-q^{\frac{m-2}{2}}+N-1}{N} & \mbox{$\frac{(q^m-1)(q-1)(N-1)}{N}$ times.}
                                \end{cases}
 \end{eqnarray*}
Furthermore, the dimension of $\overline{\widehat{\cC_4}}$ is $m+1$ as the zero codeword in $\overline{\widehat{\cC_4}}$ occurs only once. Then the parameters and weight distribution of $\overline{\widehat{\cC_4}}$ directly follow.

 In the following, we will determine the minimum distance of $\overline{\widehat{\cC_4}}^{\perp}$. Denote by $$w_1=\frac{(q-1)q^{m-1}-(q-1)(N-1)q^{\frac{m-2}{2}}}{N}, w_2=\frac{(q-1)q^{m-1}+(q-1)q^{\frac{m-2}{2}}}{N},$$ $$w_3=\frac{q^m-q^{m-1}+(N-1)(q^{\frac{m-2}{2}}+1)}{N}, w_4=\frac{q^m-q^{m-1}-q^{\frac{m-2}{2}}+N-1}{N} \text{ and } w_5=\frac{q^m-1}{N}+1.$$ By the second, third and fourth Pless power moments in \cite[Page 260]{H}, we have
 \begin{eqnarray*}
 \left\{\begin{array}{l}
         \sum_{i=1}^{5}w_iA_{w_i}=q^{m}(qn-n-A_1^{\perp}), \\ \sum_{i=1}^{5}w_i^2A_{w_i}=q^{m-1}\left((q-1)n(qn-n+1)-(2qn-q-2n+2)A_1^{\perp}+2A_2^{\perp}\right),\\ \sum_{i=1}^{5}w_i^3A_{w_i}=q^{m-2}\left[(q-1)n(q^2n^2-2qn^2+3qn-q+n^2-3n+2)\right.\\
         \qquad\qquad\qquad\quad\left.-(3q^2n^2-3q^2n-6qn^2+12qn+q^2-6q+3n^2-9n+6)A_1^{\perp}\right.\\
         \qquad\qquad\qquad\quad\left. +6(qn-q-n+2)A_2^{\perp}-6A_3^\perp\right],
       \end{array}\right.
\end{eqnarray*}
where $n=\frac{q^m-1}{N}+1$ and $A_{w_i} (1 \leq i \leq 5)$  denote the frequency of the weight $w_i$ in Table \ref{tab4.15}. Solving the above system of linear equations gives
\begin{eqnarray*}
A_1^{\perp}=A_2^{\perp}=0,\ A_3^{\perp}=\frac{(q^m-1)(q^2-3q+2)(q^m+q^{\frac{m+4}{2}}+2q^2+q-q^{\frac{m+2}{2}})}{6(q+1)^3}.
\end{eqnarray*}
Then $d(\overline{\widehat{\cC_4}}^{\perp})=3$ and $\overline{\widehat{\cC_4}}^{\perp}$ has parameters $[\frac{q^m-1}{N}+1, \frac{q^m-1}{N}-m, 3]$. Besides, it is easy to verify that $\overline{\widehat{\cC_4}}^{\perp}$ is at least almost optimal according to the sphere-packing bound as the $q$-ary $[\frac{q^m-1}{N}+1, \frac{q^m-1}{N}-m, 5]$ code does not exist.

 It is easy to deduce that $\mathbf{1} \in \overline{\widehat{\cC_4}}$ and $\overline{\widehat{\cC_4}}$ is a $p$-divisible code. Then $\overline{\widehat{\cC_4}}$ is a self-orthogonal code by Theorem \ref{th-selforthogonal}.
\end{proof}

\begin{theorem}\label{th-C4even}
 Let $q=p^e$ with odd prime $p$ and a positive integer $e$. Let $m=2s$ with even integer $s \geq 2$ and $N=p^e+1$. Then the code $\overline{\widehat{\cC_4}}$ defined in Equation (\ref{eq-C4bar}) is a self-orthogonal $p$-divisible code with parameters $[\frac{q^m-1}{N}+1, m+1, \frac{q^m-q^{m-1}-(N-1)(q^{\frac{m-2}{2}}-1)}{N}]$ and its weight distribution is listed in Table \ref{tab4.16}. Besides, $\overline{\widehat{\cC_4}}^{\perp}$ has parameters $[\frac{q^m-1}{N}+1,\frac{q^m-1}{N}-m,3]$ and $\overline{\widehat{\cC_4}}^{\perp}$ is at least almost optimal according to the sphere-packing bound.
 \begin{table}[h]
\begin{center}
\caption{The weight distribution of $\overline{\widehat{\cC_4}}$ in Theorem \ref{th-C4even}.}\label{tab4.16}
\begin{tabular}{@{}ll@{}}
\toprule%
Weight & Frequency  \\
\midrule
$0$ & $1$\\
$\frac{(q-1)q^{m-1}+(q-1)(N-1)q^{\frac{m-2}{2}}}{N}$ & $\frac{q^m-1}{N}$\\
$\frac{(q-1)q^{m-1}-(q-1)q^{\frac{m-2}{2}}}{N}$ & $\frac{(q^m-1)(N-1)}{N}$\\
$\frac{q^m-q^{m-1}-(N-1)(q^{\frac{m-2}{2}}-1)}{N}$ & $\frac{(q^m-1)(q-1)}{N}$\\
$\frac{q^m-q^{m-1}+q^{\frac{m-2}{2}}+N-1}{N}$ & $\frac{(q^m-1)(q-1)(N-1)}{N}$\\
$\frac{q^m-1}{N}+1$ & $q-1$\\
\bottomrule
\end{tabular}
\end{center}
\end{table}
\end{theorem}

\begin{proof}
It is obvious that the length of $\overline{\widehat{\cC_4}}$ is $n=\frac{q^m-1}{N}+1$.
   Choose any codeword $\bc_{(a,b,c)} \in \overline{\widehat{\cC_4}}$. For $c = 0$, we have $\text{wt}(\bc_{(a,b,0)})=n-1-N(0)=\frac{q^m-1}{N}-N(0)$, where $N(0)$ is defined in Lemma \ref{lem-Nc4even}. Then by Lemma \ref{lem-Nc4even}, we derive
 \begin{eqnarray*}
   \text{wt}(\bc_{(a,b,0)}) &=& \begin{cases}
                                  0 & \mbox{$1$ time,} \\
                                 \frac{(q-1)q^{m-1}+(q-1)(N-1)q^{\frac{m-2}{2}}}{N} & \mbox{$\frac{q^m-1}{N}$ times,} \\
                                 \frac{(q-1)q^{m-1}-(q-1)q^{\frac{m-2}{2}}}{N} & \mbox{$\frac{(q^m-1)(N-1)}{N}$ times.}
                                \end{cases}
 \end{eqnarray*}
 For $c \neq 0$, we have $\text{wt}(\bc_{(a,b,c)})=n-N(-c)=p^m-N(-c)$, where $N(-c)$ is defined in Lemma \ref{lem-Nc4odd}. By Lemma \ref{lem-Nc4odd}, we deduce that
 \begin{eqnarray*}
   \text{wt}(\bc_{(a,b,c)}) &=& \begin{cases}
                                  \frac{q^m-1}{N}+1 & \mbox{$q-1$ times,} \\
                                 \frac{q^m-q^{m-1}-(N-1)(q^{\frac{m-2}{2}}-1)}{N} & \mbox{$\frac{(q^m-1)(q-1)}{N}$ times,}\\
                                  \frac{q^m-q^{m-1}+q^{\frac{m-2}{2}}+N-1}{N} & \mbox{$\frac{(q^m-1)(q-1)(N-1)}{N}$ times.}
                                \end{cases}
 \end{eqnarray*}
Furthermore, the dimension of $\overline{\widehat{\cC_4}}$ is $m+1$ as the zero codeword in $\overline{\widehat{\cC_4}}$ occurs only once. Then the parameters and weight distribution of $\overline{\widehat{\cC_4}}$ directly follow.

 In the following, we will determine the minimum distance of $\overline{\widehat{\cC_4}}^{\perp}$. Denote by $$w_1=\frac{(q-1)q^{m-1}+(q-1)(N-1)q^{\frac{m-2}{2}}}{N}, w_2=\frac{(q-1)q^{m-1}-(q-1)q^{\frac{m-2}{2}}}{N},$$ $$w_3=\frac{q^m-q^{m-1}-(N-1)(q^{\frac{m-2}{2}}-1)}{N}, w_4=\frac{q^m-q^{m-1}+q^{\frac{m-2}{2}}+N-1}{N} \text{ and } w_5=\frac{q^m-1}{N}+1.$$ By the second, third and fourth Pless power moments in \cite[Page 260]{H}, we have
 \begin{eqnarray*}
 \left\{\begin{array}{l}
         \sum_{i=1}^{5}w_iA_{w_i}=q^{m}(qn-n-A_1^{\perp}), \\ \sum_{i=1}^{5}w_i^2A_{w_i}=q^{m-1}\left((q-1)n(qn-n+1)-(2qn-q-2n+2)A_1^{\perp}+2A_2^{\perp}\right),\\ \sum_{i=1}^{5}w_i^3A_{w_i}=q^{m-2}\left[(q-1)n(q^2n^2-2qn^2+3qn-q+n^2-3n+2)\right.\\
         \qquad\qquad\qquad\quad\left.-(3q^2n^2-3q^2n-6qn^2+12qn+q^2-6q+3n^2-9n+6)A_1^{\perp}\right.\\
         \qquad\qquad\qquad\quad\left. +6(qn-q-n+2)A_2^{\perp}-6A_3^\perp\right],
       \end{array}\right.
\end{eqnarray*}
where $n=\frac{q^m-1}{N}+1$ and $A_{w_i} (1 \leq i \leq 5)$ denots the frequency of the weight $w_i$ in Table \ref{tab4.16}. Solving the above system of linear equations gives
\begin{eqnarray*}
A_1^{\perp}=A_2^{\perp}=0,\ A_3^{\perp}=\frac{(q^m-1)(q^2-3q+2)(q^m+q^{\frac{m+2}{2}}+2q^2+q-q^{\frac{m+4}{2}})}{6(q+1)^3}.
\end{eqnarray*}
Then $d(\overline{\widehat{\cC_4}}^{\perp})=3$ and $\overline{\widehat{\cC_4}}^{\perp}$ has parameters $[\frac{q^m-1}{N}+1, \frac{q^m-1}{N}-m, 3]$. Besides, it is easy to verify that $\overline{\widehat{\cC_4}}^{\perp}$ is at least almost optimal according to the sphere-packing bound.

 It is easy to deduce that $\mathbf{1} \in \overline{\widehat{\cC_4}}$ and $\overline{\widehat{\cC_4}}$ is a $p$-divisible code. Then $\overline{\widehat{\cC_4}}$ is also a self-orthogonal code by Theorem \ref{th-selforthogonal}.
\end{proof}

\subsection{The fifth family of self-orthogonal codes from the BCH codes with designed distance $\delta_2=(q-1)q^{m-1}-1-q^{\lfloor(m-1)/2\rfloor}$}
In this subsection, let $q$ be an odd prime. Let $m$ be a positive integer with $m \geq 3$. Define $\tilde{g}_{(q,m,\delta_2)}(x)=(x-1)g_{(q,m,\delta_2)}(x)$ with $\delta_2=(q-1)q^{m-1}-1-q^{\lfloor(m-1)/2\rfloor}$. Let $\tilde{\cC}_{(q,m,\delta_2)}$ and $\cC_{(q,m,\delta_2)}$ respectively be the cyclic codes over $\gf_q$ with generator polynomials $\tilde{g}_{(q,m,\delta_2)}(x)$ and $g_{(q,m,\delta_2)}(x)$. By \cite{DFZ}, $\tilde{\cC}_{(q,m,\delta_2)}$ can be written as
\begin{eqnarray*}
  \tilde{\cC}_{(q,m,\delta_2)} &=& \left\{\left(\tr_{q^m/q}(ax^{1+q^{\lfloor(m-1)/2\rfloor+1}}+bx)\right)_{x \in \gf_{q^m}^*}: a \in \gf_{q^m}, b \in \gf_{q^m}\right\},
\end{eqnarray*}
and $\cC_{(q,m,\delta_2)}$ can be written as
\begin{eqnarray*}
  \cC_{(q,m,\delta_2)} &=& \left\{\left(\tr_{q^m/q}(ax^{1+q^{\lfloor(m-1)/2\rfloor+1}}+bx)+c\right)_{x \in \gf_{q^m}^*}: a \in \gf_{q^m}, b \in \gf_{q^m}, c \in \gf_q\right\}.
\end{eqnarray*}
Note that $\cC_{(q,m,\delta_2)}$ is the augmented code of $\tilde{\cC}_{(q,m,\delta_2)}$.

Ding et al. used character sums of quadratic functions to determine the weight distribution of $\tilde{\cC}_{(q,m,\delta_2)}$ in \cite{DFZ}. Subsequently, Li et al. presented the weight distributions of $\cC_{(q,m,\delta_2)}$ and $\overline{\cC}_{(q,m,\delta_2)}$ in \cite{LWL}, where $\overline{\cC}_{(q,m,\delta_2)}$ is the extended code of $\cC_{(q,m,\delta_2)}$.

 The objective of this subsection is to prove that $\overline{\cC}_{(q,m,\delta_2)}$ is self-orthogonal and determine the dual distance of $\overline{\cC}_{(q,m,\delta_2)}$. Besides, we will prove that $\overline{\cC}_{(q,m,\delta_2)}$ holds a $1$-design or a $2$-design and determine the locality of $\overline{\cC}_{(q,m,\delta_2)}$.
 In the following, we recall the weight distributions of $\overline{\cC}_{(q,m,\delta_2)}$ for odd $m$ and even $m$, respectively.

\begin{lemma}\cite{LWL}\label{lem-D1odd}
  Let $q$ be an odd prime and $m$ be an odd integer with $m \geq 3$. Then the extended code $\overline{\cC}_{(q,m,\delta_2)}$ has parameters $[q^m, 2m+1, (q-1)q^{m-1}-q^{\frac{m-1}{2}}]$ and its weight distribution is listed in Table \ref{tab6.3}.
\end{lemma}
 \begin{table}[h]
\begin{center}
\caption{The weight distribution of $\overline{\cC}_{(q,m,\delta_2)}$ in Lemma \ref{lem-D1odd}.}\label{tab6.3}
\begin{tabular}{@{}ll@{}}
\toprule%
Weight & Frequency  \\
\midrule
$0$ & $1$\\
$(q-1)q^{m-1}-q^{\frac{m-1}{2}}$ & $\frac{q^m(q^m-1)(q-1)}{2}$\\
$(q-1)q^{m-1}$ & $(q^m+q)(q^m-1)$\\
$(q-1)q^{m-1}+q^{\frac{m-1}{2}}$ & $\frac{q^m(q^m-1)(q-1)}{2}$\\
$q^m$ & $q-1$\\
\bottomrule
\end{tabular}
\end{center}
\end{table}

\begin{lemma}\cite{LWL}\label{lem-D1even}
  Let $q$ be an odd prime and $m$ be an even integer with $m \geq 4$. Then the extended code $\overline{\cC}_{(q,m,\delta_2)}$ has parameters $[q^m, \frac{3m}{2}+1, (q-1)q^{m-1}-q^{\frac{m-2}{2}}]$ and its weight distribution is listed in Table \ref{tab6.4}.
\end{lemma}
 \begin{table}[h]
\begin{center}
\caption{The weight distribution of $\overline{\cC}_{(q,m,\delta_2)}$ in Lemma \ref{lem-D1even}.}\label{tab6.4}
\begin{tabular}{@{}ll@{}}
\toprule%
Weight & Frequency  \\
\midrule
$0$ & $1$\\
$(q-1)q^{m-1}-q^{\frac{m-2}{2}}$ & $q^m(q^{\frac{m}{2}}-1)(q-1)$\\
$(q-1)q^{m-1}$ & $q(q^m-1)$\\
$(q-1)(q^{m-1}+q^{\frac{m-2}{2}})$ & $q^{\frac{3m}{2}}-q^m$\\
$q^m$ & $q-1$\\
\bottomrule
\end{tabular}
\end{center}
\end{table}

To give our first result in this subsection, we denote by
\begin{eqnarray*}
  U(i,j) &=& u\binom{(q-1)q^m-q^{\frac{m-1}{2}}}{i}\binom{n-(q-1)q^m+q^{\frac{m-1}{2}}}{j}+v\binom{(q-1)q^{m-1}}{i}\binom{n-(q-1)q^{m-1}}{j}\\
  &&+w\binom{(q-1)q^{m-1}+q^{\frac{m-1}{2}}}{i}\binom{n-(q-1)q^{m-1}-q^{\frac{m-1}{2}}}{j},
\end{eqnarray*}
where $n=q^m$, $u=\frac{q^m(q^m-1)(q-1)}{2}$, $v=(q^m+q)(q^m-1)$ and $w=\frac{q^m(q^m-1)(q-1)}{2}$.

\begin{theorem}\label{th-C5odd}
  Let $q$ be an odd prime and $m$ be an odd integer with $m \geq 3$. Then the extended code $\overline{\cC}_{(q,m,\delta_2)}$ is a self-orthogonal $q$-divisible code and its dual code $\overline{\cC}_{(q,m,\delta_2)}^{\perp}$ has parameters $[q^m, q^m-2m-1, 5]$ for $q=3$ and $[q^m, q^m-2m-1, 4]$ for $q > 3$. The weight enumerator $A^{\perp}(z)$ of $\overline{\cC}_{(q,m,\delta_2)}^{\perp}$ is given by
  \begin{eqnarray*}
    q^{2m+1}A^{\perp}(z) &=& \sum_{l=0}^n\left[\binom{n}{l}(q-1)^l+\sum_{i+j=l}U(i,j)(-1)^i(q-1)^j+t\binom{n}{l}(-1)^l\right]z^l,
  \end{eqnarray*}
  where $n=q^m$ and $t=q-1$.
  In particular, the following conclusions hold:
  \begin{itemize}
  \item For $q=3$, $\overline{\cC}_{(q,m,\delta_2)}^{\perp}$ is at least almost optimal according to the sphere-packing bound. For $q>3$, $\overline{\cC}_{(q,m,\delta_2)}^{\perp}$ is optimal according to the sphere-packing bound.
    \item For $q=3$, the minimum weight codewords in $\overline{\cC}_{(q,m,\delta_2)}$ hold a $$2-\left(3^m, 2\cdot 3^{m-1}-3^{\frac{m-1}{2}}, \frac{(2\cdot 3^{m-1}-3^{\frac{m-1}{2}})(2 \cdot 3^{m-1}-3^{\frac{m-1}{2}}-1)}{2}\right)$$  design and the minimum weight codewords in $\overline{\cC}_{(q,m,\delta_2)}^{\perp}$ hold a $$2-\left(3^m, 5, \frac{5(3^{m-1}-1)}{2}\right)$$ design.
    \item For $q > 3$, the minimum weight codewords in $\overline{\cC}_{(q,m,\delta_2)}$ hold a $1$-design and the minimum weight codewords in $\overline{\cC}_{(q,m,\delta_2)}^{\perp}$ also hold a $1$-design.
    \item The extended code $\overline{\cC}_{(q,m,\delta_2)}$ has locality $4$ for $q=3$ and has locality $3$ for $q > 3$.
  \end{itemize}
\end{theorem}

\begin{proof}
Firstly, by Lemma \ref{lem-D1odd}, we derive that $\overline{\cC}_{(q,m,\delta_2)}$ is $q$-divisible and $\mathbf{1} \in \overline{\cC}_{(q,m,\delta_2)}$. Then by Theorem \ref{th-selforthogonal}, we deduce that $\overline{\cC}_{(q,m,\delta_2)}$ is self-orthogonal.

Then we determine the weight enumerator $A^{\perp}(z)$ and minimum distance of $\overline{\cC}_{(q,m,\delta_2)}^{\perp}$.
  Let $n=q^m$. By Lemma \ref{lem-D1odd}, we derive that the weight enumerator of $\overline{\cC}_{(q,m,\delta_2)}$ is
  \begin{eqnarray*}
    A(z) &=& 1+uz^{(q-1)q^{m-1}-q^{\frac{m-1}{2}}}+vz^{(q-1)q^{m-1}}+wz^{(q-1)q^{m-1}+q^{\frac{m-1}{2}}}+tz^{q^m},
  \end{eqnarray*}
  where $u=\frac{q^m(q^m-1)(q-1)}{2}$, $v=(q^m+q)(q^m-1)$, $w=\frac{q^m(q^m-1)(q-1)}{2}$ and $t=q-1$. By Lemma \ref{lem-Mac}, the weight enumerator of $\overline{\cC}_{(q,m,\delta_2)}$ satisfies
  \begin{eqnarray}\label{eq-A}
    q^{2m+1}A^{\perp}(z) &=& (1+(q-1)z)^nA\left(\frac{1-z}{1+(q-1)z}\right)\nonumber\\
    &=&(1+(q-1)z)^n+u(1-z)^{(q-1)q^{m-1}-q^{\frac{m-1}{2}}}(1+(q-1)z)^{n-(q-1)q^{m-1}+q^{\frac{m-1}{2}}}\nonumber\\
    && +v(1-z)^{(q-1)q^{m-1}}(1+(q-1)z)^{n-(q-1)q^{m-1}}\nonumber\\
    && +w(1-z)^{(q-1)q^{m-1}+q^{\frac{m-1}{2}}}(1+(q-1)z)^{n-(q-1)q^{m-1}-q^{\frac{m-1}{2}}}+t(1-z)^n.
  \end{eqnarray}
Furthermore, we have
\begin{eqnarray}\label{eq-A1}
  (1+(q-1)z)^n &=& \sum_{l=0}^n\binom{n}{l}(q-1)^lz^l,
\end{eqnarray}
\begin{eqnarray}\label{eq-A2}
  &&u(1-z)^{(q-1)q^{m-1}-q^{\frac{m-1}{2}}}(1+(q-1)z)^{n-(q-1)q^{m-1}+q^{\frac{m-1}{2}}}\nonumber \\
  &=& u\sum_{l=0}^n\left(\sum_{i+j=l}\binom{(q-1)q^{m-1}-q^{\frac{m-1}{2}}}{i}\binom{n-(q-1)q^{m-1}+q^{\frac{m-1}{2}}}{j}(-1)^i(q-1)^j\right)z^l,
\end{eqnarray}
\begin{eqnarray}\label{eq-A3}
  &&v(1-z)^{(q-1)q^{m-1}}(1+(q-1)z)^{n-(q-1)q^{m-1}}\nonumber\\
  &=&v\sum_{l=0}^n\left(\sum_{i+j=l}\binom{(q-1)q^{m-1}}{i}\binom{n-(q-1)q^{m-1}}{j}(-1)^i(q-1)^j\right)z^l,
\end{eqnarray}
\begin{eqnarray}\label{eq-A4}
&&w(1-z)^{(q-1)q^{m-1}+q^{\frac{m-1}{2}}}(1+(q-1)z)^{n-(q-1)q^{m-1}-q^{\frac{m-1}{2}}}\nonumber \\
  &=& w\sum_{l=0}^n\left(\sum_{i+j=l}\binom{(q-1)q^{m-1}+q^{\frac{m-1}{2}}}{i}\binom{n-(q-1)q^{m-1}-q^{\frac{m-1}{2}}}{j}(-1)^i(q-1)^j\right)z^l,
\end{eqnarray}
\begin{eqnarray}\label{eq-A5}
  t(1-z)^n &=& t\sum_{l=0}^{n}\binom{n}{l}(-1)^lz^l.
\end{eqnarray}
Substituting Equations (\ref{eq-A1}), (\ref{eq-A2}), (\ref{eq-A3}), (\ref{eq-A4}) and (\ref{eq-A5}) into (\ref{eq-A}) yields
\begin{eqnarray*}
    q^{2m+1}A^{\perp}(z) &=& \sum_{l=0}^n\left[\binom{n}{l}(q-1)^l+\sum_{i+j=l}U(i,j)(-1)^i(q-1)^j+t\binom{n}{l}(-1)^l\right]z^l.
  \end{eqnarray*}
Then we have
\begin{eqnarray*}
  q^{2m+1}A_1^{\perp} &=& \binom{n}{1}(q-1)+\sum_{i+j=1}U(i,j)(-1)^i(q-1)^j-(q-1)\binom{n}{1}\\
  &=& -U(1,0)+U(0,1)(q-1)=0,
\end{eqnarray*}
\begin{eqnarray*}
  q^{2m+1}A_2^{\perp} &=& \binom{n}{2}(q-1)^2+\sum_{i+j=2}U(i,j)(-1)^i(q-1)^j+(q-1)\binom{n}{2}\\
  &=& \frac{q^{m+1}(q^m-1)(q-1)}{2}+U(2,0)-(q-1)U(1,1)+(q-1)^2U(0,2)\\
  &=&0,
\end{eqnarray*}
\begin{eqnarray*}
  q^{2m+1}A_3^{\perp} &=& \binom{n}{3}(q-1)^3+\sum_{i+j=3}U(i,j)(-1)^i(q-1)^j-(q-1)\binom{n}{3}\\
  &=& \frac{q^{m+1}(q^m-1)(q^m-2)(q-2)(q-1)}{6}+(q-1)^3U(0,3)\\&&-(q-1)^2U(1,2)+(q-1)U(2,1)-U(3,0)\\
  &=&0,
\end{eqnarray*}
\begin{eqnarray*}
  q^{2m+1}A_4^{\perp} &=& \binom{n}{4}(q-1)^4+\sum_{i+j=4}U(i,j)(-1)^i(q-1)^j+(q-1)\binom{n}{4}\\
  &=& \frac{q^{m+1}(q^m-1)(q^m-2)(q^m-3)(q^2-3q+3)(q-1)}{24}+(q-1)^4U(0,4)\\&&-(q-1)^3U(1,3)+(q-1)^2U(2,2)-(q-1)U(3,1)+U(4,0)\\
  &=&\frac{q^{3m+1}(q-3)(q-1)(q-2)(q^m-1)}{24},
\end{eqnarray*}
which implies that $A_4^{\perp}=\frac{q^{m}(q-3)(q-1)(q-2)(q^m-1)}{24}$. If $q > 3$, then $A_1^{\perp}=A_2^{\perp}=A_3^{\perp}=0$ and $A_4^{\perp} > 0$. Hence, the minimum distance of $\overline{\cC}_{(q,m,\delta_2)}^{\perp}$ is $4$. If $q=3$, then $A_1^{\perp}=A_2^{\perp}=A_3^{\perp}=A_4^{\perp}=0$. Now let $q=3$ and we have
\begin{eqnarray*}
  3^{2m+1}A_5^{\perp} &=& \binom{3^m}{5}2^5+\sum_{i+j=5}U(i,j)(-1)^i2^j-2\binom{3^m}{5}\\
  &=& \frac{3^{m}(3^m-1)(3^m-2)(3^m-3)(3^m-4)}{4}+2^5U(0,5)-2^4U(1,4)\\&&+2^3T(2,3)-4T(3,2)+2T(4,1)-T(5,0)\\
  &=&\frac{3^{3m+1}(3^m-1)(3^{m-1}-1)}{4},
\end{eqnarray*}
which implies that $A_5^{\perp}=\frac{3^{m}(3^m-1)(3^{m-1}-1)}{4} > 0$. Then the minimum distance of $\overline{\cC}_{(q,m,\delta_2)}^{\perp}$ is $5$. Furthermore, it is easy to verify that $\overline{\cC}_{(q,m,\delta_2)}^{\perp}$ is at least almost optimal according to the sphere-packing bound for $q=3$ and $\overline{\cC}_{(q,m,\delta_2)}^{\perp}$ is optimal according to the sphere-packing bound for $q > 3$.

Finally, we prove that the locality of $\overline{\cC}_{(q,m,\delta_2)}$ is $4$ or $3$. By Lemma \ref{lem-AMth} and Equation (\ref{eqn-t}), we deduce that the minimum weight codewords in $\overline{\cC}_{(q,m,\delta_2)}$ hold a $$2-\left(3^m, 2 \cdot 3^{m-1}-3^{\frac{m-1}{2}}, \frac{(2 \cdot 3^{m-1}-3^{\frac{m-1}{2}})(2\cdot 3^{m-1}-3^{\frac{m-1}{2}}-1)}{2}\right)$$ design and the minimum weight codewords in $\overline{\cC}_{(q,m,\delta_2)}^{\perp}$ hold a $$2-\left(3^m, 5, \frac{5(3^{m-1}-1)}{2}\right)$$ design when $q=3$. We also deduce that the minimum weight codewords in $\overline{\cC}_{(q,m,\delta_2)}$ hold a $1$-design and the minimum weight codewords in $\overline{\cC}_{(q,m,\delta_2)}^{\perp}$ hold a $1$-design when $q>3$. Then by Lemma \ref{lem-locality}, the locality of $\overline{\cC}_{(q,m,\delta_2)}$ is $d(\overline{\cC}_{(q,m,\delta_2)}^{\perp})-1$. Then the desired conclusions follow.
\end{proof}

To give the second result in this subsection, we denote by
\begin{eqnarray*}
  V(i,j) &=& u\binom{(q-1)q^{m-1}-q^{\frac{m-2}{2}}}{i}\binom{n-(q-1)q^{m-1}+q^{\frac{m-2}{2}}}{j}\\
  &&+v\binom{(q-1)q^{m-1}}{i}\binom{n-(q-1)q^{m-1}}{j}\\
  &&+w\binom{(q-1)(q^{m-1}+q^{\frac{m-2}{2}})}{i}\binom{n-(q-1)(q^{m-1}+q^{\frac{m-2}{2}})}{j},
\end{eqnarray*}
where $n=q^m$, $u=q^m(q^{\frac{m}{2}}-1)(q-1)$, $v=q(q^m-1)$ and $w=q^{\frac{3m}{2}}-q^m$.
\begin{theorem}\label{th-C5even}
  Let $q$ be an odd prime and $m$ be an even integer with $m \geq 4$. Then the extended code $\overline{\cC}_{(q,m,\delta_2)}$ is a self-orthogonal $q$-divisible code and its dual code $\overline{\cC}_{(q,m,\delta_2)}^{\perp}$ has parameters $[q^m, q^m-\frac{3m}{2}-1, 4]$. Besides, $\overline{\cC}_{(q,m,\delta_2)}^{\perp}$ is optimal according to the sphere-packing bound. The weight enumerator $A^{\perp}(z)$ of $\overline{\cC}_{(q,m,\delta_2)}^{\perp}$ is given by
  \begin{eqnarray*}
    q^{\frac{3m}{2}+1}A^{\perp}(z) &=& \sum_{l=0}^n\left[\binom{n}{l}(q-1)^l+\sum_{i+j=l}V(i,j)(-1)^i(q-1)^j+t\binom{n}{l}(-1)^l\right]z^l,
  \end{eqnarray*}
  where $n=q^m$ and $t=q-1$.
  In particular, the minimum weight codewords in $\overline{\cC}_{(q,m,\delta_2)}$ hold a $1$-design and the minimum weight codewords in $\overline{\cC}_{(q,m,\delta_2)}^{\perp}$ also hold a $1$-design. Then $\overline{\cC}_{(q,m,\delta_2)}$ has locality $3$.
\end{theorem}

\begin{proof}
Firstly, by Lemma \ref{lem-D1even}, we derive that $\overline{\cC}_{(q,m,\delta_2)}$ is $q$-divisible and $\mathbf{1} \in \overline{\cC}_{(q,m,\delta_2)}$. Then by Theorem \ref{th-selforthogonal}, we deduce that $\overline{\cC}_{(q,m,\delta_2)}$ is self-orthogonal.

Then we determine the weight enumerator $A^{\perp}(z)$ and minimum distance of $\overline{\cC}_{(q,m,\delta_2)}^{\perp}$.
  Let $n=q^m$. By Lemma \ref{lem-D1even}, we derive that the weight enumerator of $\overline{\cC}_{(q,m,\delta_2)}$ is
  \begin{eqnarray*}
    A(z) &=& 1+uz^{(q-1)q^{m-1}-q^{\frac{m-2}{2}}}+vz^{(q-1)q^{m-1}}+wz^{(q-1)(q^{m-1}+q^{\frac{m-2}{2}})}+tz^{q^m},
  \end{eqnarray*}
  where $u=q^m(q^{\frac{m}{2}}-1)(q-1)$, $v=q(q^m-1)$, $w=q^{\frac{3m}{2}}-q^m$ and $t=q-1$. By Lemma \ref{lem-Mac}, the weight enumerator of $\overline{\cC}_{(q,m,\delta_2)}$ satisfies
  \begin{eqnarray}\label{eq-B}
    q^{\frac{3m}{2}+1}A^{\perp}(z) &=& (1+(q-1)z)^nA\left(\frac{1-z}{1+(q-1)z}\right)\nonumber\\
    &=&(1+(q-1)z)^n+u(1-z)^{(q-1)q^{m-1}-q^{\frac{m-2}{2}}}(1+(q-1)z)^{n-(q-1)q^{m-1}+q^{\frac{m-2}{2}}}\nonumber\\
    && +v(1-z)^{(q-1)q^{m-1}}(1+(q-1)z)^{n-(q-1)q^{m-1}}\nonumber\\
    && +w(1-z)^{(q-1)(q^{m-1}+q^{\frac{m-2}{2}})}(1+(q-1)z)^{n-(q-1)(q^{m-1}+q^{\frac{m-2}{2}})}+t(1-z)^n.
  \end{eqnarray}
Furthermore, we have
\begin{eqnarray}\label{eq-B1}
  (1+(q-1)z)^n &=& \sum_{l=0}^n\binom{n}{l}(q-1)^lz^l,
\end{eqnarray}
\begin{eqnarray}\label{eq-B2}
  &&u(1-z)^{(q-1)q^{m-1}-q^{\frac{m-2}{2}}}(1+(q-1)z)^{n-(q-1)q^{m-1}+q^{\frac{m-2}{2}}}\nonumber \\
  &=& u\sum_{l=0}^n\left(\sum_{i+j=l}\binom{(q-1)q^{m-1}-q^{\frac{m-2}{2}}}{i}\binom{n-(q-1)q^{m-1}+q^{\frac{m-2}{2}}}{j}(-1)^i(q-1)^j\right)z^l,
\end{eqnarray}
\begin{eqnarray}\label{eq-B3}
  &&v(1-z)^{(q-1)q^{m-1}}(1+(q-1)z)^{n-(q-1)q^{m-1}}\nonumber\\
  &=&v\sum_{l=0}^n\left(\sum_{i+j=l}\binom{(q-1)q^{m-1}}{i}\binom{n-(q-1)q^{m-1}}{j}(-1)^i(q-1)^j\right)z^l,
\end{eqnarray}
\begin{eqnarray}\label{eq-B4}
&&w(1-z)^{(q-1)(q^{m-1}+q^{\frac{m-2}{2}})}(1+(q-1)z)^{n-(q-1)(q^{m-1}+q^{\frac{m-2}{2}})}\nonumber \\
  &=& w\sum_{l=0}^n\left(\sum_{i+j=l}\binom{(q-1)(q^{m-1}+q^{\frac{m-2}{2}})}{i}\binom{n-(q-1)(q^{m-1}+q^{\frac{m-2}{2}})}{j}(-1)^i(q-1)^j\right)z^l,
\end{eqnarray}
\begin{eqnarray}\label{eq-B5}
  t(1-z)^n &=& t\sum_{l=0}^{n}\binom{n}{l}(-1)^lz^l.
\end{eqnarray}
Substituting Equations (\ref{eq-B1}), (\ref{eq-B2}), (\ref{eq-B3}), (\ref{eq-B4}) and (\ref{eq-B5}) into (\ref{eq-B}) yields
\begin{eqnarray*}
    q^{\frac{3m}{2}+1}A^{\perp}(z) &=& \sum_{l=0}^n\left[\binom{n}{l}(q-1)^l+\sum_{i+j=l}V(i,j)(-1)^i(q-1)^j+t\binom{n}{l}(-1)^l\right]z^l.
  \end{eqnarray*}
Then we have
\begin{eqnarray*}
  q^{\frac{3m}{2}+1}A_1^{\perp} &=& \binom{n}{1}(q-1)+\sum_{i+j=1}V(i,j)(-1)^i(q-1)^j-(q-1)\binom{n}{1}\\
  &=& -V(1,0)+V(0,1)(q-1)=0,
\end{eqnarray*}
\begin{eqnarray*}
  q^{\frac{3m}{2}+1}A_2^{\perp} &=& \binom{n}{2}(q-1)^2+\sum_{i+j=2}V(i,j)(-1)^i(q-1)^j+(q-1)\binom{n}{2}\\
  &=& \frac{q^{m+1}(q^m-1)(q-1)}{2}+V(2,0)-(q-1)V(1,1)+(q-1)^2V(0,2)\\
  &=&0,
\end{eqnarray*}
\begin{eqnarray*}
  q^{\frac{3m}{2}+1}A_3^{\perp} &=& \binom{n}{3}(q-1)^3+\sum_{i+j=3}V(i,j)(-1)^i(q-1)^j-(q-1)\binom{n}{3}\\
  &=& \frac{q^{m+1}(q^m-1)(q^m-2)(q-2)(q-1)}{6}+(q-1)^3V(0,3)\\&&-(q-1)^2V(1,2)+(q-1)V(2,1)-V(3,0)\\
  &=&0,
\end{eqnarray*}
\begin{eqnarray*}
  q^{\frac{3m}{2}+1}A_4^{\perp} &=& \binom{n}{4}(q-1)^4+\sum_{i+j=4}V(i,j)(-1)^i(q-1)^j+(q-1)\binom{n}{4}\\
  &=& \frac{q^{m+1}(q^m-1)(q^m-2)(q^m-3)(q^2-3q+3)(q-1)}{24}+(q-1)^4V(0,4)\\&&-(q-1)^3V(1,3)+(q-1)^2V(2,2)-(q-1)V(3,1)+V(4,0)\\
  &=&\frac{q^{\frac{5m}{2}+1}(q-1)(q^m-1)((q^{\frac{m}{2}}+1)(q^2-3q+3)-3q+3)}{24},
\end{eqnarray*}
which implies  $A_4^{\perp}=\frac{q^{m}(q-1)(q^m-1)((q^{\frac{m}{2}}+1)(q^2-3q+3)-3q+3)}{24}$. Then $A_1^{\perp}=A_2^{\perp}=A_3^{\perp}=0$ and $A_4^{\perp} > 0$. Hence, the minimum distance of $\overline{\cC}_{(q,m,\delta_2)}^{\perp}$ is $4$. Besides, it is easy to deduce that $\overline{\cC}_{(q,m,\delta_2)}^{\perp}$ is optimal according to the sphere-packing bound as $[q^m, q^m-\frac{3m}{2}-1, 5]$ linear code over $\gf_q$ does not exist.

Finally, we prove that the locality of $\overline{\cC}_{(q,m,\delta_2)}$ is $3$. By Lemma \ref{lem-AMth}, we deduce that the minimum weight codewords in $\overline{\cC}_{(q,m,\delta_2)}$ hold a $1$-design and the minimum weight codewords in $\overline{\cC}_{(q,m,\delta_2)}^{\perp}$ also hold a $1$-design. Then by Lemma \ref{lem-locality}, the locality of $\overline{\cC}_{(q,m,\delta_2)}$ is $d(\bar{\cC}_{(q,m,\delta_2)}^{\perp})-1=3$.
\end{proof}

\subsection{The sisth family of self-orthogonal codes from the BCH codes with designed distance $\delta_3=(q-1)q^{m-1}-1-q^{\lfloor(m+1)/2\rfloor}$}
In this subsection, let $q$ be an odd prime. Let $m$ be a positive integer with $m \geq 3$. Let $h=\lfloor(m-1)/2\rfloor+1$. Define $\tilde{g}_{(q,m,\delta_3)}(x)=(x-1)g_{(q,m,\delta_3)}(x)$ with $\delta_3=(q-1)q^{m-1}-1-q^{\lfloor(m+1)/2\rfloor}$. Let $\tilde{\cC}_{(q,m,\delta_3)}$ and $\cC_{(q,m,\delta_3)}$ be the cyclic codes over $\gf_q$ with generator polynomials $\tilde{g}_{(q,m,\delta_3)}(x)$ and $g_{(q,m,\delta_3)}(x)$, respectively.
By \cite{DFZ},
\begin{eqnarray*}
  \tilde{\cC}_{(q,m,\delta_3)} = \left\{\left(\tr_{q^m/q}(ax+bx^{1+q^{h}}+cx^{1+q^{h+1}})\right)_{x \in \gf_{q^m}^*}: a \in \gf_{q^m}, b \in \gf_{q^m}, c \in \gf_{q^m}\right\},
\end{eqnarray*}
and
\begin{eqnarray*}
  \cC_{(q,m,\delta_3)} = \left\{\left(\tr_{q^m/q}(ax+bx^{1+q^{h}}+cx^{1+q^{h+1}})+e\right)_{x \in \gf_{q^m}^*}: a \in \gf_{q^m}, b \in \gf_{q^m}, c \in \gf_{q^m}, e \in \gf_q\right\}.
\end{eqnarray*}
Note that $\cC_{(q,m,\delta_3)}$ is the augmented code of $\tilde{\cC}_{(q,m,\delta_3)}$.

Ding et al. determined the weight distribution of $\tilde{\cC}_{(q,m,\delta_3)}$ in \cite{DFZ}. Li et al. presented the weight distributions of $\cC_{(q,m,\delta_3)}$ and $\overline{\cC}_{(q,m,\delta_3)}$ in \cite{LWL}, where $\overline{\cC}_{(q,m,\delta_3)}$ is the extended code of $\cC_{(q,m,\delta_3)}$.
In what follows, we will prove that $\overline{\cC}_{(q,m,\delta_3)}$ is self-orthogonal and determine the dual distance of $\overline{\cC}_{(q,m,\delta_3)}$.

 We recall the weight distributions of $\overline{\cC}_{(q,m,\delta_3)}$ for odd $m$ and even $m$, respectively.

\begin{lemma}\cite{LWL}\label{lem-D2odd}
  Let $q$ be an odd prime and $m$ be an odd integer with $m \geq 5$. Then the extended code $\overline{\cC}_{(q,m,\delta_3)}$ with $\delta_3=(q-1)q^{m-1}-1-q^{\lfloor(m+1)/2\rfloor}$ has parameters $[q^m, 3m+1, (q-1)q^{m-1}-q^{\frac{m+1}{2}}]$ and its weight distribution is listed in Table \ref{tab6.5}.
\end{lemma}
 \begin{table}[h]
\begin{center}
\caption{The weight distribution of $\overline{\cC}_{(q,m,\delta_3)}$ in Lemma \ref{lem-D2odd}.}\label{tab6.5}
\begin{tabular}{@{}ll@{}}
\toprule%
Weight & Frequency  \\
\midrule
$0$ & $1$\\
$(q-1)q^{m-1}-q^{\frac{m+1}{2}}$ & $\frac{q^{m-2}(q^{m-1}-1)(q^m-1)}{2(q+1)}$\\
$(q-1)(q^{m-1}-q^{\frac{m-1}{2}})$ & $\frac{q^{m-1}(q^m-1)(q^{m-1}+q^{\frac{m-1}{2}})}{2}$\\
$(q-1)q^{m-1}-q^{\frac{m-1}{2}}$ & $\frac{q^{m-1}(q^m-1)(q^{m+3}-q^{m+2}-q^{m-1}-q^{\frac{m+3}{2}}+q^{\frac{m-1}{2}}+q^3)}{2(q+1)}$\\
$(q-1)q^{m-1}$ & $q(q^m-1)((q-1)(2q^{2m-2}+q^{2m-4}+q^{m-2})+q^{m-3}+1)$\\
$(q-1)q^{m-1}+q^{\frac{m-1}{2}}$ & $\frac{q^{m-1}(q^m-1)(q^{m+3}-q^{m+2}-q^{m-1}+q^{\frac{m+3}{2}}-q^{\frac{m-1}{2}}+q^3)}{2(q+1)}$\\
$(q-1)(q^{m-1}+q^{\frac{m-1}{2}})$ & $\frac{q^{m-1}(q^m-1)(q^{m-1}-q^{\frac{m-1}{2}})}{2}$\\
$(q-1)q^{m-1}+q^{\frac{m+1}{2}}$ & $\frac{q^{m-2}(q^{m-1}-1)(q^m-1)}{2(q+1)}$\\
$q^m$ & $q-1$\\
\bottomrule
\end{tabular}
\end{center}
\end{table}

\begin{lemma}\cite{LWL}\label{lem-D2even}
  Let $q$ be an odd prime and $m$ be an even integer with $m \geq 4$. Then the extended code $\overline{\cC}_{(q,m,\delta_3)}$ has parameters $[q^m, \frac{5m}{2}+1, (q-1)q^{m-1}-q^{\frac{m}{2}}]$ and its weight distribution is listed in Table \ref{tab6.4'}.
\end{lemma}
 \begin{table}[h]
\begin{center}
\caption{The weight distribution of $\overline{\cC}_{(q,m,\delta_3)}$ in Lemma \ref{lem-D2even}.}\label{tab6.4'}
\begin{tabular}{@{}ll@{}}
\toprule%
Weight & Frequency  \\
\midrule
$0$ & $1$\\
$(q-1)q^{m-1}-q^{\frac{m}{2}}$ & $\frac{q^{\frac{m}{2}-2}(q^m-1)(q^{m+2}-q^m+2q^{m-1}-2q^{\frac{m}{2}})}{2(q+1)}$\\
$(q-1)(q^{m-1}-q^{\frac{m}{2}-1})$ & $\frac{q^{m+1}(q^{\frac{m}{2}}+1)(q^m-1)}{2(q+1)}$\\
$(q-1)q^{m-1}-q^{\frac{m}{2}-1}$ & $\frac{q^m(q^{m+1}-2q^m+q)(q^{\frac{m}{2}}-1)}{2}$\\
$(q-1)q^{m-1}$ & $q(q^m-1)(1+q^{\frac{3m}{2}-1}-q^{\frac{3m}{2}-2}+q^{\frac{3m}{2}-3}-q^{m-2})$\\
$(q-1)q^{m-1}+q^{\frac{m}{2}-1}$ & $\frac{q^{m+1}(q-1)(q^{\frac{m}{2}}+1)(q^m-1)}{2(q+1)}$\\
$(q-1)(q^{m-1}+q^{\frac{m}{2}-1})$ & $\frac{q^{m}(q^{m+1}-2q^m+q)(q^\frac{m}{2}-1)}{2(q-1)}$\\
$(q-1)q^{m-1}+q^{\frac{m}{2}}$ & $\frac{q^{\frac{3m}{2}-2}(q^m-1)(q-1)}{2}$\\
$(q-1)(q^{m-1}+q^{\frac{m}{2}})$ & $\frac{q^{m-2}(q^m-1)(q^{\frac{m}{2}-1}-1)}{q^2-1}$\\
$q^m$ & $q-1$\\
\bottomrule
\end{tabular}
\end{center}
\end{table}

To give the following theorem, we denote by
\begin{eqnarray*}
  T(i,j) &=& u_1\binom{(q-1)q^{m-1}-q^{\frac{m+1}{2}}}{i}\binom{n-(q-1)q^{m-1}+q^{\frac{m+1}{2}}}{j}+u_2\binom{(q-1)(q^{m-1}-q^{\frac{m-1}{2}})}{i}\\&&\binom{n-(q-1)(q^{m-1}-q^{\frac{m-1}{2}})}{j}+
 u_3\binom{(q-1)q^{m-1}-q^{\frac{m-1}{2}}}{i}\binom{n-(q-1)q^{m-1}+q^{\frac{m-1}{2}}}{j}+\\
 && u_4\binom{(q-1)q^{m-1}}{i}\binom{n-(q-1)q^{m-1}}{j}+u_5\binom{(q-1)q^{m-1}+q^{\frac{m-1}{2}}}{i}\\
 && \binom{n-(q-1)q^{m-1}-q^{\frac{m-1}{2}}}{j}+u_6\binom{(q-1)(q^{m-1}+q^{\frac{m-1}{2}})}{i}\binom{n-(q-1)(q^{m-1}+q^{\frac{m-1}{2}})}{j}\\
 && +u_7\binom{(q-1)q^{m-1}+q^{\frac{m+1}{2}}}{i}\binom{n-(q-1)q^{m-1}-q^{\frac{m+1}{2}}}{j},
\end{eqnarray*}
where $n=q^m$, $u_1=\frac{q^{m-2}(q^{m-1}-1)(q^m-1)}{2(q+1)}$, $u_2=\frac{q^{m-1}(q^m-1)(q^{m-1}+q^{\frac{m-1}{2}})}{2}$, $u_3=q^{m-1}(q^m-1)(q^{m+3}-q^{m+2}-q^{m-1}-q^{\frac{m+3}{2}}+q^{\frac{m-1}{2}}+q^3)/2(q+1)$, $u_4=q(q^m-1)((q-1)(2q^{2m-2}+q^{2m-4}+q^{m-2})+q^{m-3}+1)$, $u_5=\frac{q^{m-1}(q^m-1)(q^{m+3}-q^{m+2}-q^{m-1}+q^{\frac{m+3}{2}}-q^{\frac{m-1}{2}}+q^3)}{2(q+1)}$, $u_6=\frac{q^{m-1}(q^m-1)(q^{m-1}-q^{\frac{m-1}{2}})}{2}$ and $u_7=\frac{q^{m-2}(q^{m-1}-1)(q^m-1)}{2(q+1)}$.

\begin{theorem}\label{th-C6odd}
  Let $q$ be an odd prime and $m$ be an odd integer with $m \geq 5$. Then the extended code $\overline{\cC}_{(q,m,\delta_3)}$ is a self-orthogonal $q$-divisible code and its dual code $\overline{\cC}_{(q,m,\delta_3)}^{\perp}$ has parameters $[q^m, q^m-3m-1, 6]$ for $q=3$ and $[q^m, q^m-3m-1, 4]$ for $q > 3$. Besides, The weight enumerator $A^{\perp}(z)$ of $\overline{\cC}_{(q,m,\delta_3)}^{\perp}$ is given by
  \begin{eqnarray*}
    q^{3m+1}A^{\perp}(z) &=& \sum_{l=0}^n\left[\binom{n}{l}(q-1)^l+\sum_{i+j=l}T(i,j)(-1)^i(q-1)^j+t\binom{n}{l}(-1)^l\right]z^l,
  \end{eqnarray*}
  where $n=q^m$ and $t=q-1$.
\end{theorem}

\begin{proof}
Firstly, by Lemma \ref{lem-D2odd}, we derive that $\overline{\cC}_{(q,m,\delta_3)}$ is $q$-divisible and $\mathbf{1} \in \overline{\cC}_{(q,m,\delta_3)}$. Then by Theorem \ref{th-selforthogonal}, we deduce that $\overline{\cC}_{(q,m,\delta_3)}$ is self-orthogonal.

Then we determine the weight enumerator $A^{\perp}(z)$ and minimum distance of $\overline{\cC}_{(q,m,\delta_3)}^{\perp}$.
  Let $n=q^m$. By Lemma \ref{lem-D2odd}, we derive that the weight enumerator of $\overline{\cC}_{(q,m,\delta_3)}$ is
  \begin{eqnarray*}
    A(z) &=& 1+u_1z^{(q-1)q^{m-1}-q^{\frac{m+1}{2}}}+u_2z^{(q-1)(q^{m-1}-q^{\frac{m-1}{2}})}+u_3z^{(q-1)q^{m-1}-q^{\frac{m-1}{2}}}+u_4z^{(q-1)q^{m-1}}+\\
    &&u_5z^{(q-1)q^{m-1}+q^{\frac{m-1}{2}}}+u_6z^{(q-1)(q^{m-1}+q^{\frac{m-1}{2}})}+u_7z^{(q-1)q^{m-1}+q^{\frac{m+1}{2}}}+tz^{q^m},
  \end{eqnarray*}
  where $u_1=\frac{q^{m-2}(q^{m-1}-1)(q^m-1)}{2(q+1)}$, $u_2=\frac{q^{m-1}(q^m-1)(q^{m-1}+q^{\frac{m-1}{2}})}{2}$, $u_3=q^{m-1}(q^m-1)(q^{m+3}-q^{m+2}-q^{m-1}-q^{\frac{m+3}{2}}+q^{\frac{m-1}{2}}+q^3)/2(q+1)$, $u_4=q(q^m-1)((q-1)(2q^{2m-2}+q^{2m-4}+q^{m-2})+q^{m-3}+1)$, $u_5=\frac{q^{m-1}(q^m-1)(q^{m+3}-q^{m+2}-q^{m-1}+q^{\frac{m+3}{2}}-q^{\frac{m-1}{2}}+q^3)}{2(q+1)}$, $u_6=\frac{q^{m-1}(q^m-1)(q^{m-1}-q^{\frac{m-1}{2}})}{2}$, $u_7=\frac{q^{m-2}(q^{m-1}-1)(q^m-1)}{2(q+1)}$ and $t=q-1$. By Lemma \ref{lem-Mac}, the weight enumerator of $\overline{\cC}_{(q,m,\delta_3)}$ satisfies
  \begin{eqnarray}\label{eq-E}
    q^{3m+1}A^{\perp}(z) &=& (1+(q-1)z)^nA\left(\frac{1-z}{1+(q-1)z}\right)\nonumber\\
    &=&(1+(q-1)z)^n+u_1(1-z)^{(q-1)q^{m-1}-q^{\frac{m+1}{2}}}(1+(q-1)z)^{n-(q-1)q^{m-1}+q^{\frac{m+1}{2}}}\nonumber\\
    && +u_2(1-z)^{(q-1)(q^{m-1}-q^{\frac{m-1}{2}})}(1+(q-1)z)^{n-(q-1)q^{m-1}(q-1)(q^{m-1}-q^{\frac{m-1}{2}})}\nonumber\\
    && +u_3(1-z)^{(q-1)q^{m-1}-q^{\frac{m-1}{2}}}(1+(q-1)z)^{n-(q-1)q^{m-1}+q^{\frac{m-1}{2}}}\nonumber\\
    &&+u_4(1-z)^{(q-1)q^{m-1}}(1+(q-1)z)^{n-(q-1)q^{m-1}}\nonumber\\
    &&+u_5(1-z)^{(q-1)q^{m-1}+q^{\frac{m-1}{2}}}(1+(q-1)z)^{n-(q-1)q^{m-1}-q^{\frac{m-1}{2}}}\nonumber\\
    && +u_6(1-z)^{(q-1)(q^{m-1}+q^{\frac{m-1}{2}})}(1+(q-1)z)^{n-(q-1)(q^{m-1}+q^{\frac{m-1}{2}})}\nonumber\\
    &&+u_7(1-z)^{(q-1)q^{m-1}+q^{\frac{m+1}{2}}}(1+(q-1)z)^{n-(q-1)q^{m-1}-q^{\frac{m+1}{2}}}+t(1-z)^n.
  \end{eqnarray}
Furthermore, we have
\begin{eqnarray}\label{eq-E1}
  (1+(q-1)z)^n &=& \sum_{l=0}^n\binom{n}{l}(q-1)^lz^l,
\end{eqnarray}
\begin{eqnarray}\label{eq-E2}
  &&u_1(1-z)^{(q-1)q^{m-1}-q^{\frac{m+1}{2}}}(1+(q-1)z)^{n-(q-1)q^{m-1}+q^{\frac{m+1}{2}}}\nonumber \\
  &=& u_1\sum_{l=0}^n\left(\sum_{i+j=l}\binom{(q-1)q^{m-1}-q^{\frac{m+1}{2}}}{i}\binom{n-(q-1)q^{m-1}+q^{\frac{m+1}{2}}}{j}(-1)^i(q-1)^j\right)z^l,
\end{eqnarray}
\begin{eqnarray}\label{eq-E3}
  &&u_2(1-z)^{(q-1)(q^{m-1}-q^{\frac{m-1}{2}})}(1+(q-1)z)^{n-(q-1)(q^{m-1}-q^{\frac{m-1}{2}})}\nonumber\\
  &=&u_2\sum_{l=0}^n\left(\sum_{i+j=l}\binom{(q-1)(q^{m-1}-q^{\frac{m-1}{2}})}{i}\binom{n-(q-1)(q^{m-1}-q^{\frac{m-1}{2}})}{j}(-1)^i(q-1)^j\right)z^l,
\end{eqnarray}
\begin{eqnarray}\label{eq-E4}
&&u_3(1-z)^{(q-1)q^{m-1}-q^{\frac{m-1}{2}}}(1+(q-1)z)^{n-(q-1)q^{m-1}+q^{\frac{m-1}{2}}}\nonumber \\
  &=& u_3\sum_{l=0}^n\left(\sum_{i+j=l}\binom{(q-1)q^{m-1}-q^{\frac{m-1}{2}}}{i}\binom{n-(q-1)q^{m-1}+q^{\frac{m-1}{2}}}{j}(-1)^i(q-1)^j\right)z^l,
\end{eqnarray}
\begin{eqnarray}\label{eq-E5}
&&u_4(1-z)^{(q-1)q^{m-1}}(1+(q-1)z)^{n-(q-1)q^{m-1}}\nonumber \\
  &=& u_4\sum_{l=0}^n\left(\sum_{i+j=l}\binom{(q-1)q^{m-1}}{i}\binom{n-(q-1)q^{m-1}}{j}(-1)^i(q-1)^j\right)z^l,
\end{eqnarray}
\begin{eqnarray}\label{eq-E6}
&&u_5(1-z)^{(q-1)q^{m-1}+q^{\frac{m-1}{2}}}(1+(q-1)z)^{n-(q-1)q^{m-1}-q^{\frac{m-1}{2}}}\nonumber \\
  &=& u_5\sum_{l=0}^n\left(\sum_{i+j=l}\binom{(q-1)q^{m-1}+q^{\frac{m-1}{2}}}{i}\binom{n-(q-1)q^{m-1}-q^{\frac{m-1}{2}}}{j}(-1)^i(q-1)^j\right)z^l,
\end{eqnarray}
\begin{eqnarray}\label{eq-E7}
&&u_6(1-z)^{(q-1)(q^{m-1}+q^{\frac{m-1}{2}})}(1+(q-1)z)^{n-(q-1)(q^{m-1}+q^{\frac{m-1}{2}})}\nonumber \\
  &=& u_6\sum_{l=0}^n\left(\sum_{i+j=l}\binom{(q-1)(q^{m-1}+q^{\frac{m-1}{2}})}{i}\binom{n-(q-1)(q^{m-1}+q^{\frac{m-1}{2}})}{j}(-1)^i(q-1)^j\right)z^l,
\end{eqnarray}
\begin{eqnarray}\label{eq-E8}
&&u_7(1-z)^{(q-1)q^{m-1}+q^{\frac{m+1}{2}}}(1+(q-1)z)^{n-(q-1)q^{m-1}-q^{\frac{m+1}{2}}}\nonumber \\
  &=& u_7\sum_{l=0}^n\left(\sum_{i+j=l}\binom{(q-1)q^{m-1}+q^{\frac{m+1}{2}}}{i}\binom{n-(q-1)q^{m-1}-q^{\frac{m+1}{2}}}{j}(-1)^i(q-1)^j\right)z^l,
\end{eqnarray}
\begin{eqnarray}\label{eq-E9}
  t(1-z)^n &=& t\sum_{l=0}^{n}\binom{n}{l}(-1)^lz^l.
\end{eqnarray}
Substituting Equations (\ref{eq-E1}), (\ref{eq-E2}), (\ref{eq-E3}), (\ref{eq-E4}), (\ref{eq-E5}), (\ref{eq-E6}), (\ref{eq-E7}), (\ref{eq-E8}) and (\ref{eq-E9}) into (\ref{eq-E}) yields
\begin{eqnarray*}
    q^{3m+1}A^{\perp}(z) &=& \sum_{l=0}^n\left[\binom{n}{l}(q-1)^l+\sum_{i+j=l}T(i,j)(-1)^i(q-1)^j+t\binom{n}{l}(-1)^l\right]z^l.
  \end{eqnarray*}
Then we have
\begin{eqnarray*}
  q^{3m+1}A_1^{\perp} &=& \binom{n}{1}(q-1)+\sum_{i+j=1}T(i,j)(-1)^i(q-1)^j-(q-1)\binom{n}{1}\\
  &=& -T(1,0)+T(0,1)(q-1)=0,
\end{eqnarray*}
\begin{eqnarray*}
  q^{3m+1}A_2^{\perp} &=& \binom{n}{2}(q-1)^2+\sum_{i+j=2}T(i,j)(-1)^i(q-1)^j+(q-1)\binom{n}{2}\\
  &=& \frac{q^{m+1}(q^m-1)(q-1)}{2}+T(2,0)-(q-1)T(1,1)+(q-1)^2T(0,2)\\
  &=&0,
\end{eqnarray*}
\begin{eqnarray*}
  q^{3m+1}A_3^{\perp} &=& \binom{n}{3}(q-1)^3+\sum_{i+j=3}T(i,j)(-1)^i(q-1)^j-(q-1)\binom{n}{3}\\
  &=& \frac{q^{m+1}(q^m-1)(q^m-2)(q-2)(q-1)}{6}+(q-1)^3T(0,3)\\&&-(q-1)^2T(1,2)+(q-1)T(2,1)-T(3,0)\\
  &=&0,
\end{eqnarray*}
\begin{eqnarray*}
  q^{3m+1}A_4^{\perp} &=& \binom{n}{4}(q-1)^4+\sum_{i+j=4}T(i,j)(-1)^i(q-1)^j+(q-1)\binom{n}{4}\\
  &=& \frac{q^{m+1}(q^m-1)(q^m-2)(q^m-3)(q^2-3q+3)(q-1)}{24}+(q-1)^4T(0,4)\\&&-(q-1)^3T(1,3)+(q-1)^2T(2,2)-(q-1)T(3,1)+T(4,0)\\
  &=&\frac{q^{4m+1}(q-3)(q-1)(q-2)(q^m-1)}{24},
\end{eqnarray*}
which yields that $A_4^{\perp}=\frac{q^{m}(q-3)(q-1)(q-2)(q^m-1)}{24}$. If $q > 3$, then $A_1^{\perp}=A_2^{\perp}=A_3^{\perp}=0$ and $A_4^{\perp} > 0$. Hence, the minimum distance of $\overline{\cC}_{(q,m,\delta_3)}^{\perp}$ is $4$. If $q=3$, then $A_1^{\perp}=A_2^{\perp}=A_3^{\perp}=A_4^{\perp}=0$. Now let $q=3$, we have
\begin{eqnarray*}
  3^{3m+1}A_5^{\perp} &=& \binom{3^m}{5}2^5+\sum_{i+j=5}T(i,j)(-1)^i2^j-2\binom{3^m}{5}\\
  &=& \frac{3^{m}(3^m-1)(3^m-2)(3^m-3)(3^m-4)}{4}+2^5T(0,5)-2^4T(1,4)\\&&+2^3T(2,3)-4T(3,2)+2T(4,1)-T(5,0)\\
  &=&0,
\end{eqnarray*}
\begin{eqnarray*}
  3^{3m+1}A_6^{\perp} &=& \binom{3^m}{6}2^6+\sum_{i+j=6}T(i,j)(-1)^i2^j+2\binom{3^m}{6}\\
  &=& \frac{11\cdot3^{m-1}(3^m-1)(3^m-2)(3^m-3)(3^m-4)(3^m-5)}{40}+2^6T(0,6)-2^5T(1,5)\\&&+2^4T(2,4)-2^3T(3,3)+4T(4,2)-2T(5,1)+T(6,0)\\
  &=& \frac{3^{4m}(3^m-1)(3^{m-1}-1)}{2},
\end{eqnarray*}
which implies that $A_6^{\perp}=\frac{3^{m-1}(3^m-1)(3^{m-1}-1)}{2} > 0$. Then the minimum distance of $\overline{\cC}_{(q,m,\delta_3)}^{\perp}$ is $6$.
The desired conclusions follow.
\end{proof}

To give the next theorem, we denote by
\begin{eqnarray*}
  R(i,j) &=& u_1\binom{(q-1)q^{m-1}-q^{\frac{m}{2}}}{i}\binom{n-(q-1)q^{m-1}+q^{\frac{m}{2}}}{j}+u_2\binom{(q-1)(q^{m-1}-q^{\frac{m}{2}-1})}{i}\\&&\binom{n-(q-1)(q^{m-1}-q^{\frac{m}{2}-1})}{j}+
 u_3\binom{(q-1)q^{m-1}-q^{\frac{m}{2}-1}}{i}\binom{n-(q-1)q^{m-1}+q^{\frac{m}{2}-1}}{j}+\\
 && u_4\binom{(q-1)q^{m-1}}{i}\binom{n-(q-1)q^{m-1}}{j}+u_5\binom{(q-1)q^{m-1}+q^{\frac{m}{2}-1}}{i}\\
 && \binom{n-(q-1)q^{m-1}-q^{\frac{m}{2}-1}}{j}+u_6\binom{(q-1)(q^{m-1}+q^{\frac{m}{2}-1})}{i}\binom{n-(q-1)(q^{m-1}+q^{\frac{m}{2}-1})}{j}\\
 && +u_7\binom{(q-1)q^{m-1}+q^{\frac{m}{2}}}{i}\binom{n-(q-1)q^{m-1}-q^{\frac{m}{2}}}{j}\\
 && +u_8\binom{(q-1)(q^{m-1}+q^{\frac{m}{2}})}{i}\binom{n-(q-1)(q^{m-1}+q^{\frac{m}{2}})}{j},
\end{eqnarray*}
where $n=q^m$, $u_1=\frac{q^{\frac{m}{2}-2}(q^m-1)(q^{m+2}-q^m+2q^{m-1}-2q^{\frac{m}{2}})}{2(q+1)}$, $u_2=\frac{q^{m+1}(q^{\frac{m}{2}}+1)(q^m-1)}{2(q+1)}$, $u_3=\frac{q^m(q^{m+1}-2q^m+q)(q^{\frac{m}{2}}-1)}{2}$, $u_4=q(q^m-1)(1+q^{\frac{3m}{2}-1}-q^{\frac{3m}{2}-2}+q^{\frac{3m}{2}-3}-q^{m-2})$, $u_5=\frac{q^{m+1}(q-1)(q^{\frac{m}{2}}+1)(q^m-1)}{2(q+1)}$, $u_6=q^{m}(q^{m+1}-2q^m+q)(q^\frac{m}{2}-1)/2(q-1)$, $u_7=\frac{q^{\frac{3m}{2}-2}(q^m-1)(q-1)}{2}$ and $u_8=\frac{q^{m-2}(q^m-1)(q^{\frac{m}{2}-1}-1)}{q^2-1}$.

\begin{theorem}\label{th-C6even}
  Let $q$ be an odd prime and $m$ be an even integer with $m \geq 4$. Then the extended code $\overline{\cC}_{(q,m,\delta_3)}$ is a self-orthogonal $q$-divisible code and its dual code $\overline{\cC}_{(q,m,\delta_3)}^{\perp}$ has parameters $[q^m, q^m-\frac{5m}{2}-1, 6]$ for $q=3$ and $[q^m, q^m-\frac{5m}{2}-1, 4]$ for $q > 3$. Besides, The weight enumerator $A^{\perp}(z)$ of $\overline{\cC}_{(q,m,\delta_3)}^{\perp}$ is given by
  \begin{eqnarray*}
    q^{\frac{5m}{2}+1}A^{\perp}(z) &=& \sum_{l=0}^n\left[\binom{n}{l}(q-1)^l+\sum_{i+j=l}R(i,j)(-1)^i(q-1)^j+t\binom{n}{l}(-1)^l\right]z^l,
  \end{eqnarray*}
  where $n=q^m$ and $t=q-1$. In particular, $\overline{\cC}_{(q,m,\delta_3)}^{\perp}$ is optimal according to the sphere-packing bound for $q=3$.
\end{theorem}

\begin{proof}
Firstly, by Lemma \ref{lem-D2even}, we derive that $\overline{\cC}_{(q,m,\delta_3)}$ is $q$-divisible and $\mathbf{1} \in \overline{\cC}_{(q,m,\delta_3)}$. Then by Theorem \ref{th-selforthogonal}, we deduce that $\overline{\cC}_{(q,m,\delta_3)}$ is self-orthogonal.

Then we determine the weight enumerator $A^{\perp}(z)$ and minimum distance of $\overline{\cC}_{(q,m,\delta_3)}^{\perp}$.
  Let $n=q^m$. By Lemma \ref{lem-D2even}, we derive that the weight enumerator of $\overline{\cC}_{(q,m,\delta_3)}$ is
  \begin{eqnarray*}
    A(z) &=& 1+u_1z^{(q-1)q^{m-1}-q^{\frac{m}{2}}}+u_2z^{(q-1)(q^{m-1}-q^{\frac{m}{2}-1})}+u_3z^{(q-1)q^{m-1}-q^{\frac{m}{2}-1}}+u_4z^{(q-1)q^{m-1}}+\\
    &&u_5z^{(q-1)q^{m-1}+q^{\frac{m}{2}-1}}+u_6z^{(q-1)(q^{m-1}+q^{\frac{m}{2}-1})}+u_7z^{(q-1)q^{m-1}+q^{\frac{m}{2}}}
     u_8z^{(q-1)(q^{m-1}+q^{\frac{m}{2}})}+tz^{q^m},
  \end{eqnarray*}
  where $u_1=\frac{q^{\frac{m}{2}-2}(q^m-1)(q^{m+2}-q^m+2q^{m-1}-2q^{\frac{m}{2}})}{2(q+1)}$, $u_2=\frac{q^{m+1}(q^{\frac{m}{2}}+1)(q^m-1)}{2(q+1)}$, $u_3=\frac{q^m(q^{m+1}-2q^m+q)(q^{\frac{m}{2}}-1)}{2}$, $u_4=q(q^m-1)(1+q^{\frac{3m}{2}-1}-q^{\frac{3m}{2}-2}+q^{\frac{3m}{2}-3}-q^{m-2})$, $u_5=\frac{q^{m+1}(q-1)(q^{\frac{m}{2}}+1)(q^m-1)}{2(q+1)}$, $u_6=q^{m}(q^{m+1}-2q^m+q)(q^\frac{m}{2}-1)/2(q-1)$, $u_7=\frac{q^{\frac{3m}{2}-2}(q^m-1)(q-1)}{2}$ and $u_8=\frac{q^{m-2}(q^m-1)(q^{\frac{m}{2}-1}-1)}{q^2-1}$ and $t=q-1$. By Lemma \ref{lem-Mac}, the weight enumerator of $\overline{\cC}_{(q,m,\delta_3)}$ satisfies
  \begin{eqnarray}\label{eq-D}
    q^{\frac{5m}{2}+1}A^{\perp}(z) &=& (1+(q-1)z)^nA\left(\frac{1-z}{1+(q-1)z}\right)\nonumber\\
    &=&(1+(q-1)z)^n+u_1(1-z)^{(q-1)q^{m-1}-q^{\frac{m}{2}}}(1+(q-1)z)^{n-(q-1)q^{m-1}+q^{\frac{m}{2}}}\nonumber\\
    && +u_2(1-z)^{(q-1)(q^{m-1}-q^{\frac{m}{2}-1})}(1+(q-1)z)^{n-(q-1)(q^{m-1}-q^{\frac{m}{2}-1})}\nonumber\\
    && +u_3(1-z)^{(q-1)q^{m-1}-q^{\frac{m}{2}-1}}(1+(q-1)z)^{n-(q-1)q^{m-1}+q^{\frac{m}{2}-1}}\nonumber\\
    &&+u_4(1-z)^{(q-1)q^{m-1}}(1+(q-1)z)^{n-(q-1)q^{m-1}}\nonumber\\
    &&+u_5(1-z)^{(q-1)q^{m-1}+q^{\frac{m}{2}-1}}(1+(q-1)z)^{n-(q-1)q^{m-1}-q^{\frac{m}{2}-1}}\nonumber\\
    && +u_6(1-z)^{(q-1)(q^{m-1}+q^{\frac{m}{2}-1})}(1+(q-1)z)^{n-(q-1)(q^{m-1}+q^{\frac{m}{2}-1})}\nonumber\\
    &&+u_7(1-z)^{(q-1)q^{m-1}+q^{\frac{m}{2}}}(1+(q-1)z)^{n-(q-1)q^{m-1}-q^{\frac{m}{2}}}\nonumber\\
    &&+u_8(1-z)^{(q-1)(q^{m-1}+q^{\frac{m}{2}})}(1+(q-1)z)^{n-(q-1)(q^{m-1}+q^{\frac{m}{2}})}+t(1-z)^n.
  \end{eqnarray}
Furthermore, we have
\begin{eqnarray}\label{eq-D1}
  (1+(q-1)z)^n &=& \sum_{l=0}^n\binom{n}{l}(q-1)^lz^l,
\end{eqnarray}
\begin{eqnarray}\label{eq-D2}
  &&u_1(1-z)^{(q-1)q^{m-1}-q^{\frac{m}{2}}}(1+(q-1)z)^{n-(q-1)q^{m-1}+q^{\frac{m}{2}}}\nonumber \\
  &=& u_1\sum_{l=0}^n\left(\sum_{i+j=l}\binom{(q-1)q^{m-1}-q^{\frac{m}{2}}}{i}\binom{n-(q-1)q^{m-1}+q^{\frac{m}{2}}}{j}(-1)^i(q-1)^j\right)z^l,
\end{eqnarray}
\begin{eqnarray}\label{eq-D3}
  &&u_2(1-z)^{(q-1)(q^{m-1}-q^{\frac{m}{2}-1})}(1+(q-1)z)^{n-(q-1)(q^{m-1}-q^{\frac{m}{2}-1})}\nonumber\\
  &=&u_2\sum_{l=0}^n\left(\sum_{i+j=l}\binom{(q-1)(q^{m-1}-q^{\frac{m}{2}-1})}{i}\binom{n-(q-1)(q^{m-1}-q^{\frac{m}{2}-1})}{j}(-1)^i(q-1)^j\right)z^l,
\end{eqnarray}
\begin{eqnarray}\label{eq-D4}
&&u_3(1-z)^{(q-1)q^{m-1}-q^{\frac{m}{2}-1}}(1+(q-1)z)^{n-(q-1)q^{m-1}+q^{\frac{m}{2}-1}}\nonumber \\
  &=& u_3\sum_{l=0}^n\left(\sum_{i+j=l}\binom{(q-1)q^{m-1}-q^{\frac{m}{2}-1}}{i}\binom{n-(q-1)q^{m-1}+q^{\frac{m}{2}-1}}{j}(-1)^i(q-1)^j\right)z^l,
\end{eqnarray}
\begin{eqnarray}\label{eq-D5}
&&u_4(1-z)^{(q-1)q^{m-1}}(1+(q-1)z)^{n-(q-1)q^{m-1}}\nonumber \\
  &=& u_4\sum_{l=0}^n\left(\sum_{i+j=l}\binom{(q-1)q^{m-1}}{i}\binom{n-(q-1)q^{m-1}}{j}(-1)^i(q-1)^j\right)z^l,
\end{eqnarray}
\begin{eqnarray}\label{eq-D6}
&&u_5(1-z)^{(q-1)q^{m-1}+q^{\frac{m}{2}-1}}(1+(q-1)z)^{n-(q-1)q^{m-1}-q^{\frac{m}{2}-1}}\nonumber \\
  &=& u_5\sum_{l=0}^n\left(\sum_{i+j=l}\binom{(q-1)q^{m-1}+q^{\frac{m}{2}-1}}{i}\binom{n-(q-1)q^{m-1}-q^{\frac{m}{2}-1}}{j}(-1)^i(q-1)^j\right)z^l,
\end{eqnarray}
\begin{eqnarray}\label{eq-D7}
&&u_6(1-z)^{(q-1)(q^{m-1}+q^{\frac{m}{2}-1})}(1+(q-1)z)^{n-(q-1)(q^{m-1}+q^{\frac{m}{2}-1})}\nonumber \\
  &=& u_6\sum_{l=0}^n\left(\sum_{i+j=l}\binom{(q-1)(q^{m-1}+q^{\frac{m}{2}-1})}{i}\binom{n-(q-1)(q^{m-1}+q^{\frac{m}{2}-1})}{j}(-1)^i(q-1)^j\right)z^l,
\end{eqnarray}
\begin{eqnarray}\label{eq-D8}
&&u_7(1-z)^{(q-1)q^{m-1}+q^{\frac{m}{2}}}(1+(q-1)z)^{n-(q-1)q^{m-1}-q^{\frac{m}{2}}}\nonumber \\
  &=& u_7\sum_{l=0}^n\left(\sum_{i+j=l}\binom{(q-1)q^{m-1}+q^{\frac{m}{2}}}{i}\binom{n-(q-1)q^{m-1}-q^{\frac{m}{2}}}{j}(-1)^i(q-1)^j\right)z^l,
\end{eqnarray}
\begin{eqnarray}\label{eq-D9}
&&u_8(1-z)^{(q-1)(q^{m-1}+q^{\frac{m}{2}})}(1+(q-1)z)^{n-(q-1)(q^{m-1}+q^{\frac{m}{2}})}\nonumber \\
  &=& u_8\sum_{l=0}^n\left(\sum_{i+j=l}\binom{(q-1)(q^{m-1}+q^{\frac{m}{2}})}{i}\binom{n-(q-1)(q^{m-1}+q^{\frac{m}{2}})}{j}(-1)^i(q-1)^j\right)z^l,
\end{eqnarray}
\begin{eqnarray}\label{eq-D10}
  t(1-z)^n &=& t\sum_{l=0}^{n}\binom{n}{l}(-1)^lz^l.
\end{eqnarray}
Substituting Equations (\ref{eq-D1}), (\ref{eq-D2}), (\ref{eq-D3}), (\ref{eq-D4}), (\ref{eq-D5}), (\ref{eq-D6}), (\ref{eq-D7}), (\ref{eq-D8}), (\ref{eq-D9}) and (\ref{eq-D10}) into (\ref{eq-D}) yields that
\begin{eqnarray*}
    q^{\frac{5m}{2}+1}A^{\perp}(z) &=& \sum_{l=0}^n\left[\binom{n}{l}(q-1)^l+\sum_{i+j=l}R(i,j)(-1)^i(q-1)^j+t\binom{n}{l}(-1)^l\right]z^l.
  \end{eqnarray*}
Then we have
\begin{eqnarray*}
  q^{\frac{5m}{2}+1}A_1^{\perp} &=& \binom{n}{1}(q-1)+\sum_{i+j=1}R(i,j)(-1)^i(q-1)^j-(q-1)\binom{n}{1}\\
  &=& -R(1,0)+R(0,1)(q-1)=0,
\end{eqnarray*}
\begin{eqnarray*}
  q^{\frac{5m}{2}+1}A_2^{\perp} &=& \binom{n}{2}(q-1)^2+\sum_{i+j=2}R(i,j)(-1)^i(q-1)^j+(q-1)\binom{n}{2}\\
  &=& \frac{q^{m+1}(q^m-1)(q-1)}{2}+R(2,0)-(q-1)R(1,1)+(q-1)^2R(0,2)\\
  &=&0,
\end{eqnarray*}
\begin{eqnarray*}
  q^{\frac{5m}{2}+1}A_3^{\perp} &=& \binom{n}{3}(q-1)^3+\sum_{i+j=3}R(i,j)(-1)^i(q-1)^j-(q-1)\binom{n}{3}\\
  &=& \frac{q^{m+1}(q^m-1)(q^m-2)(q-2)(q-1)}{6}+(q-1)^3R(0,3)\\&&-(q-1)^2R(1,2)+(q-1)R(2,1)-R(3,0)\\
  &=&0,
\end{eqnarray*}
\begin{eqnarray*}
  q^{\frac{5m}{2}+1}A_4^{\perp} &=& \binom{n}{4}(q-1)^4+\sum_{i+j=4}R(i,j)(-1)^i(q-1)^j+(q-1)\binom{n}{4}\\
  &=& \frac{q^{m+1}(q^m-1)(q^m-2)(q^m-3)(q^2-3q+3)(q-1)}{24}+(q-1)^4R(0,4)\\&&-(q-1)^3R(1,3)+(q-1)^2R(2,2)-(q-1)R(3,1)+R(4,0)\\
  &=&\frac{q^{\frac{7m}{2}+1}(q-3)(q-1)(q-2)(q^m-1)}{24},
\end{eqnarray*}
which yields that $A_4^{\perp}=\frac{q^{m}(q-3)(q-1)(q-2)(q^m-1)}{24}$. If $q > 3$, then $A_1^{\perp}=A_2^{\perp}=A_3^{\perp}=0$ and $A_4^{\perp} > 0$. Hence, the minimum distance of $\overline{\cC}_{(q,m,\delta_3)}^{\perp}$ is $4$. If $q=3$, then $A_1^{\perp}=A_2^{\perp}=A_3^{\perp}=A_4^{\perp}=0$. Now let $q=3$, we have
\begin{eqnarray*}
  3^{\frac{5m}{2}+1}A_5^{\perp} &=& \binom{3^m}{5}2^5+\sum_{i+j=5}R(i,j)(-1)^i2^j-2\binom{3^m}{5}\\
  &=& \frac{3^{m}(3^m-1)(3^m-2)(3^m-3)(3^m-4)}{4}+2^5R(0,5)-2^4R(1,4)\\&&+2^3R(2,3)-4R(3,2)+2R(4,1)-R(5,0)\\
  &=&0,
\end{eqnarray*}
\begin{eqnarray*}
  3^{\frac{5m}{2}+1}A_6^{\perp} &=& \binom{3^m}{6}2^6+\sum_{i+j=6}R(i,j)(-1)^i2^j+2\binom{3^m}{6}\\
  &=& \frac{11\cdot3^{m-1}(3^m-1)(3^m-2)(3^m-3)(3^m-4)(3^m-5)}{40}+2^6R(0,6)-2^5R(1,5)\\&&+2^4R(2,4)-2^3R(3,3)+4R(4,2)-2R(5,1)+R(6,0)\\
  &=& \frac{3^{\frac{7m}{2}}(3^m-1)(11\cdot3^{\frac{3m}{2}-1}+11\cdot3^{m-1}+7-33\cdot3^{\frac{m}{2}})}{40},
\end{eqnarray*}
which implies that $A_6^{\perp}=\frac{3^{m-1}(3^m-1)(11\cdot3^{\frac{3m}{2}-1}+11\cdot3^{m-1}+7-11\cdot3^{\frac{m+2}{2}})}{40} > 0$. Then the minimum distance of $\overline{\cC}_{(q,m,\delta_3)}^{\perp}$ is $6$.
The desired conclusion follows.
\end{proof}

\begin{remark}
In this section, we obtain many (almost) optimal codes by the Code Tables in \cite{Codetable}. We list them in Table \ref{tab-optimalcode}.
  \begin{table}[!h]
\begin{center}
\caption{Optimal codes or almost optimal codes derived in Theorem \ref{th-Ceven}, \ref{th-C4odd}, \ref{th-C4even}, \ref{th-C5odd}, \ref{th-C5even}, \ref{th-C6odd}  and \ref{th-C6even}.}\label{tab-optimalcode}
\begin{tabular}{llll}
\toprule
Conditions & Code & Parameters & Optimality \\
\midrule
$p=5,k=1,m=3$ & $\overline{\widehat{\cC_1}}^{\perp}$ & $[125,118,3]$ & Almost optimal\\
$p=3,f=1,s=2$ & $\overline{\widehat{\cC_4}}$ & $[21,5,12]$ & Optimal\\
$p=3,f=1,s=2$ & $\overline{\widehat{\cC_4}}^{\perp}$ & $[21,16,3]$ & Optimal\\
$p=3,f=1,s=3$ & $\overline{\widehat{\cC_4}}^{\perp}$ & $[183,176,3]$ & Optimal\\
$p=5,f=1,s=2$ & $\overline{\widehat{\cC_4}}$ & $[105,5,80]$ & Optimal\\
$p=5,f=1,s=2$ & $\overline{\widehat{\cC_4}}^{\perp}$ & $[105,100,3]$ & Optimal\\
$q=3,m=3$     & $\overline{\cC}_{(q,m,\delta_2)}$    & $[27,7,15]$  & Optimal\\
$q=3,m=3$     & $\overline{\cC}_{(q,m,\delta_2)}^{\perp}$    & $[27,20,5]$  & Optimal\\
$q=3,m=5$     & $\overline{\cC}_{(q,m,\delta_2)}$    & $[243,11,153]$  & Optimal\\
$q=3,m=5$     & $\overline{\cC}_{(q,m,\delta_2)}^{\perp}$    & $[243,232,5]$  & Optimal\\
$q=5,m=3$     & $\overline{\cC}_{(q,m,\delta_2)}$    & $[125,7,95]$  & Optimal\\
$q=5,m=3$     & $\overline{\cC}_{(q,m,\delta_2)}^{\perp}$    & $[125,118,4]$  & Optimal\\
$q=3,m=4$     & $\overline{\cC}_{(q,m,\delta_2)}$    & $[81,7,51]$  & Optimal\\
$q=3,m=4$     & $\overline{\cC}_{(q,m,\delta_2)}^{\perp}$    & $[81,74,4]$  & Optimal\\
$q=3,m=4$     & $\overline{\cC}_{(q,m,\delta_3)}$    & $[81,11,45]$  & Optimal\\
$q=3,m=4$     & $\overline{\cC}_{(q,m,\delta_3)}^{\perp}$    & $[81,70,6]$  & Optimal\\
$q=3,m=5$     & $\overline{\cC}_{(q,m,\delta_3)}$    & $[243,16,135]$  & Optimal\\
$q=3,m=5$     & $\overline{\cC}_{(q,m,\delta_3)}^{\perp}$    & $[243,227,6]$  & Optimal\\
\bottomrule
\end{tabular}
\end{center}
\end{table}
\end{remark}

\section{Two families of self-orthogonal codes with locality $2$}\label{sec5}
In this section, we construct two families of self-orthogonal codes with locality $2$.
\subsection{The first family of self-orthogonal codes with locality $2$}\label{sec5.1}

In this subsection, let $q$ be an odd prime power. Let $m, m_1, m_2$ be three positive integers with $m_2 \mid m_1 \mid m$ and $m \geq 3m_1$. Denote by the defining set $D=\{x \in \gf_{q^m}: \tr_{q^m/q^{m_1}}(x)=0\}$. Define a family of linear codes over $\gf_{q^{m_2}}$ by
$$\cC_D=\left\{\left(\tr_{q^m/q^{m_2}}(bx^2)\right)_{x \in D}:b \in \gf_{q^m}\right\}.$$
Its augmented code is given by
$$\overline{\cC_D}=\left\{\left(\tr_{q^m/q^{m_2}}(bx^2)\right)_{x \in D}+c\mathbf{1}:b \in \gf_{q^m}, c \in \gf_{q^{m_2}}\right\}, $$
where $\mathbf{1}$ is all-$1$ vector of length $| D |$. From now on, denote by $\chi_1$, $\lambda_1$ and $\phi_1$ the canonical additive characters of $\gf_{q^m}$, $\gf_{q^{m_1}}$ and $\gf_{q^{m_2}}$, respectively. Let $\eta$, $\eta'$ and $\eta_0$ respectively denote the quadratic multiplicative characters of $\gf_{q^m}$, $\gf_{q^{m_1}}$ and $\gf_{q^{m_2}}$. In the following, we will first prove that the augmented code $\overline{\cC_D}$ is self-orthogonal and then determine the locality of $\overline{\cC_D}$.

\begin{lemma}\label{lem-N}
Let $q$ be an odd prime power. Let $m, m_1, m_2$ be three positive integers with $m_2 \mid m_1 \mid m$ and $m \geq 3m_1$. Let $\gf_{q^m}^*=\langle\alpha\rangle$. Denote by $N=| \{b \in \langle\alpha^2\rangle: \tr_{q^m/q^{m_1}}(b)=a\} |$, where $a \in \gf_{q^{m_1}}$. When $\frac{m}{m_1}$ is odd,
\begin{eqnarray*}
  N = \begin{cases}
        \frac{q^{m-m_1}-1}{2} & \mbox{if $a=0$,} \\
        \frac{q^m+G(\eta,\chi_1)G(\eta',\lambda_1)\eta'(-a)}{2q^{m_1}} & \mbox{if $a \in \gf_{q^{m_1}}^*$}.
      \end{cases}
\end{eqnarray*}
When $\frac{m}{m_1}$ is even,
\begin{eqnarray*}
  N = \begin{cases}
          \frac{q^m-q^{m_1}+(q^{m_1}-1)G(\eta,\chi_1)}{2q^{m_1}} & \mbox{if $a=0$,} \\
          \frac{q^m-G(\eta,\chi_1)}{2q^{m_1}} & \mbox{if $a \in \gf_{q^{m_1}}^*$.}
        \end{cases}
\end{eqnarray*}
\end{lemma}

\begin{proof}
  By the orthogonal relation of additive characters and Lemma \ref{weilsum}, we have
  \begin{eqnarray*}
    N &=& \frac{1}{q^{m_1}}\sum_{y \in \gf_{q^{m_1}}}\sum_{i=0}^{\frac{q^m-3}{2}}\zeta_p^{\tr_{q^{m_1}/p}(y(\tr_{q^m/q^{m_1}}(\alpha^{2i})-a))}\\
    &=& \frac{1}{2q^{m_1}}\sum_{y \in \gf_{q^{m_1}}}\sum_{x \in \gf_{q^m}^*}\chi_1(yx^2)\lambda_1(-ya)\\
    &=& \frac{1}{2q^{m_1}}\sum_{y \in \gf_{q^{m_1}}^*}\lambda_1(-ya)\sum_{x \in \gf_{q^m}}\chi_1(yx^2)-\frac{1}{2q^{m_1}}\sum_{y \in \gf_{q^{m_1}}^*}\lambda_1(-ya)+\frac{q^m-1}{2q^{m_1}}\\
    &=& \begin{cases}
          \frac{1}{2q^{m_1}}\sum\limits_{y \in \gf_{q^{m_1}}^*}\sum\limits_{x \in \gf_{q^m}}\chi_1(yx^2)+\frac{q^{m_1-1}-1}{2} & \mbox{if $a=0$} \\
          \frac{1}{2q^{m_1}}\sum\limits_{y \in \gf_{q^{m_1}}^*}\lambda_1(-ya)\sum\limits_{x \in \gf_{q^m}}\chi_1(yx^2)+\frac{q^{m-m_1}}{2} & \mbox{if $a \neq 0$}
        \end{cases}\\
    &=& \begin{cases}
          \frac{G(\eta,\chi_1)}{2q^{m_1}}\sum\limits_{y \in \gf_{q^{m_1}}^*}\eta(y)+\frac{q^{m_1-1}-1}{2} & \mbox{if $a=0$} \\
          \frac{G(\eta,\chi_1)}{2q^{m_1}}\sum\limits_{y \in \gf_{q^{m_1}}^*}\lambda_1(-ya)\eta(y)+\frac{q^{m-m_1}}{2} & \mbox{if $a \neq 0$}
        \end{cases}\\
    &=& \begin{cases}
          \frac{G(\eta,\chi_1)}{2q^{m_1}}\sum\limits_{y \in \gf_{q^{m_1}}^*}\eta'(y)+\frac{q^{m_1-1}-1}{2} & \mbox{if $a=0$ and $\frac{m}{m_1}$ is odd} \\
          \frac{(q^{m_1}-1)G(\eta,\chi_1)}{2q^{m_1}}+\frac{q^{m-m_1}-1}{2} & \mbox{if $a=0$ and $\frac{m}{m_1}$ is even} \\
          \frac{G(\eta,\chi_1)}{2q^{m_1}}\sum\limits_{y \in \gf_{q^{m_1}}^*}\lambda_1(-ya)\eta'(-ya)\eta'(-a)+\frac{q^{m-m_1}}{2} & \mbox{if $a \neq 0$ and $\frac{m}{m_1}$ is odd}\\
          \frac{G(\eta,\chi_1)}{2q^{m_1}}\sum\limits_{y \in \gf_{q^{m_1}}^*}\lambda_1(-ya)+\frac{q^{m-m_1}}{2} & \mbox{if $a \neq 0$ and $\frac{m}{m_1}$ is odd}
        \end{cases}\\
    &=& \begin{cases}
          \frac{q^{m_1-1}-1}{2} & \mbox{if $a=0$ and $\frac{m}{m_1}$ is odd,} \\
          \frac{q^m-q^{m_1}+(q^{m_1}-1)G(\eta,\chi_1)}{2q^{m_1}} & \mbox{if $a=0$ and $\frac{m}{m_1}$ is even,} \\
          \frac{q^m+G(\eta,\chi_1)G(\eta',\lambda_1)\eta'(-a)}{2q^{m_1}} & \mbox{if $a \neq 0$ and $\frac{m}{m_1}$ is odd,}\\
          \frac{q^m-G(\eta,\chi_1)}{2q^{m_1}} & \mbox{if $a \neq 0$ and $\frac{m}{m_1}$ is even.}
        \end{cases}
  \end{eqnarray*}

Then the desired conclusions follow.
\end{proof}

\begin{lemma}\label{lem-Omega}
   Let $q$ be an odd prime power. Let $m, m_1, m_2$ be three positive integers with $m_2 \mid m_1 \mid m$ and $m \geq 3m_1$. Denote by
   \begin{eqnarray*}
     \Omega(b,c) = \sum_{z \in \gf_{q^{m_2}}^*}\sum_{x \in \gf_{q^m}}\chi_1(zbx^2)\phi_1(zc),
   \end{eqnarray*}
   where $b \in \gf_{q^m}^*$ and $c \in \gf_{q^{m_2}}$.
   If $\frac{m}{m_2}$ is odd, then
   \begin{eqnarray*}
     \Omega(b,c) = \begin{cases}
                     0 & \mbox{if $c=0$,}  \\
                     G(\eta,\chi_1)G(\eta_0,\phi_1)\eta(b)\eta_0(c) & \mbox{if $c \neq 0$.}
                   \end{cases}
   \end{eqnarray*}
   If $\frac{m}{m_2}$ is even, then
   \begin{eqnarray*}
     \Omega(b,c) = \begin{cases}
                     (q^{m_2}-1)G(\eta,\chi_1)\eta(b) & \mbox{if $c=0$,}  \\
                     -G(\eta,\chi_1)\eta(b) & \mbox{if $c \neq 0$.}
                   \end{cases}
   \end{eqnarray*}
\end{lemma}

\begin{proof}
  By the orthogonal relation of additive characters and Lemma \ref{weilsum}, we have
  \begin{eqnarray*}
    \Omega(b,c) &=& \sum_{z \in \gf_{q^{m_2}}^*}\eta(zb)G(\eta,\chi_1)\phi_1(zc) \\
     &=& G(\eta,\chi_1)\eta(b)\sum_{z \in \gf_{q^{m_2}}^*}\eta(z)\phi_1(zc) \\
     &=& \begin{cases}
           G(\eta,\chi_1)\eta(b)\sum_{z \in \gf_{q^{m_2}}^*}\eta_0(z)\phi_1(zc) & \mbox{if $\frac{m}{m_2}$ is odd}  \\
           G(\eta,\chi_1)\eta(b)\sum_{z \in \gf_{q^{m_2}}^*}\phi_1(zc), & \mbox{if $\frac{m}{m_2}$ is even}
         \end{cases}\\
     &=& \begin{cases}
           0 & \mbox{if $\frac{m}{m_2}$ is odd and $c=0$,}  \\
           G(\eta,\chi_1)G(\eta_0,\phi_1)\eta(b)\eta_0(c) & \mbox{if $\frac{m}{m_2}$ is odd and $c\neq0$,}  \\
           (q^{m_2}-1)G(\eta,\chi_1)\eta(b), & \mbox{if $\frac{m}{m_2}$ is even and $c=0$,}\\
           -G(\eta,\chi_1)\eta(b), & \mbox{if $\frac{m}{m_2}$ is even and $c\neq0$.}
         \end{cases}
  \end{eqnarray*}
\end{proof}

\begin{lemma}\label{lem-Delta}
  Let $q$ be an odd prime power. Let $m, m_1, m_2$ be three positive integers with $m_2 \mid m_1 \mid m$ and $m \geq 3m_1$. Denote by
  \begin{equation*}
    \Delta(b,c)=\sum_{x \in \gf_{q^m}}\sum_{y \in \gf_{q^{m_1}}^*}\sum_{z \in \gf_{q^{m_2}}^*}\chi_1(zbx^2+yx)\phi_1(zc),
  \end{equation*}
  where $b \in \gf_{q^m}^*$ and $c \in \gf_{q^{m_2}}$.
  If $\frac{m}{m_2}$ is odd and $\frac{m_1}{m_2}$ is odd, then
  \begin{eqnarray*}
    \Delta(b,c) = \left\{\begin{array}{ll}
                   0 \qquad\qquad\qquad\qquad\qquad\qquad \text{if $\tr_{q^m/q^{m_1}}(b^{-1})=0$ and $c=0$,}\\
                   (q^{m_1}-1)G(\eta,\chi_1)G(\eta_0, \phi_1)\eta(b)\eta_0(c)\\ \qquad\qquad\qquad\qquad\qquad\qquad \;\text{if $\tr_{q^m/q^{m_1}}(b^{-1})=0$ and $c\neq0$,}\\
                   (q^{m_2}-1)G(\eta,\chi_1)G(\eta',\lambda_1)\eta(b)\eta'(-\tr_{q^m/q^{m_1}}(b^{-1}))\\ \qquad\qquad\qquad\qquad\qquad\qquad \;\text{if $\tr_{q^m/q^{m_1}}(b^{-1})\neq0$ and $c=0$,}\\
                   -G(\eta,\chi_1)G(\eta',\lambda_1)\eta(b)\eta'(-\tr_{q^m/q^{m_1}}(b^{-1}))-G(\eta,\chi_1)G(\eta_0,\phi_1)\eta(b)\eta_0(c)\\ \qquad\qquad\qquad\qquad\qquad\qquad \;\text{if $\tr_{q^m/q^{m_1}}(b^{-1})\neq 0$ and $c\neq 0$.}
                  \end{array}\right.
  \end{eqnarray*}
  If $\frac{m}{m_2}$ is even and $\frac{m_1}{m_2}$ is odd, then
  \begin{eqnarray*}
    \Delta(b,c) = \left\{\begin{array}{ll}
                   (q^{m_1}-1)(q^{m_2}-1)G(\eta,\chi_1)\eta(b) \qquad\qquad \text{if $\tr_{q^m/q^{m_1}}(b^{-1})=0$ and $c=0$,}\\
                   -(q^{m_1}-1)G(\eta,\chi_1)\eta(b)  \qquad\qquad\qquad \quad\text{if $\tr_{q^m/q^{m_1}}(b^{-1})=0$ and $c\neq0$,}\\
                   -(q^{m_2}-1)G(\eta,\chi_1)\eta(b)  \qquad\qquad\qquad \quad\text{if $\tr_{q^m/q^{m_1}}(b^{-1})\neq0$ and $c=0$,}\\
                   G(\eta,\chi_1)G(\eta',\lambda_1)G(\phi_1, \eta_0)\eta(b)\eta'(-\tr_{q^m/q^{m_1}}(b^{-1}))\eta_0(c)+G(\eta,\chi_1)\eta(b)\\ \qquad\qquad\qquad\qquad\qquad\qquad\qquad\qquad \quad\text{if $\tr_{q^m/q^{m_1}}(b^{-1})\neq 0$ and $c\neq 0$.}
                  \end{array}\right.
  \end{eqnarray*}
  If $\frac{m}{m_2}$ is even and $\frac{m_1}{m_2}$ is even, then
  \begin{eqnarray*}
    \Delta(b,c) = \left\{\begin{array}{ll}
                   (q^{m_1}-1)(q^{m_2}-1)G(\eta,\chi_1)\eta(b) \qquad\qquad \text{if $\tr_{q^m/q^{m_1}}(b^{-1})=0$ and $c=0$,}\\
                   -(q^{m_1}-1)G(\eta,\chi_1)\eta(b)  \qquad\qquad\qquad \quad\text{if $\tr_{q^m/q^{m_1}}(b^{-1})=0$ and $c\neq0$,}\\
                   (q^{m_2}-1)G(\eta,\chi_1)\eta(b)\left(G(\eta',\lambda_1)\eta'(-\tr_{q^m/q^{m_1}}(b^{-1}))-1\right)  \\ \qquad\qquad\qquad\qquad\qquad\qquad\qquad\qquad \quad\text{if $\tr_{q^m/q^{m_1}}(b^{-1})\neq0$ and $c=0$,}\\
                   -G(\eta,\chi_1)G(\eta',\lambda_1)\eta(b)\eta'(-\tr_{q^m/q^{m_1}}(b^{-1}))+G(\eta,\chi_1)\eta(b)\\ \qquad\qquad\qquad\qquad\qquad\qquad\qquad\qquad \quad\text{if $\tr_{q^m/q^{m_1}}(b^{-1})\neq 0$ and $c\neq 0$.}
                  \end{array}\right.
  \end{eqnarray*}
\end{lemma}

\begin{proof}
  By the orthogonal relation of additive characters and Lemma \ref{weilsum}, we have
  \begin{eqnarray*}
    \Delta(b, c) &=& \sum_{y \in \gf_{q^{m_1}}^*}\sum_{z \in \gf_{q^{m_2}}^*}\chi_1(-\frac{y^2}{4zb})\eta(zb)G(\eta,\chi_1)\phi_1(zc) \\
     &=& \sum_{y \in \gf_{q^{m_1}}^*}\sum_{z \in \gf_{q^{m_2}}^*}\lambda_1(-\frac{y^2}{4z}\tr_{q^m/q^{m_1}}(b^{-1}))\eta(zb)G(\eta,\chi_1)\phi_1(zc)\\
     &=& G(\eta,\chi_1)\eta(b)\sum_{y \in \gf_{q^{m_1}}^*}\sum_{z \in \gf_{q^{m_2}}^*}\lambda_1(-\frac{y^2}{4z}\tr_{q^m/q^{m_1}}(b^{-1}))\eta(z)\phi_1(zc).
  \end{eqnarray*}
  To calculate the value of $\Delta(b, c)$, we consider the following four cases:

  {Case 1:} Let $\tr_{q^m/q^{m_1}}(b^{-1})=0$ and $c=0$.
   \begin{eqnarray*}
    \Delta(b, c) &=& G(\eta,\chi_1)\eta(b)\sum_{y \in \gf_{q^{m_1}}^*}\sum_{z \in \gf_{q^{m_2}}^*}\eta(z)\\
    &=& \begin{cases}
          G(\eta,\chi_1)\eta(b)\sum\limits_{y \in \gf_{q^{m_1}}^*}\sum\limits_{z \in \gf_{q^{m_2}}^*}\eta_0(z) & \mbox{if $\frac{m}{m_2}$ is odd} \\
          (q^{m_1}-1)(q^{m_2}-1)G(\eta,\chi_1)\eta(b) & \mbox{if $\frac{m}{m_2}$ is even}
        \end{cases}\\
    &=& \begin{cases}
          0 & \mbox{if $\frac{m}{m_2}$ is odd,} \\
          (q^{m_1}-1)(q^{m_2}-1)G(\eta,\chi_1)\eta(b) & \mbox{if $\frac{m}{m_2}$ is even}.
        \end{cases}
  \end{eqnarray*}

  {Case 2:} Let $\tr_{q^m/q^{m_1}}(b^{-1})=0$ and $c\neq0$.
   \begin{eqnarray*}
    \Delta(b, c) &=& G(\eta,\chi_1)\eta(b)\sum_{y \in \gf_{q^{m_1}}^*}\sum_{z \in \gf_{q^{m_2}}^*}\eta(z)\phi_1(zc)\\
    &=& \begin{cases}
          G(\eta,\chi_1)\eta(b)\sum\limits_{y \in \gf_{q^{m_1}}^*}\sum\limits_{z \in \gf_{q^{m_2}}^*}\eta_0(z)\phi_1(zc) & \mbox{if $\frac{m}{m_2}$ is odd} \\
          G(\eta,\chi_1)\eta(b)\sum\limits_{y \in \gf_{q^{m_1}}^*}\sum\limits_{z \in \gf_{q^{m_2}}^*}\phi_1(zc) & \mbox{if $\frac{m}{m_2}$ is even}
        \end{cases}\\
    &=& \begin{cases}
          G(\eta,\chi_1)\eta(b)\eta_0(c)\sum\limits_{y \in \gf_{q^{m_1}}^*}\sum\limits_{z \in \gf_{q^{m_2}}^*}\eta_0(zc)\phi_1(zc) & \mbox{if $\frac{m}{m_2}$ is odd} \\
          -(q^{m_1}-1)G(\eta,\chi_1)\eta(b) & \mbox{if $\frac{m}{m_2}$ is even}
        \end{cases}\\
    &=& \begin{cases}
          (q^{m_1}-1)G(\eta,\chi_1)G(\eta_0, \phi_1)\eta(b)\eta_0(c) & \mbox{if $\frac{m}{m_2}$ is odd,} \\
          -(q^{m_1}-1)G(\eta,\chi_1)\eta(b) & \mbox{if $\frac{m}{m_2}$ is even}.
        \end{cases}
  \end{eqnarray*}

  {Case 3:} Let $\tr_{q^m/q^{m_1}}(b^{-1}) \neq 0$ and $c=0$.
   \begin{eqnarray*}
    \Delta(b, c) &=& G(\eta,\chi_1)\eta(b)\sum_{y \in \gf_{q^{m_1}}^*}\sum_{z \in \gf_{q^{m_2}}^*}\lambda_1(-\frac{y^2}{4z}\tr_{q^m/q^{m_1}}(b^{-1}))\eta(z)\\
    &=& G(\eta,\chi_1)\eta(b)\sum_{z \in \gf_{q^{m_2}}^*}\left(\sum_{y \in \gf_{q^{m_1}}}\lambda_1(-y^2z\tr_{q^m/q^{m_1}}(b^{-1}))\eta(z)-\eta(z)\right)\\
    &=& G(\eta,\chi_1)\eta(b)\sum_{z \in \gf_{q^{m_2}}^*}\left(G(\eta',\lambda_1)\eta'(-z\tr_{q^m/q^{m_1}}(b^{-1}))\eta(z)-\eta(z)\right)\\
    &=& \begin{cases}
          G(\eta,\chi_1)\eta(b)G(\eta',\lambda_1)\sum\limits_{z \in \gf_{q^{m_2}}^*}\eta'(-z\tr_{q^m/q^{m_1}}(b^{-1}))\eta_0(z)-0 & \mbox{if $\frac{m}{m_2}$ is odd} \\
          G(\eta,\chi_1)\eta(b)\sum\limits_{z \in \gf_{q^{m_2}}^*}\left(G(\eta',\lambda_1)\eta'(-z\tr_{q^m/q^{m_1}}(b^{-1}))-1\right) & \mbox{if $\frac{m}{m_2}$ is even}
        \end{cases}\\
    &=& \left\{\begin{array}{ll}
          (q^{m_2}-1)G(\eta,\chi_1)G(\eta',\lambda_1)\eta(b)\eta'(-\tr_{q^m/q^{m_1}}(b^{-1})) \\ \qquad\qquad\qquad\qquad\qquad\qquad\qquad\qquad \qquad\text{if $\frac{m}{m_2}$ is odd and $\frac{m_1}{m_2}$ is odd} \\
          G(\eta,\chi_1)\eta(b)\sum\limits_{z \in \gf_{q^{m_2}}^*}\left(G(\eta',\lambda_1)\eta_0(z)\eta'(-\tr_{q^m/q^{m_1}}(b^{-1}))-1\right) \\ \qquad\qquad\qquad\qquad\qquad\qquad\qquad\qquad\qquad \text{if $\frac{m}{m_2}$ is even and $\frac{m_1}{m_2}$ is odd}\\
          G(\eta,\chi_1)\eta(b)\sum\limits_{z \in \gf_{q^{m_2}}^*}\left(G(\eta',\lambda_1)\eta'(-\tr_{q^m/q^{m_1}}(b^{-1}))-1\right) \\ \qquad\qquad\qquad\qquad\qquad\qquad\qquad\qquad\qquad \text{if $\frac{m}{m_2}$ is even and $\frac{m_1}{m_2}$ is even}
        \end{array}\right.\\
    &=& \left\{\begin{array}{ll}
          (q^{m_2}-1)G(\eta,\chi_1)G(\eta',\lambda_1)\eta(b)\eta'(-\tr_{q^m/q^{m_1}}(b^{-1})) \\ \qquad\qquad\qquad\qquad\qquad\qquad\qquad\qquad \qquad\text{if $\frac{m}{m_2}$ is odd and $\frac{m_1}{m_2}$ is odd,} \\
          -(q^{m_2}-1)G(\eta,\chi_1)\eta(b) \qquad\qquad\qquad\qquad \;\text{if $\frac{m}{m_2}$ is even and $\frac{m_1}{m_2}$ is odd},\\
          (q^{m_2}-1)G(\eta,\chi_1)\eta(b)\left(G(\eta',\lambda_1)\eta'(-\tr_{q^m/q^{m_1}}(b^{-1}))-1\right) \\ \qquad\qquad\qquad\qquad\qquad\qquad\qquad\qquad \qquad\text{if $\frac{m}{m_2}$ is even and $\frac{m_1}{m_2}$ is even}.
         \end{array}\right.
  \end{eqnarray*}

  {Case 4:} Let $\tr_{q^m/q^{m_1}}(b^{-1}) \neq 0$ and $c\neq0$.
  \begin{eqnarray*}
    \Delta(b, c) &=& G(\eta,\chi_1)\eta(b)\sum_{y \in \gf_{q^{m_1}}^*}\sum_{z \in \gf_{q^{m_2}}^*}\lambda_1(-\frac{y^2}{4z}\tr_{q^m/q^{m_1}}(b^{-1}))\eta(z)\phi_1(zc)\\
    &=& G(\eta,\chi_1)\eta(b)\sum_{z \in \gf_{q^{m_2}}^*}\phi_1(zc)\left(\sum_{y \in \gf_{q^{m_1}}}\lambda_1(-\frac{y^2}{4z}\tr_{q^m/q^{m_1}}(b^{-1}))\eta(z)-\eta(z)\right)\\
    &=& G(\eta,\chi_1)\eta(b)\sum_{z \in \gf_{q^{m_2}}^*}\phi_1(zc)\left(G(\eta',\lambda_1)\eta'(-z\tr_{q^m/q^{m_1}}(b^{-1}))\eta(z)-\eta(z)\right)\\
    &=& G(\eta,\chi_1)G(\eta',\lambda_1)\eta(b)\eta'(-\tr_{q^m/q^{m_1}}(b^{-1}))\sum_{z \in \gf_{q^{m_2}}^*}\phi_1(zc)\eta'(z)\eta(z)\\&&-G(\eta,\chi_1)\eta(b)\sum_{z \in \gf_{q^{m_2}}^*}\phi_1(zc)\eta(z)\\
    &=& \left\{\begin{array}{ll}
         G(\eta,\chi_1)G(\eta',\lambda_1)\eta(b)\eta'(-\tr_{q^m/q^{m_1}}(b^{-1}))\sum\limits_{z \in \gf_{q^{m_2}}^*}\phi_1(zc)\eta_0(z)\eta_0(z)\\-G(\eta,\chi_1)\eta(b)\sum\limits_{z \in \gf_{q^{m_2}}^*}\phi_1(zc)\eta_0(z)\\
           \qquad\qquad\qquad\qquad \qquad\qquad\qquad\qquad\qquad\text{if $\frac{m}{m_2}$ is odd and $\frac{m_1}{m_2}$ is odd} \\
         G(\eta,\chi_1)G(\eta',\lambda_1)\eta(b)\eta'(-\tr_{q^m/q^{m_1}}(b^{-1}))\sum\limits_{z \in \gf_{q^{m_2}}^*}\phi_1(zc)\eta_0(z)\\-G(\eta,\chi_1)\eta(b)\sum\limits_{z \in \gf_{q^{m_2}}^*}\phi_1(zc)\\
           \qquad\qquad\qquad\qquad \qquad\qquad\qquad\qquad\qquad\text{if $\frac{m}{m_2}$ is even and $\frac{m_1}{m_2}$ is odd} \\
           G(\eta,\chi_1)G(\eta',\lambda_1)\eta(b)\eta'(-\tr_{q^m/q^{m_1}}(b^{-1}))\sum\limits_{z \in \gf_{q^{m_2}}^*}\phi_1(zc)\\-G(\eta,\chi_1)\eta(b)\sum\limits_{z \in \gf_{q^{m_2}}^*}\phi_1(zc)\\
           \qquad\qquad\qquad\qquad \qquad\qquad\qquad\qquad\qquad\text{if $\frac{m}{m_2}$ is even and $\frac{m_1}{m_2}$ is even}
        \end{array}\right.\\
    &=& \left\{\begin{array}{ll}
         -G(\eta,\chi_1)G(\eta',\lambda_1)\eta(b)\eta'(-\tr_{q^m/q^{m_1}}(b^{-1}))-G(\eta,\chi_1)G(\eta_0,\phi_1)\eta(b)\eta_0(c)\\
           \qquad\qquad\qquad\qquad \qquad\qquad\qquad\qquad\qquad\text{if $\frac{m}{m_2}$ is odd and $\frac{m_1}{m_2}$ is odd,} \\
         G(\eta,\chi_1)G(\eta',\lambda_1)G(\phi_1, \eta_0)\eta(b)\eta'(-\tr_{q^m/q^{m_1}}(b^{-1}))\eta_0(c)+G(\eta,\chi_1)\eta(b)\\
           \qquad\qquad\qquad\qquad \qquad\qquad\qquad\qquad\qquad\text{if $\frac{m}{m_2}$ is even and $\frac{m}{m_2}$ is odd,} \\
           -G(\eta,\chi_1)G(\eta',\lambda_1)\eta(b)\eta'(-\tr_{q^m/q^{m_1}}(b^{-1}))+G(\eta,\chi_1)\eta(b)\\
           \qquad\qquad\qquad\qquad \qquad\qquad\qquad\qquad\qquad\text{if $\frac{m}{m_2}$ is even and $\frac{m}{m_2}$ is even.}
        \end{array}\right.
  \end{eqnarray*}

  Summarizing the above discussions, the desire conclusions follow.
\end{proof}

\begin{theorem}\label{th-5.1}
  Let $q$ be an odd prime power. Let $m, m_1, m_2$ be three positive integers with $m_2 \mid m_1 \mid m$ and $m \geq 3m_1$. In particular, let $m > 3m_1$ when $\frac{m}{m_2}$ is odd and $\frac{m_1}{m_2}$ is odd. Denote by $l_1=\frac{(p-1)(m+m_2)e}{4}$, $l_2=\frac{(p-1)(m+3m_1)e}{4}$, $l_3=\frac{(p-1)em}{4}+1$, $l_4=\frac{(p-1)(m+3m_1+m_2)e}{4}+1$ and $l_5=\frac{(p-1)(m+m_1)e}{4}$. Then $\overline{\cC_D}$ is a self-orthogonal $q$-divisible code with parameters $[q^{m-m_1}, \frac{m}{m_2}+1, d]$, where
   \begin{eqnarray*}
     d &=& \begin{cases}
             q^{m-m_1}-q^{m-m_1-m_2}-q^{\frac{m-m_2}{2}}, & \mbox{if both $\frac{m}{m_2}$ and $\frac{m_1}{m_2}$ are odd,}  \\
             q^{m-m_1}-q^{m-m_1-m_2}-(q^{m_2}-1)q^{\frac{m-2m_2}{2}}, & \mbox{otherwise}.
           \end{cases}
   \end{eqnarray*}
   Besides, its weight distributions are listed in Tables \ref{tab5.1}, \ref{tab5.2}, \ref{tab5.31} and \ref{tab5.32} for four cases, respectively.
\end{theorem}

\begin{table}[h]
\begin{center}
\caption{The weight distribution of $\overline{\cC_{D}}$ in Theorem \ref{th-5.1} ($\frac{m}{m_2}$ is odd and $\frac{m_1}{m_2}$ is odd).}\label{tab5.1}
\begin{tabular}{@{}ll@{}}
\toprule%
Weight & Frequency  \\
\midrule
$0$ & $1$\\
$q^{m-m_1}$ & $q^{m_2}-1$ \\
$q^{m-m_1}-q^{m-m_1-m_2}$ & $q^{m-m_1}-1$\\
$q^{m-m_1}-q^{m-m_1-m_2}-(-1)^{l_1}q^{\frac{m-m_2}{2}}$ & $\frac{(q^{m_2}-1)(q^{m-m_1}-1)}{2}$\\
$q^{m-m_1}-q^{m-m_1-m_2}+(-1)^{l_1}q^{\frac{m-m_2}{2}}$ & $\frac{(q^{m_2}-1)(q^{m-m_1}-1)}{2}$\\
$q^{m-m_1}-q^{m-m_1-m_2}-(-1)^{l_2}(q^{m_2}-1)q^{\frac{m-m_1-2m_2}{2}}$ & $\frac{q^{m-m_1}+(-1)^{l_2}(q^{m_1}-1)q^{\frac{m-m_1}{2}}}{2}$ \\
$q^{m-m_1}-q^{m-m_1-m_2}+(-1)^{l_2}(q^{m_2}-1)q^{\frac{m-m_1-2m_2}{2}}$ & $\frac{q^{m-m_1}-(-1)^{l_2}(q^{m_1}-1)q^{\frac{m-m_1}{2}}}{2}$ \\
$q^{m-m_1}-q^{m-m_1-m_2}+(-1)^{l_2}q^{\frac{m-m_1-2m_2}{2}}$ & $\frac{(q^{m_2}-1)(q^{m-m_1}+(-1)^{l_2}(q^{m_1}-1)q^{\frac{m-m_1}{2}})}{2}$\\
$q^{m-m_1}-q^{m-m_1-m_2}-(-1)^{l_2}q^{\frac{m-m_1-2m_2}{2}}$ & $\frac{(q^{m_2}-1)(q^{m-m_1}-(-1)^{l_2}(q^{m_1}-1)q^{\frac{m-m_1}{2}})}{2}$\\
\bottomrule
\end{tabular}
\end{center}
\end{table}

\begin{table}[h]
\begin{center}
\caption{The weight distribution of $\overline{\cC_{D}}$ in Theorem \ref{th-5.1} ($\frac{m}{m_2}$ is even and $\frac{m_1}{m_2}$ is odd).}\label{tab5.2}
\begin{tabular}{@{}ll@{}}
\toprule%
Weight & Frequency  \\
\midrule
$0$ & $1$\\
$q^{m-m_1}$ & $q^{m_2}-1$ \\
$q^{m-m_1}-q^{m-m_1-m_2}$ & $q^{m}-q^{m-m_1}$\\
$q^{m-m_1}-q^{m-m_1-m_2}-(-1)^{l_3}(q^{m_2}-1)q^{\frac{m-2m_2}{2}}$ & $\frac{q^{m-m_1}-1+(-1)^{l_3}(q^{m_1}-1)q^{\frac{m-2m_1}{2}}}{2}$\\
$q^{m-m_1}-q^{m-m_1-m_2}+(-1)^{l_3}(q^{m_2}-1)q^{\frac{m-2m_2}{2}}$ & $\frac{q^{m-m_1}-1-(-1)^{l_3}(q^{m_1}-1)q^{\frac{m-2m_1}{2}}}{2}$\\
$q^{m-m_1}-q^{m-m_1-m_2}+(-1)^{l_3}q^{\frac{m-2m_2}{2}}$ & $\frac{(q^{m_2}-1)(q^{m-m_1}-1+(-1)^{l_3}(q^{m_1}-1)q^{\frac{m-2m_1}{2}})}{2}$\\
$q^{m-m_1}-q^{m-m_1-m_2}-(-1)^{l_3}q^{\frac{m-2m_2}{2}}$ & $\frac{(q^{m_2}-1)(q^{m-m_1}-1-(-1)^{l_3}(q^{m_1}-1)q^{\frac{m-2m_1}{2}})}{2}$\\
$q^{m-m_1}-q^{m-m_1-m_2}+(-1)^{l_4}q^{\frac{m-m_1-m_2}{2}}$ & $\frac{q^{m-m_1}(q^{m_2}-1)(q^{m_1}-1)}{2}$\\
$q^{m-m_1}-q^{m-m_1-m_2}-(-1)^{l_4}q^{\frac{m-m_1-m_2}{2}}$ & $\frac{q^{m-m_1}(q^{m_2}-1)(q^{m_1}-1)}{2}$\\
\bottomrule
\end{tabular}
\end{center}
\end{table}

\begin{table}[h]
\begin{center}
\caption{The weight distribution of $\overline{\cC_{D}}$ in Theorem \ref{th-5.1} ($\frac{m}{m_2}$ is even, $\frac{m_1}{m_2}$ is even and $\frac{m}{m_1}$ is odd ).}\label{tab5.31}
\begin{tabular}{@{}ll@{}}
\toprule%
Weight & Frequency  \\
\midrule
$0$ & $1$\\
$q^{m-m_1}$ & $q^{m_2}-1$ \\
$q^{m-m_1}-q^{m-m_1-m_2}-(-1)^{l_3}(q^{m_2}-1)q^{\frac{m-2m_2}{2}}$ & $\frac{q^{m-m_1}-1}{2}$\\
$q^{m-m_1}-q^{m-m_1-m_2}+(-1)^{l_3}(q^{m_2}-1)q^{\frac{m-2m_2}{2}}$ & $\frac{q^{m-m_1}-1}{2}$\\
$q^{m-m_1}-q^{m-m_1-m_2}+(-1)^{l_3}q^{\frac{m-2m_2}{2}}$ & $\frac{(q^{m-m_1}-1)(q^{m_2}-1)}{2}$\\
$q^{m-m_1}-q^{m-m_1-m_2}-(-1)^{l_3}q^{\frac{m-2m_2}{2}}$ & $\frac{(q^{m-m_1}-1)(q^{m_2}-1)}{2}$\\
$q^{m-m_1}-q^{m-m_1-m_2}-(-1)^{l_5}(q^{m_2}-1)q^{\frac{m-m_1-2m_2}{2}}$ & $\frac{(q^{m_1}-1)(q^{m-m_1}+(-1)^{l_5}q^{\frac{m-m_1}{2}})}{2}$\\
$q^{m-m_1}-q^{m-m_1-m_2}+(-1)^{l_5}(q^{m_2}-1)q^{\frac{m-m_1-2m_2}{2}}$ & $\frac{(q^{m_1}-1)(q^{m-m_1}-(-1)^{l_5}q^{\frac{m-m_1}{2}})}{2}$\\
$q^{m-m_1}-q^{m-m_1-m_2}+(-1)^{l_5}q^{\frac{m-m_1-2m_2}{2}}$ & $\frac{(q^{m_2}-1)(q^{m_1}-1)(q^{m-m_1}+(-1)^{l_5}q^{\frac{m-m_1}{2}})}{2}$\\
$q^{m-m_1}-q^{m-m_1-m_2}-(-1)^{l_5}q^{\frac{m-m_1-2m_2}{2}}$ & $\frac{(q^{m_2}-1)(q^{m_1}-1)(q^{m-m_1}-(-1)^{l_5}q^{\frac{m-m_1}{2}})}{2}$\\
\bottomrule
\end{tabular}
\end{center}
\end{table}

\begin{table}[h]
\begin{center}
\caption{The weight distribution of $\overline{\cC_{D}}$ in Theorem \ref{th-5.1} ($\frac{m}{m_2}$ is even, $\frac{m_1}{m_2}$ is even and $\frac{m}{m_1}$ is even ).}\label{tab5.32}
\begin{tabular}{@{}ll@{}}
\toprule%
Weight & Frequency  \\
\midrule
$0$ & $1$\\
$q^{m-m_1}$ & $q^{m_2}-1$ \\
$q^{m-m_1}-q^{m-m_1-m_2}-(-1)^{l_3}(q^{m_2}-1)q^{\frac{m-2m_2}{2}}$ & $\frac{q^{m-m_1}-1+(-1)^{l_3}(q^{m_1}-1)q^{\frac{m-2m_1}{2}}}{2}$\\
$q^{m-m_1}-q^{m-m_1-m_2}+(-1)^{l_3}(q^{m_2}-1)q^{\frac{m-2m_2}{2}}$ & $\frac{q^{m-m_1}-1-(-1)^{l_3}(q^{m_1}-1)q^{\frac{m-2m_1}{2}}}{2}$\\
$q^{m-m_1}-q^{m-m_1-m_2}+(-1)^{l_3}q^{\frac{m-2m_2}{2}}$ & $\frac{(q^{m_2}-1)(q^{m-m_1}-1+(-1)^{l_3}(q^{m_1}-1)q^{\frac{m-2m_1}{2}})}{2}$\\
$q^{m-m_1}-q^{m-m_1-m_2}-(-1)^{l_3}q^{\frac{m-2m_2}{2}}$ & $\frac{(q^{m_2}-1)(q^{m-m_1}-1-(-1)^{l_3}(q^{m_1}-1)q^{\frac{m-2m_1}{2}})}{2}$\\
$q^{m-m_1}-q^{m-m_1-m_2}-(-1)^{l_5}(q^{m_2}-1)q^{\frac{m-m_1-2m_2}{2}}$ & $\frac{q^{m-m_1}(q^{m_1}-1)}{2}$\\
$q^{m-m_1}-q^{m-m_1-m_2}+(-1)^{l_5}(q^{m_2}-1)q^{\frac{m-m_1-2m_2}{2}}$ & $\frac{q^{m-m_1}(q^{m_1}-1)}{2}$\\
$q^{m-m_1}-q^{m-m_1-m_2}+(-1)^{l_5}q^{\frac{m-m_1-2m_2}{2}}$ & $\frac{q^{m-m_1}(q^{m_1}-1)(q^{m_2}-1)}{2}$\\
$q^{m-m_1}-q^{m-m_1-m_2}-(-1)^{l_5}q^{\frac{m-m_1-2m_2}{2}}$ & $\frac{q^{m-m_1}(q^{m_1}-1)(q^{m_2}-1)}{2}$\\
\bottomrule
\end{tabular}
\end{center}
\end{table}

\begin{proof}
  It is obviously that $n=| D |=q^{m-m_1}$. For any codeword $\bc_{(b,c)}=(\tr_{q^m/q^{m_2}}(bx^2)+c)_{x \in D} \in \overline{\cC_{D}}$, by the orthogonal relation of additive characters, the Hamming weight of $\bc_{(b,c)}$ is
  \begin{eqnarray}\label{eq-wt}
    \text{wt}(\bc_{(b,c)}) = \begin{cases}
                                 0 & \mbox{if $b=0$ and $c=0$,} \\
                                 q^{m-m_1} & \mbox{if $b=0$ and $c \neq 0$,} \\
                                 q^{m-m_1}-\frac{1}{q^{m_1+m_2}}\sum\limits_{x \in \gf_{q^m}}\sum\limits_{y \in \gf_{q^{m_1}}}\sum\limits_{z \in \gf_{q^{m_2}}}\chi_1(zbx^2+yx)\phi_1(zc) & \mbox{if $b \neq 0$}.
                               \end{cases}
  \end{eqnarray}

  Denote by $S=\sum\limits_{x \in \gf_{q^m}}\sum\limits_{y \in \gf_{q^{m_1}}}\sum\limits_{z \in \gf_{q^{m_2}}}\chi_1(zbx^2+yx)\phi_1(zc)$.
  Then by the orthogonal relation of additive characters,
  \begin{eqnarray*}
    S &=& q^m+\sum\limits_{y \in \gf_{q^{m_1}}^*}\sum\limits_{x \in \gf_{q^m}}\chi_1(yx)+\sum\limits_{z \in \gf_{q^{m_2}}^*}\sum\limits_{x \in \gf_{q^m}}\chi_1(zbx^2)\phi_1(zc)\\
    & &+\sum\limits_{y \in \gf_{q^{m_1}}^*}\sum\limits_{z \in \gf_{q^{m_2}}^*}\sum\limits_{x \in \gf_{q^m}}\chi_1(zbx^2+yx)\phi_1(zc)\\
    &=& q^m+\Omega(b,c)+\Delta(b,c),
  \end{eqnarray*}
  where $\Omega(b,c)$ and $\Delta(b,c)$ are defined in Lemmas \ref{lem-Omega} and \ref{lem-Delta}, respectively. By Lemmas \ref{lem-Omega} and \ref{lem-Delta}, the value of $S$ directly follows.
  Substituting the value of $S$ into Equation (\ref{eq-wt}) yields that
  \begin{eqnarray*}
    \text{wt}(\bc_{(b,c)}) &=& \left\{\begin{array}{lll}
                                0 \qquad\qquad \qquad\mbox{if $b=c=0$,}  \\
                                q^{m-m_1} \qquad \qquad\mbox{if $b=0$ and $c \neq 0$,}  \\
                                q^{m-m_1}-q^{m-m_1-m_2}\quad \mbox{if $b\neq0$, $c = 0$ and $\tr_{q^m/q^{m_1}}(b^{-1})=0$,}  \\
                                q^{m-m_1}-q^{m-m_1-m_2}-\frac{G(\eta,\chi_1)G(\eta_0,\phi_1)}{q^{m_2}}\\ \quad\qquad \mbox{if $b\neq0$, $c \neq 0$, $\tr_{q^m/q^{m_1}}(b^{-1})=0$ and $\eta(b)\eta_0(c)=1$,}  \\
                                q^{m-m_1}-q^{m-m_1-m_2}+\frac{G(\eta,\chi_1)G(\eta_0,\phi_1)}{q^{m_2}}\\ \quad\qquad \mbox{if $b\neq0$, $c \neq 0$, $\tr_{q^m/q^{m_1}}(b^{-1})=0$ and $\eta(b)\eta_0(c)=-1$,}  \\
                                q^{m-m_1}-q^{m-m_1-m_2}-\frac{(q^{m_2}-1)G(\eta,\chi_1)G(\eta',\lambda_1)\eta(-1)}{q^{m_1+m_2}}\\ \quad\qquad \mbox{if $b\neq0$, $c = 0$, $\tr_{q^m/q^{m_1}}(b^{-1})\neq0$ and $\eta(b)\eta'(\tr_{q^m/q^{m_1}}(b^{-1}))=1$,}  \\
                                q^{m-m_1}-q^{m-m_1-m_2}+\frac{(q^{m_2}-1)G(\eta,\chi_1)G(\eta',\lambda_1)\eta(-1)}{q^{m_1+m_2}}\\ \quad\qquad \mbox{if $b\neq0$, $c = 0$, $\tr_{q^m/q^{m_1}}(b^{-1})\neq0$ and $\eta(b)\eta'(\tr_{q^m/q^{m_1}}(b^{-1}))=-1$,}  \\
                                q^{m-m_1}-q^{m-m_1-m_2}+\frac{G(\eta,\chi_1)G(\eta',\lambda_1)\eta(-1)}{q^{m_1+m_2}}\\ \quad\qquad \mbox{if $b\neq0$, $c \neq 0$, $\tr_{q^m/q^{m_1}}(b^{-1})\neq0$ and $\eta(b)\eta'(\tr_{q^m/q^{m_1}}(b^{-1}))=1$,}  \\
                                q^{m-m_1}-q^{m-m_1-m_2}-\frac{G(\eta,\chi_1)G(\eta',\lambda_1)\eta(-1)}{q^{m_1+m_2}}\\ \quad\qquad \mbox{if $b\neq0$, $c \neq 0$, $\tr_{q^m/q^{m_1}}(b^{-1})\neq0$ and $\eta(b)\eta'(\tr_{q^m/q^{m_1}}(b^{-1}))=-1$,}
                               \end{array} \right.
  \end{eqnarray*}
  for odd $\frac{m}{m_2}$ and odd $\frac{m_1}{m_2}$,
  \begin{eqnarray*}
    \text{wt}(\bc_{(b,c)}) &=& \left\{\begin{array}{lll}
                                0 \qquad\qquad \qquad\mbox{if $b=c=0$,}  \\
                                q^{m-m_1} \qquad \qquad\mbox{if $b=0$ and $c \neq 0$,}  \\
                                q^{m-m_1}-q^{m-m_1-m_2}\quad \mbox{if $b\neq0$, $c = 0$ and $\tr_{q^m/q^{m_1}}(b^{-1})\neq0$,}  \\
                                q^{m-m_1}-q^{m-m_1-m_2}-\frac{(q^{m_2}-1)G(\eta,\chi_1)}{q^{m_2}}\\ \quad\qquad \mbox{if $b\neq0$, $c = 0$, $\tr_{q^m/q^{m_1}}(b^{-1})=0$ and $\eta(b)=1$,}  \\
                                q^{m-m_1}-q^{m-m_1-m_2}+\frac{(q^{m_2}-1)G(\eta,\chi_1)}{q^{m_2}}\\ \quad\qquad \mbox{if $b\neq0$, $c = 0$, $\tr_{q^m/q^{m_1}}(b^{-1})=0$ and $\eta(b)=-1$,}  \\
                                q^{m-m_1}-q^{m-m_1-m_2}+\frac{G(\eta,\chi_1)}{q^{m_2}}\\ \quad\qquad \mbox{if $b\neq0$, $c \neq 0$, $\tr_{q^m/q^{m_1}}(b^{-1})=0$ and $\eta(b)=1$,}  \\
                                q^{m-m_1}-q^{m-m_1-m_2}-\frac{G(\eta,\chi_1)}{q^{m_2}}\\ \quad\qquad \mbox{if $b\neq0$, $c \neq 0$, $\tr_{q^m/q^{m_1}}(b^{-1})=0$ and $\eta(b)=-1$,}  \\
                                q^{m-m_1}-q^{m-m_1-m_2}+\frac{G(\eta,\chi_1)G(\eta',\lambda_1)G(\eta_0,\phi_1)\eta(-1)}{q^{m_1+m_2}}\\ \quad\qquad \mbox{if $b\neq0$, $c \neq 0$, $\tr_{q^m/q^{m_1}}(b^{-1})\neq0$ and $\eta(b)\eta'(\tr_{q^m/q^{m_1}}(b^{-1}))\eta_0(c)=1$,}  \\
                                q^{m-m_1}-q^{m-m_1-m_2}-\frac{G(\eta,\chi_1)G(\eta',\lambda_1)G(\eta_0,\phi_1)\eta(-1)}{q^{m_1+m_2}}\\ \quad\qquad \mbox{if $b\neq0$, $c \neq 0$, $\tr_{q^m/q^{m_1}}(b^{-1})\neq0$ and $\eta(b)\eta'(\tr_{q^m/q^{m_1}}(b^{-1}))\eta_0(c)=-1$,}  \\
                               \end{array} \right.
  \end{eqnarray*}
  for even $\frac{m}{m_2}$ and odd $\frac{m_1}{m_2}$,
  \begin{eqnarray*}
    \text{wt}(\bc_{(b,c)}) &=& \left\{\begin{array}{lll}
                                0 \qquad\qquad \qquad\mbox{if $b=c=0$,}  \\
                                q^{m-m_1} \qquad \qquad\mbox{if $b=0$ and $c \neq 0$,}  \\
                                q^{m-m_1}-q^{m-m_1-m_2}-\frac{(q^{m_2}-1)G(\eta,\chi_1)}{q^{m_2}}\\ \quad\qquad \mbox{if $b\neq0$, $c = 0$, $\tr_{q^m/q^{m_1}}(b^{-1})=0$ and $\eta(b)=1$,}  \\
                                q^{m-m_1}-q^{m-m_1-m_2}+\frac{(q^{m_2}-1)G(\eta,\chi_1)}{q^{m_2}}\\ \quad\qquad \mbox{if $b\neq0$, $c = 0$, $\tr_{q^m/q^{m_1}}(b^{-1})=0$ and $\eta(b)=-1$,}  \\
                                q^{m-m_1}-q^{m-m_1-m_2}+\frac{G(\eta,\chi_1)}{q^{m_2}}\\ \quad\qquad \mbox{if $b\neq0$, $c \neq 0$, $\tr_{q^m/q^{m_1}}(b^{-1})=0$ and $\eta(b)=1$,}  \\
                                q^{m-m_1}-q^{m-m_1-m_2}-\frac{G(\eta,\chi_1)}{q^{m_2}}\\ \quad\qquad \mbox{if $b\neq0$, $c \neq 0$, $\tr_{q^m/q^{m_1}}(b^{-1})=0$ and $\eta(b)=-1$,}  \\
                                q^{m-m_1}-q^{m-m_1-m_2}-\frac{(q^{m_2}-1)G(\eta,\chi_1)G(\eta',\lambda_1)\eta(-1)}{q^{m_1+m_2}}\\ \quad\qquad \mbox{if $b\neq0$, $c = 0$, $\tr_{q^m/q^{m_1}}(b^{-1})\neq0$ and $\eta(b)\eta'(\tr_{q^m/q^{m_1}}(b^{-1}))=1$,}  \\
                                q^{m-m_1}-q^{m-m_1-m_2}+\frac{(q^{m_2}-1)G(\eta,\chi_1)G(\eta',\lambda_1)\eta(-1)}{q^{m_1+m_2}}\\ \quad\qquad \mbox{if $b\neq0$, $c = 0$, $\tr_{q^m/q^{m_1}}(b^{-1})\neq0$ and $\eta(b)\eta'(\tr_{q^m/q^{m_1}}(b^{-1}))=-1$,}  \\
                                q^{m-m_1}-q^{m-m_1-m_2}+\frac{G(\eta,\chi_1)G(\eta',\lambda_1)\eta(-1)}{q^{m_1+m_2}}\\ \quad\qquad \mbox{if $b\neq0$, $c \neq 0$, $\tr_{q^m/q^{m_1}}(b^{-1})\neq0$ and $\eta(b)\eta'(\tr_{q^m/q^{m_1}}(b^{-1}))=1$,}  \\
                                q^{m-m_1}-q^{m-m_1-m_2}-\frac{G(\eta,\chi_1)G(\eta',\lambda_1)\eta(-1)}{q^{m_1+m_2}}\\ \quad\qquad \mbox{if $b\neq0$, $c \neq 0$, $\tr_{q^m/q^{m_1}}(b^{-1})\neq0$ and $\eta(b)\eta'(\tr_{q^m/q^{m_1}}(b^{-1}))=-1$,}  \\
                               \end{array} \right.
  \end{eqnarray*}
  for even $\frac{m}{m_2}$ and even $\frac{m_1}{m_2}$.

By the above discussions, we derive that the dimension of $\overline{\cC_D}$ is $\frac{m}{m_2}+1$ as the zero codeword in $\overline{\cC_D}$ occurs only once. In the following, we will determine the frequency of each nonzero Hamming weight of $\overline{\cC_D}$ in three cases:

  {Case 1:} Let $\frac{m}{m_2}$ be odd and $\frac{m_1}{m_2}$ be odd.

  Denote by $w_1=q^{m-m_1}$, $w_2=q^{m-m_1}-q^{m-m_1-m_2}$, $w_3=q^{m-m_1}-q^{m-m_1-m_2}-\frac{G(\eta,\chi_1)G(\eta_0,\phi_1)}{q^{m_2}}$, $w_4=q^{m-m_1}-q^{m-m_1-m_2}+\frac{G(\eta,\chi_1)G(\eta_0,\phi_1)}{q^{m_2}}$, $w_5=q^{m-m_1}-q^{m-m_1-m_2}-\frac{(q^{m_2}-1)G(\eta,\chi_1)G(\eta',\lambda_1)\eta(-1)}{q^{m_1+m_2}}$, $w_6=q^{m-m_1}-q^{m-m_1-m_2}+\frac{(q^{m_2}-1)G(\eta,\chi_1)G(\eta',\lambda_1)\eta(-1)}{q^{m_1+m_2}}$, $w_7=q^{m-m_1}-q^{m-m_1-m_2}+\frac{G(\eta,\chi_1)G(\eta',\lambda_1)\eta(-1)}{q^{m_1+m_2}}$ and $w_8=q^{m-m_1}-q^{m-m_1-m_2}-\frac{G(\eta,\chi_1)G(\eta',\lambda_1)\eta(-1)}{q^{m_1+m_2}}$.
Now we determine the frequency $A_{w_i}$, $1 \leq i \leq 8$. It is obvious that $A_{w_1}=q^{m_2}-1$ and $A_{w_2}=q^{m-m_1}-1$. Let $\gf_{q^m}^*=\langle\alpha\rangle$ and $\gf_{q^{m_1}}^*=\langle\beta\rangle$. Note that
\begin{eqnarray*}
  A_{w_3} &=& \left|\{(b,c) \in \gf_{q^m}\times\gf_{q^{m_2}}: \tr_{q^m/q^{m_1}}(b^{-1})=0, \eta(b)\eta_0(c)=1\}\right|\\
  &=& \frac{q^{m_2}-1}{2}\left|\{b \in \langle\alpha^2\rangle: \tr_{q^m/q^{m_1}}(b^{-1})=0\}\right| + \frac{q^{m_2}-1}{2}\left|\{b \in \alpha\langle\alpha^2\rangle: \tr_{q^m/q^{m_1}}(b^{-1})=0\}\right|\\
  &=& \frac{q^{m_2}-1}{2}\left|\{b \in \gf_{q^m}^*: \tr_{q^m/q^{m_1}}(b^{-1})=0\}\right|\\
  &=& \frac{(q^{m_2}-1)(q^{m-m_1}-1)}{2}.
\end{eqnarray*}
Then $A_{w_4}=(q^{m_2}-1)(q^{m-m_1}-1)-A_{w_3}=\frac{(q^{m_2}-1)(q^{m-m_1}-1)}{2}$. By Lemma \ref{lem-N}, we have
\begin{eqnarray*}
  A_{w_5} &=& \left| \{b \in \gf_{q^m}^*:\eta(b)\eta'(\tr_{q^m/q^{m_1}}(b^{-1}))=1\} \right|\\
  &=& \left| \{b \in \langle\alpha^2\rangle : \tr_{q^m/q^{m_1}}(b^{-1}) \in \langle\beta^2\rangle\} \right| + \left| \{b \in \alpha\langle\alpha^2\rangle: \tr_{q^m/q^{m_1}}(b^{-1}) \in \beta\langle\beta^2\rangle\} \right|\\
  &=& \frac{q^m+G(\eta,\chi_1)G(\eta',\lambda_1)\eta'(-1)}{2q^{m_1}} \cdot \frac{q^{m_1}-1}{2}+\\&&\left(q^{m-m_1}-\frac{q^m-G(\eta,\chi_1)G(\eta',\lambda_1)\eta'(-1)}{2q^{m_1}}\right) \cdot \frac{q^{m_1}-1}{2}\\
  &=& \frac{(q^m+G(\eta,\chi_1)G(\eta',\phi_1)\eta(-1))(q^{m_1}-1)}{2q^{m_1}}
\end{eqnarray*}
and
\begin{eqnarray*}
  A_{w_6} &=& \left| \{b \in \gf_{q^m}^*:\eta(b)\eta'(\tr_{q^m/q^{m_1}}(b^{-1}))=-1\} \right|\\
  &=& \left| \{b \in \langle\alpha^2\rangle : \tr_{q^m/q^{m_1}}(b^{-1}) \in \beta\langle\beta^2\rangle\} \right| + \left| \{b \in \alpha\langle\alpha^2\rangle: \tr_{q^m/q^{m_1}}(b^{-1}) \in \langle\beta^2\rangle\} \right|\\
  &=& \frac{q^m-G(\eta,\chi_1)G(\eta',\lambda_1)\eta'(-1)}{2q^{m_1}} \cdot \frac{q^{m_1}-1}{2}+\\&&\left(q^{m-m_1}-\frac{q^m+G(\eta,\chi_1)G(\eta',\lambda_1)\eta'(-1)}{2q^{m_1}}\right) \cdot \frac{q^{m_1}-1}{2}\\
  &=& \frac{(q^m-G(\eta,\chi_1)G(\eta',\phi_1)\eta(-1))(q^{m_1}-1)}{2q^{m_1}}.
\end{eqnarray*}

Obviously, $$A_{w_7}=(q^{m_2}-1)A_{w_5}=\frac{(q^{m_2}-1)(q^m+G(\eta,\chi_1)G(\eta',\phi_1)\eta(-1))(q^{m_1}-1)}{2q^{m_1}}$$ and $$A_{w_8}=(q^{m_2}-1)A_{w_6}=\frac{(q^{m_2}-1)(q^m-G(\eta,\chi_1)G(\eta',\phi_1)\eta(-1))(q^{m_1}-1)}{2q^{m_1}}.$$

In this case, by Lemma \ref{Guassum}, we have
\begin{eqnarray*}
  G(\eta,\chi_1)G(\eta_0,\phi_1)=(-1)^{l_1}q^{\frac{m+m_2}{2}} \mbox{ and } G(\eta,\chi_1)G(\eta',\lambda_1)\eta'(-1)=(-1)^{l_2}q^{\frac{m+m_1}{2}},
\end{eqnarray*}
where $l_1=\frac{(p-1)(m+m_2)e}{4}$ and $l_2=\frac{(p-1)(m+3m_1)e}{4}$.
Then the weight distribution of $\overline{\cC_D}$ for odd $\frac{m}{m_2}$ and odd $\frac{m_1}{m_2}$ is derived and is listed in Table \ref{tab5.1}.

{Case 2:} Let $\frac{m}{m_2}$ be even and $\frac{m_1}{m_2}$ be odd.

Denote by $w_1=q^{m-m_1}$, $w_2=q^{m-m_1}-q^{m-m_1-m_2}$, $w_3=q^{m-m_1}-q^{m-m_1-m_2}-\frac{(q^{m_2}-1)G(\eta,\chi_1)}{q^{m_2}}$, $w_4=q^{m-m_1}-q^{m-m_1-m_2}+\frac{(q^{m_2}-1)G(\eta,\chi_1)}{q^{m_2}}$, $w_5=q^{m-m_1}-q^{m-m_1-m_2}+\frac{G(\eta,\chi_1)}{q^{m_2}}$, $w_6=q^{m-m_1}-q^{m-m_1-m_2}-\frac{G(\eta,\chi_1)}{q^{m_2}}$, $w_7=q^{m-m_1}-q^{m-m_1-m_2}+\frac{G(\eta,\chi_1)G(\eta',\lambda_1)G(\eta_0,\phi_1)\eta(-1)}{q^{m_1+m_2}}$ and $w_8=q^{m-m_1}-q^{m-m_1-m_2}-\\ \frac{G(\eta,\chi_1)G(\eta',\lambda_1)G(\eta_0,\phi_1)\eta(-1)}{q^{m_1+m_2}}$. Now we determine the frequency $A_{w_i}$, $1 \leq i \leq 8$. It is obvious that $A_{w_1}=q^{m_2}-1$ and $A_{w_2}=q^{m}-q^{m-m_1}$. By Lemma \ref{lem-N}, we have
\begin{eqnarray*}
  A_{w_3} = \left|\{b \in \langle\alpha^2\rangle: \tr_{q^m/q^{m_1}}(b^{-1})=0\}\right|=\frac{q^m-q^{m_1}+(q^{m_1}-1)G(\eta,\chi_1)}{2q^{m_1}}
\end{eqnarray*}
and $A_{w_4}=q^{m-m_1}-1-A_{w_3}=\frac{q^m-q^{m_1}-(q^{m_1}-1)G(\eta,\chi_1)}{2q^{m_1}}$. Furthermore, we deduce that $A_{w_5}=(q^{m_2}-1)A_{w_3}=\frac{(q^{m_2}-1)(q^m-q^{m_1}+(q^{m_1}-1)G(\eta,\chi_1))}{2q^{m_1}}$ and $A_{w_6}=(q^{m_2}-1)A_{w_4}=\frac{(q^{m_2}-1)(q^m-q^{m_1}-(q^{m_1}-1)G(\eta,\chi_1))}{2q^{m_1}}$.
Note that
\begin{eqnarray*}
  A_{w_7} &=& \left|\{(b,c) \in \gf_{q^m}^* \times \gf_{q^{m_2}}^*: \eta(b)\eta'(\tr_{q^m/q^{m_1}}(b^{-1}))\eta_0(c)=1\}\right|\\
  &=& \frac{q^{m_2}-1}{2}\left|\{b \in \langle\alpha^2\rangle: \tr_{q^m/q^{m_1}}(b^{-1}) \in \langle\beta^2\rangle\}\right|\\
  & &+\frac{q^{m_2}-1}{2}\left|\{b \in \alpha\langle\alpha^2\rangle: \tr_{q^m/q^{m_1}}(b^{-1}) \in \beta\langle\beta^2\rangle\}\right|\\
  &&+\frac{q^{m_2}-1}{2}\left|\{b \in \langle\alpha^2\rangle: \tr_{q^m/q^{m_1}}(b^{-1}) \in \beta\langle\beta^2\rangle\}\right|\\
  &&+\frac{q^{m_2}-1}{2}\left|\{b \in \alpha\langle\alpha^2\rangle: \tr_{q^m/q^{m_1}}(b^{-1}) \in \langle\beta^2\rangle\}\right|\\
  &=&\frac{q^{m_2}-1}{2}\left|\{b \in \gf_{q^m}^*: \tr_{q^m/q^{m_1}}(b^{-1}) \neq 0\}\right|\\
  &=& \frac{q^{m-m_1}(q^{m_2}-1)(q^{m_1}-1)}{2}.
\end{eqnarray*}
By the same way, we also derive that $A_{w_8}=\frac{q^{m-m_1}(q^{m_2}-1)(q^{m_1}-1)}{2}$.

In this case, by Lemma \ref{Guassum}, we have
\begin{eqnarray*}
  G(\eta,\chi_1)=(-1)^{l_3}q^{\frac{m}{2}} \mbox{ and } G(\eta,\chi_1)G(\eta',\lambda_1)G(\eta_0,\phi_1)\eta'(-1)=(-1)^{l_4}q^{\frac{m+m_1+m_2}{2}},
\end{eqnarray*}
where $l_3=\frac{(p-1)em}{4}+1$ and $l_4=\frac{(p-1)(m+3m_1+m_2)e}{4}+1$.
Then the weight distribution of $\overline{\cC_D}$ for even $\frac{m}{m_2}$ and odd $\frac{m_1}{m_2}$ is derived and is listed in Table \ref{tab5.2}.

{Case 3:} Let $\frac{m}{m_2}$ be even and $\frac{m_1}{m_2}$ be even.

In this case, we should respectively consider two subcases for even $\frac{m}{m_1}$ and for odd $\frac{m}{m_1}$ to calculate the frequency of each nonzero Hamming weight of $\overline{\cC_D}$ by Lemma \ref{lem-N}.  By Lemma \ref{Guassum}, we have
\begin{eqnarray*}
  G(\eta,\chi_1)=(-1)^{l_3}q^{\frac{m}{2}} \mbox{ and } G(\eta,\chi_1)G(\eta',\lambda_1)\eta'(-1)=(-1)^{l_5}q^{\frac{m+m_1}{2}},
\end{eqnarray*}
where $l_3=\frac{(p-1)em}{4}+1$ and $l_5=\frac{(p-1)(m+m_1)e}{4}$.
Similarly to the discussions in Case 1 and Case 2 above, the weight distribution of $\overline{\cC_D}$ for even $\frac{m}{m_2}$ and even $\frac{m_1}{m_2}$ can be derived. It is listed in Tables \ref{tab5.31} and \ref{tab5.32} for two subcases.

By Tables \ref{tab5.1}, \ref{tab5.2}, \ref{tab5.31} and \ref{tab5.32}, we deduce that the augmented code $\overline{\cC_D}$ is $q$-divisible. Then by Lemma \ref{th-selforthogonal}, $\overline{\cC_D}$ is self-orthogonal.
\end{proof}

The locality of $\overline{\cC_D}$ is determined in the following theorem.
\begin{theorem}
 Let $q$ be an odd prime power. Let $m, m_1, m_2$ be three positive integers with $m_2 \mid m_1 \mid m$ and $m \geq 3m_1$. Then  $\overline{\cC_D}$ is a locally recoverable code with locality $2$.
\end{theorem}

\begin{proof}
  Let $\gf_{q^m}^*=\langle\alpha\rangle$. Then $\{\alpha^0, \alpha^1, \cdots, \alpha^{\frac{m}{m_2}-1}\}$ is a $\gf_{q^{m_2}}$-basis of $\gf_{q^m}$. Let $d_1, d_2, \cdots, d_{n-1},d_n$ be all the elements in $D$. For convenience, let  $d_n = 0$ due to $0\in \cC_D$.
By definition, the generator matrix $G$ of $\overline{\cC_D}$ is given by
\begin{eqnarray*}
G:=\left[
\begin{array}{cccc}
1&1&\cdots&1 \\
 \tr_{q^m/q^{m_2}}(\alpha^{0}d_{1}^2)& \tr_{q^m/q^{m_2}}(\alpha^{0}d_{2}^2)& \cdots &\tr_{q^m/q^{m_2}}(\alpha^{0}d_{n}^2) \\
\tr_{q^m/q^{m_2}}(\alpha^{1}d_{1}^2)& \tr_{q^m/q^{m_2}}(\alpha^{1}d_{2}^2)& \cdots &\tr_{q^m/q^{m_2}}(\alpha^{1}d_{n}^2) \\
\vdots &\vdots &\ddots &\vdots \\
\tr_{q^m/q^{m_2}}(\alpha^{\frac{m}{m_2}-1}d_{1}^2)& \tr_{q^m/q^{m_2}}(\alpha^{\frac{m}{m_2}-1}d_{2}^2)& \cdots &\tr_{q^m/q^{m_2}}(\alpha^{\frac{m}{m_2}-1}d_{n}^2)
\end{array}\right].
\end{eqnarray*}
For convenience, let $$\bg_i:=(1, \tr_{q^m/q^{m_2}}(\alpha^0d_i^2), \tr_{q^m/q^{m_2}}(\alpha^1 d_i^2), \cdots, \tr_{q^m/q^{m_2}}(\alpha^{\frac{m}{m_2}-1}d_i^2))^T,$$ where $1 \leq i \leq n$.
 Note that $ud_i\in D$ for any $u\in \gf_q^*$ if $d_i\in D$. For fixed $\bg_i$ $(1 \leq i \leq n-1)$, we select $d_j := -d_i\in D$. Then $1 \leq j\neq i \leq n-1$
 and $\bg_i=\bg_j$.
Let $\gf_q^*=\langle\beta\rangle$. For $\bg_n=(1,0,\cdots,0)$ and any $d_i$ with $1 \leq i \leq n-1$, there exist $u \in \gf_q\ \{0,1\}$ with $\frac{u}{u-1} \in \langle\beta^2\rangle$ and  $v:=1-u$ such that $d_j := (\frac{u}{u-1})^{\frac{1}{2}}d_i\in D$ and
$\bg_n=u\bg_i+v\bg_j$, i.e.
 \begin{eqnarray*}
\left\{\begin{array}{ccc}
1&=&u+v,\\
 0&=&u\tr_{q^m/q^{m_2}}(\alpha^{0}d_i^2)+v\tr_{q^m/q^{m_2}}(\alpha^{0}d_j^2), \\
  0&=&u\tr_{q^m/q^{m_2}}(\alpha^{1}d_i^2)+v\tr_{q^m/q^{m_2}}(\alpha^{1}d_j^2), \\
 & \vdots &\\
  0&=&u\tr_{q^m/q^{m_2}}(\alpha^{\frac{m}{m_2}-1}d_i^2)+v\tr_{q^m/q^{m_2}}(\alpha^{\frac{m}{m_2}-1}d_j^2).
\end{array}\right.
\end{eqnarray*}
Then $\overline{\cC_D}$ is a locally recoverable code with locality $2$ according to Definition \ref{Def-locality}.
\end{proof}

\subsection{The second family of self-orthogonal codes with locality $2$}\label{sec5.2}

In this subsection, let $p$ be an odd prime and $m$ be an positive integer with $m \geq 4$. Let $f(x) \in \mathcal{RF}$ be a $p$-ary weakly regular bent function from $\gf_{p^m}$ to $\gf_p$. Denote by the defining set $D_f=\{x \in \gf_{p^m}: f(x)=0\} $.
In \cite{T}, the authors defined a family of $p$-ary linear codes $\cC_{D_f}$ by
$$\cC_{D_f}=\left\{(\tr_{p^m/p}(bx))_{x \in D_f}:b \in \gf_{p^m}\right\}.$$
Then its augmented code is given by
$$\overline{\cC_{D_f}}=\left\{(\tr_{p^m/p}(bx))_{x \in D_f}+c\mathbf{1}:b \in \gf_{p^m}, c \in \gf_p\right\}, $$
where $\mathbf{1}$ is the all-$1$ vector of length $\mid D_f \mid$. In this subsection, we will first prove that the augmented code $\overline{\cC_{D_f}}$ is self-orthogonal and then determine the locality of $\overline{\cC_{D_f}}$.  The parameters and weight distribution of $\overline{\cC_{D_f}}$ were determined in \cite{Heng2023}.

\begin{lemma}\cite{Heng2023}\label{lem-5.2even}
  Let $m \geq 4$ be even and $f(x) \in \mathcal{RF}$ with $\varepsilon$ the sign of the Walsh transform of $f(x)$. Let $p^*=(-1)^{\frac{p-1}{2}}p$. Then $\overline{\cC_{D_f}}$ has parameters $[\frac{1}{p}(p^m+(p-1)\varepsilon\sqrt{p^*}^m), m+1]$ and its weight distribution is listed in Table \ref{tab5.3}. Besides, $\overline{\cC_{D_f}}^{\perp}$ is a $[\frac{1}{p}(p^m+(p-1)\varepsilon\sqrt{p^*}^m), \frac{1}{p}(p^m+(p-1)\varepsilon\sqrt{p^*}^m)- m-1, 3]$ linear code.
  \begin{table}[h]
\begin{center}
\caption{The weight distribution of $\overline{\cC_{D_f}}$ in Lemma \ref{lem-5.2even} ($m$ is even).}\label{tab5.3}
\begin{tabular}{@{}ll@{}}
\toprule%
Weight & Frequency  \\
\midrule
$0$ & $1$\\
$(p-1)p^{m-2}$ & $\frac{1}{p}(p^m+(p-1)\varepsilon\sqrt{p^*}^m)-1$ \\
$\frac{1}{p}((p-1)p^{m-1}+(p-2)\varepsilon\sqrt{p^*}^m)$ & $\frac{(p-1)^2}{p}(p^m-\varepsilon\sqrt{p^*}^m)$\\
$\frac{1}{p}(p-1)(p^{m-1}+\varepsilon\sqrt{p^*}^m)$ & $\frac{p-1}{p}(2p^m+(p-2)\varepsilon\sqrt{p^*}^m-p)$\\
$\frac{1}{p}(p^m+(p-1)\varepsilon\sqrt{p^*}^m)$ & $p-1$\\
\bottomrule
\end{tabular}
\end{center}
\end{table}
\end{lemma}

\begin{lemma}\cite{Heng2023}\label{lem-5.2odd}
  Let $m > 3$ be odd and $f(x) \in \mathcal{RF}$ with $\varepsilon$ the sign of the Walsh transform of $f(x)$. Let $p^*=(-1)^{\frac{p-1}{2}}p$ and $\eta_0$ be the  quadratic multiplicative characters of $\gf_p$. Then $\overline{\cC_{D_f}}$ has parameters $[p^{m-1}, m+1]$ and its weight distribution is listed in Table \ref{tab5.4}. Besides, $\overline{\cC_{D_f}}^{\perp}$ is a $[p^{m-1}, p^{m-1}-m-1, 3]$ linear code.
  \begin{table}[h]
\begin{center}
\caption{The weight distribution of $\overline{\cC_{D_f}}$ in Lemma \ref{lem-5.2odd} ($m$ is odd).}\label{tab5.4}
\begin{tabular}{@{}ll@{}}
\toprule%
Weight & Frequency  \\
\midrule
$0$ & $1$\\
$\frac{p-1}{p^2}(p^m-\varepsilon\sqrt{p^*}^{m+1})$ & $\frac{p-1}{2}(p^{m-1}+\eta_0(-1)\varepsilon\sqrt{p^*}^{m-1})$\\
$\frac{1}{p^2}((p-1)p^m-\varepsilon\sqrt{p^*}^{m+1})$ & $\frac{(p-1)^2}{2}(p^{m-1}-\eta_0(-1)\varepsilon\sqrt{p^*}^{m-1})$\\
$(p-1)p^{m-2}$ & $p(p^{m-1}-1)$ \\
$\frac{1}{p^2}((p-1)p^m+\varepsilon\sqrt{p^*}^{m+1})$ & $\frac{(p-1)^2}{2}(p^{m-1}+\eta_0(-1)\varepsilon\sqrt{p^*}^{m-1})$\\
$\frac{p-1}{p^2}(p^m+\varepsilon\sqrt{p^*}^{m+1})$ & $\frac{p-1}{2}(p^{m-1}-\eta_0(-1)\varepsilon\sqrt{p^*}^{m-1})$\\
$p^{m-1}$ & $p-1$\\
\bottomrule
\end{tabular}
\end{center}
\end{table}
\end{lemma}

\begin{theorem}
 Let $p$ ba an odd prime, $m > 3$ be an integer and $f(x) \in \mathcal{RF}$ with $\varepsilon$ the sign of the Walsh transform of $f(x)$. Then $\overline{\cC_{D_f}}$ is a self-orthogonal $p$-divisible code. Besides, $\overline{\cC_{D_f}}$ is a locally recoverable code with locality $2$.
\end{theorem}

\begin{proof}
By Lemmas \ref{lem-5.2even} and \ref{lem-5.2odd}, we deduce that $\overline{\cC_{D_f}}$ is $p$-divisible and $\mathbf{1} \in \overline{\cC_{D_f}}$. Then by Theorem \ref{th-selforthogonal}, $\overline{\cC_{D_f}}$ is self-orthogonal.

  Let $\gf_{p^m}^*=\langle\beta\rangle$. Then $\{\beta^0, \beta^1, \cdots, \beta^{m-1}\}$ is a $\gf_p$-basis of $\gf_{p^m}$. Let $d_1, d_2, \cdots, d_{n-1},d_n$ be all the elements in $D_f$. For convenience, let  $d_n = 0$ due to $0\in \cC_{D_f}$.
By definition, the generator matrix $G$ of $\overline{\cC_{D_f}}$ is given by
\begin{eqnarray*}
G:=\left[
\begin{array}{cccc}
1&1&\cdots&1 \\
 \tr_{p^m/p}(\beta^{0}d_{1})& \tr_{p^m/p}(\beta^{0}d_{2})& \cdots &\tr_{p^m/p}(\beta^{0}d_{n}) \\
\tr_{p^m/p}(\beta^{1}d_{1})& \tr_{p^m/p}(\beta^{1}d_{2})& \cdots &\tr_{p^m/p}(\beta^{1}d_{n}) \\
\vdots &\vdots &\ddots &\vdots \\
\tr_{p^m/p}(\beta^{m-1}d_{1})& \tr_{p^m/p}(\beta^{m-1}d_{2})& \cdots &\tr_{p^m/p}(\beta^{m-1}d_{n})
\end{array}\right].
\end{eqnarray*}
For convenience, let $\bg_i:=(1, \tr_{p^m/p}(\beta^0d_i), \tr_{p^m/p}(\beta^1 d_i), \cdots, \tr_{p^m/p}(\beta^{m-1}d_i))^T$, where $1 \leq i \leq n$.
By Definition \ref{def-weakly}, we derive that $ad_i\in D_f$ for any $a\in \gf_p^*$ if $d_i\in D_f$. For fixed $\bg_i$ $(1 \leq i \leq n-1)$ and any $u\in \gf_p\backslash\{0,1\}$, then there exists $v=1-u$ such that $d_j := u^{-1}d_i\in D_f$ and
 \begin{eqnarray*}
\left\{\begin{array}{ccc}
1&=&u+v,\\
  \tr_{p^m/p}(\beta^{0}d_{i})&=&u\tr_{p^m/p}(\beta^{0}d_j)+v\tr_{p^m/p}(\beta^{0}d_n), \\
  \tr_{p^m/p}(\beta^{1}d_{i})&=&u\tr_{p^m/p}(\beta^{1}d_j)+v\tr_{p^m/p}(\beta^{1}d_n), \\
 & \vdots &\\
  \tr_{p^m/p}(\beta^{m-1}d_{i})&=&u\tr_{p^m/p}(\beta^{m-1}d_j)+v\tr_{p^m/p}(\beta^{m-1}d_n),
\end{array}\right.
\end{eqnarray*}
where $d_n=0$. This implies that $\bg_i=u\bg_j+v\bg_n$, where $1 \leq i \leq n-1$.
Since $v\neq 0,1$, it is clear that $\bg_n$ is also a linear combination of $\bg_i$ and $\bg_j$.
Then $\overline{\cC_{D_f}}$ is a locally recoverable code with locality $2$ according to Definition \ref{Def-locality}.
\end{proof}

\begin{remark}
 We remark that $\overline{\cC_{D_f}}$ was proved to be self-orthogonal for $p=3$ by Heng et al. in \cite{Heng2023}.
 However, it was open to prove the self-orthogonality of $\overline{\cC_{D_f}}$ for general odd prime $p$.
  In this subsection, we have solved this open problem.
\end{remark}

\section{Concluding remarks}\label{sec6}
In this paper, we found an interesting relationship between self-orthogonal codes and $p$-divisible codes.
This provides us a useful way to construct $q$-ary self-orthogonal codes from $p$-divisible codes.
By this result, we studied the orthogonality of projective two-weight codes and the Griesmer codes.
Besides, we constructed six families of self-orthogonal $p$-divisible codes from known cyclic codes.
Furthermore, we constructed two more families of self-orthogonal $p$-divisible codes with locality 2.
Some families of optimal or almost optimal linear codes were also derived.

We remark that the self-orthogonal codes obtained by us can be used to construct quantum codes and even lattices \cite{C, LLX, Steane1, Steane2, W}.
The self-orthogonal codes with small locality have nice application in distributed storage systems \cite{GH}.
For further research, it is interesting to classify $p$-divisible codes as they can be used to construct self-orthogonal codes.

\end{document}